\documentclass[10pt,conference,letterpaper]{IEEEtran}
\usepackage{times}
\usepackage{amsmath}
\usepackage{amsfonts}
\usepackage{epsfig}
\usepackage{graphicx}
\usepackage{balance} 
\usepackage{cite}
\usepackage{url}
\usepackage{tikz}
\usepackage[caption=false,font=footnotesize]{subfig}
\usepackage{algorithm}
\usepackage{algpseudocode}
\usepackage{booktabs}
\usepackage{multirow}
\usepackage{makecell}
\usepackage{courier}
\usepackage{tabularx}
\usepackage{amssymb}
\usepackage{myenumitem}

\usetikzlibrary{positioning,shapes,arrows}

\newtheorem{lemma}{Lemma}
\newtheorem{theorem}{Theorem}
\newtheorem{definition}{Definition}
\newtheorem{example}{Example}

\makeatletter
\newcommand*{\centerfloat}{%
    \parindent \z@
    \leftskip \z@ \@plus 1fil \@minus \textwidth
    \rightskip\leftskip
    \parfillskip \z@skip}
\makeatother

\title{An Efficient Probabilistic Approach for Graph Similarity Search [Extended Technical Report]}

\author{
    {Zijian Li{\small $~^{\#1}$}, Xun Jian{\small $~^{\#2}$}, Xiang Lian{\small $~^{ *3}$}, Lei Chen{\small $~^{\#4}$}}%
    \vspace{1.6mm}\\
    \fontsize{10}{10}\selectfont\itshape
    $^{\#}$\,Computer Science and Engineering Department, HKUST, Hong Kong, China\\
    \fontsize{9}{9}\selectfont\ttfamily\upshape
    $^{1}$\,zlicb@cse.ust.hk \hspace{1em}
    $^{2}$\,xjian@cse.ust.hk \hspace{1em}
    $^{4}$\,leichen@cse.ust.hk%
    \vspace{1.2mm}\\
    \fontsize{10}{10}\selectfont\rmfamily\itshape
    $^{*}$\, Computer Science Department, Kent State University, Kent, USA\\
    \fontsize{9}{9}\selectfont\ttfamily\upshape
    $^{3}$\,xlian@kent.edu
}

\begin{document}
    \maketitle
    
\begin{abstract}
    Graph similarity search is a common and fundamental operation in
    graph databases. One of the most popular graph similarity
    measures is the {Graph Edit Distance} (GED) mainly because
    of its broad applicability and high interpretability. Despite its prevalence, exact GED computation is proved to be NP-hard, which
    could result in unsatisfactory computational efficiency on large graphs. 
    However, exactly accurate search results are usually 
    unnecessary for real-world applications 
    especially when the responsiveness is far more important than the accuracy.
    Thus, in this paper, we propose a novel probabilistic approach to efficiently estimate 
    GED, which is further leveraged for the graph similarity search. 
    Specifically, we first take {branches} as elementary structures in
    graphs, and introduce a novel graph similarity measure by comparing branches 
    between graphs, i.e., {Graph Branch Distance} (GBD),
    which can be efficiently calculated in polynomial time. 
    Then, we formulate the relationship between GED and GBD 
    by considering branch variations as the result ascribed to 
    graph edit operations, and model this process by probabilistic approaches. 
    By applying our model, the GED between any two graphs 
    can be efficiently estimated by their GBD, 
    and these estimations are finally utilized in the graph similarity search. 
    Extensive experiments show that our approach has
    better accuracy, efficiency and scalability than other comparable methods
    in the graph similarity search over real and synthetic data sets.
\end{abstract}

\vspace{-0.4em}
\section{Introduction} \label{sec-intro}
\vspace{-0.4em}

Graph similarity search is a common and fundamental operation in 
graph databases, which has widespread applications in various fields 
including bio-informatics, sociology, and chemical analysis, over
the past few decades. For evaluating the similarity between graphs, \emph{Graph Edit Distance} (GED) \cite{sanfeliu1983distance} is one of the most prevalent
measures because of its wide applicability, that is, GED is capable
of dealing with various kinds of graphs including directed and
undirected graphs, labeled and unlabeled graphs, as well as simple
graphs and multi-graphs (which could have multiple edges between two
vertices). Furthermore, GED has high interpretability, since it
corresponds to some sequences of concrete graph edit operations (including insertion of vertices and edges, etc.) of
minimal lengths, rather than implicit graph embedding utilized in
spectral \cite{caelli2004eigenspace} or kernel
\cite{gartner2003graph} measures. Example
\ref{intro-ged-example} illustrates the basic idea of GED.

\begin{example}\label{intro-ged-example}
    Assume that we have two graphs $G_1$ and $G_2$ as shown in Figure \ref{fig-intro-ged-example}. The label sets of graph $G_1$'s vertices and edges are $\{A,B,C\}$ and $\{y,y,z\}$, respectively, and the label sets of graph $G_2$'s vertices and edges are $\{A,B,C\}$ and $\{x,y,z\}$, respectively. The Graph Edit Distance (GED) between $G_1$ and $G_2$ is the minimal number of graph edit operations to transform $G_1$ into $G_2$. It can be proved that the GED between $G_1$ and $G_2$ is 3, which can be achieved by \raisebox{.5pt}{\textcircled{\raisebox{-.9pt} {1}}} deleting the edge between $v_1$ and $v_3$ in $G_1$, and \raisebox{.5pt}{\textcircled{\raisebox{-.9pt} {2}}} inserting an isolated vertex $v_4$ with label $A$, and \raisebox{.5pt}{\textcircled{\raisebox{-.9pt} {3}}} inserting an edge between $v_3$ and $v_4$ with label $x$.
\end{example}

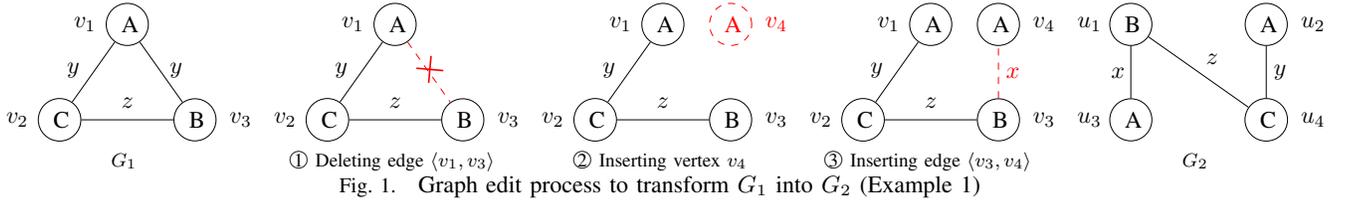
\begin{figure*}[!htbp]
    \vspace{-25pt}
    \centerfloat
    \captionsetup[subfigure]{labelformat=empty}
    \scalebox{.9}{
        \subfloat[$G_1$]{
            \begin{tikzpicture}[
            vertex/.style={draw,circle,text width=5pt,align=center},
            tag/.style={text width=5pt,align=center}
            ]
            \node[vertex] (v1) at (1,1.41) {A};
            \node[tag] (t1) at (0.3,1.41) {$v_1$};
            \node[vertex] (v2) at (0,0) {C};
            \node[tag] (t2) at (-0.7,0) {$v_2$};
            \node[vertex] (v3) at (2,0) {B};
            \node[tag] (t3) at (2.6,0) {$v_3$};
            \path (v1) edge[-] (v2);
            \path (v3) edge[-] (v2);
            \path (v1) edge[-] (v3);
            \node[tag] (te1) at (0.2,0.74) {$y$};
            \node[tag] (te2) at (1.72,0.74) {$y$};
            \node[tag] (te3) at (1,0.25) {$z$};
            \end{tikzpicture}
        }
        \hfill
        \subfloat[\raisebox{.5pt}{\textcircled{\raisebox{-.9pt} {1}}} Deleting edge $\langle v_1, v_3 \rangle$ ]{
            \begin{tikzpicture}[
            vertex/.style={draw,circle,text width=5pt,align=center},
            tag/.style={text width=5pt,align=center}
            ]
            \node[vertex] (v1) at (1,1.41) {A};
            \node[tag] (t1) at (0.3,1.41) {$v_1$};
            \node[vertex] (v2) at (0,0) {C};
            \node[tag] (t2) at (-0.7,0) {$v_2$};
            \node[vertex] (v3) at (2,0) {B};
            \node[tag] (t3) at (2.6,0) {$v_3$};
            \path (v1) edge[-] (v2);
            \path (v3) edge[-] (v2);
            \path (v1) edge[dashed,red] (v3);
            \node[tag] (te1) at (0.2,0.74) {$y$};
            \node[tag] (te3) at (1,0.25) {$z$};
            \node[tag,red,rotate=40] (te2) at (1.4,0.65) {\LARGE $\times$};
            \end{tikzpicture}
        }
        \hfill
        \subfloat[\raisebox{.5pt}{\textcircled{\raisebox{-.9pt} {2}}} Inserting vertex $v_4$]{
            \begin{tikzpicture}[
            vertex/.style={draw,circle,text width=5pt,align=center},
            tag/.style={text width=5pt,align=center}
            ]
            \node[vertex] (v1) at (1,1.41) {A};
            \node[tag] (t1) at (0.3,1.41) {$v_1$};
            \node[vertex] (v2) at (0,0) {C};
            \node[tag] (t2) at (-0.7,0) {$v_2$};
            \node[vertex] (v3) at (2,0) {B};
            \node[tag] (t3) at (2.6,0) {$v_3$};
            \node[vertex,red,dashed] (v4) at (2,1.41) {A};
            \node[tag,red] (t3) at (2.6,1.41) {$v_4$};
            \path (v1) edge[-] (v2);
            \path (v3) edge[-] (v2);
            \node[tag] (te1) at (0.2,0.74) {$y$};
            \node[tag] (te3) at (1,0.25) {$z$};
            \end{tikzpicture}
        }
        \hfill
        \subfloat[\raisebox{.5pt}{\textcircled{\raisebox{-.9pt} {3}}} Inserting edge $\langle v_3, v_4 \rangle$]{
            \begin{tikzpicture}[
            vertex/.style={draw,circle,text width=5pt,align=center},
            tag/.style={text width=5pt,align=center}
            ]
            \node[vertex] (v1) at (1,1.41) {A};
            \node[tag] (t1) at (0.3,1.41) {$v_1$};
            \node[vertex] (v2) at (0,0) {C};
            \node[tag] (t2) at (-0.7,0) {$v_2$};
            \node[vertex] (v3) at (2,0) {B};
            \node[tag] (t3) at (2.6,0) {$v_3$};
            \node[vertex] (v4) at (2,1.41) {A};
            \node[tag] (t3) at (2.6,1.41) {$v_4$};
            \path (v1) edge[-] (v2);
            \path (v3) edge[-] (v2);
            \path (v3) edge[-,red,dashed] (v4);
            \node[tag] (te1) at (0.2,0.74) {$y$};
            \node[tag] (te3) at (1,0.25) {$z$};
            \node[tag,red] (te3) at (2.2,0.7) {$x$};
            \end{tikzpicture}
        }
        \hfill
        \subfloat[$G_2$]{
            \begin{tikzpicture}[
            vertex/.style={draw,circle,text width=5pt,align=center},
            tag/.style={text width=5pt,align=center}
            ]
            \node[vertex] (v1) at (0,1.41) {B};
            \node[tag] (t2) at (-0.7,1.41) {$u_1$};
            \node[vertex] (v2) at (2,1.41) {A};
            \node[tag] (t3) at (2.6,1.41) {$u_2$};
            \node[vertex] (v3) at (0,0) {A};
            \node[tag] (t2) at (-0.7,0) {$u_3$};
            \node[vertex] (v4) at (2,0) {C};
            \node[tag] (t3) at (2.6,0) {$u_4$};
            \path (v1) edge[-] (v3);
            \path (v1) edge[-] (v4);
            \path (v2) edge[-] (v4);
            \node[tag] (te1) at (-0.2,0.7) {$x$};
            \node[tag] (te3) at (1.2,0.9) {$z$};
            \node[tag] (te3) at (2.2,0.7) {$y$};
            \end{tikzpicture}
    }}
    \vspace{-1ex}
    \caption{\small{Graph edit process to transform $G_1$ into $G_2$ (Example \ref{intro-ged-example})}}
    \label{fig-intro-ged-example}
    \vspace{-2em}
\end{figure*}

With GED as the graph similarity measure, the \emph{graph similarity search} problem is formally stated as follows.

\textbf{Problem Statement}: {\textbf{(Graph Similarity Search)}}
Given a graph database $D$, a query graph $Q$, 
and a similarity threshold $\hat{\tau}$,
the graph similarity search is to find a set of graphs
$D_0\subseteq D$, where the graph edit distance (GED) between $Q$ and 
each graph in $D_0$ is less than or equal to $\hat{\tau}$.

A straightforward solution to the problem above is to check exact GEDs for all pairs of $Q$ and graphs in database $D$. However, despite its prevalence,
GED is proved to be NP-hard for exact calculations \cite{zeng2009comparing}, 
which may lead to unsatisfactory computational efficiency when we conduct a similarity search over large graphs.
The most widely-applied approach for computing exact GED is the $A^*$ algorithm \cite{hart1968formal},
which aims to search out the optimal matching between the vertices of two graphs in a heuristic manner.
Specifically, given two graphs with $n$ and $m$ vertices, respectively,
the time complexity of $A^*$ algorithm is $O({n}^{m})$ in the worst case. 

Due to the hardness of computing exact GED
(NP-hard)\cite{zeng2009comparing}, most existing works involving exact
GED computation \cite{zeng2009comparing} \cite{hart1968formal} only conducted
experiments on graphs with less than 10 vertices. In addition, a recent study \cite{DBLP:conf/icde/GoudaH16} indicates that the $A^*$ algorithm is incapable of computing GED between graphs with more than 12 vertices, which can hardly satisfy the demand for searching real-world graphs. For instance, a common requirement in bio-informatics is to search and compare molecular structures of similar proteins \cite{ibragimov2013gedevo}. However, the structures of human proteins usually contain hundreds of vertices (i.e., atoms) \cite{cellbiology2015}, which obviously makes similarity search beyond the capability of the approaches mentioned above. Another observation is that many real-world applications do not always require an exact similarity
value, and an approximate one with some quality guarantee is also
acceptable especially in real-time applications where the
responsiveness is far more important than the accuracy. Taking the protein search as an example again, it is certainly more desirable for users to obtain an approximate solution within a second, rather than to wait for a couple of days to get the exact answer. 

To address the problems above, many approaches have been
proposed to achieve an approximate GED between the query graph $Q$ and
each graph in database $D$ in polynomial time \cite{gao2010survey}, which can be leveraged to
accelerate the graph similarity search by trading accuracy for efficiency.
Assuming that there are two graphs with $n$ and $m$ vertices, respectively, where $n \ge m$,
one well-studied method \cite{bougleux2016graph} \cite{riesen2009approximate}  
for estimating the GED between these two graphs is to solve a corresponding
\emph{linear sum assignment problem}, which requires at least 
$O(n^3)$ time for obtaining the global optimal value or $O(n^2\log{n^2})$ time for a local optimal value
by applying the greedy algorithm \cite{riesen2015approximate}.
An alternative method is \emph{spectral seriation}\cite{robles2005graph},
which first generates the vector representations of graphs by extracting their
leading eigenvalues of the adjacency matrix ($O(n^2)$ time) \cite{golub2012matrix},
and then exploits a probabilistic model based on these vectors to estimate GED
in $O(nm^2)$ time.

To further enhance the efficiency of GED estimation and better satisfy the demands 
for graph similarity search on large graphs, we propose a novel
probabilistic approach which aims at estimating the GED with less time cost $O(nd+\hat\tau ^3)$, where $n$ is the number of vertices, $d$ is the average degree of the graphs involved, and $\hat\tau$ is the similarity 
threshold in the stated graph similarity search problem. Note that the similarity threshold $\hat{\tau}$ is often 
set as a small value (i.e., $\hat{\tau} \le 10$) and does not increase with the number of vertices $n$ in previous studies \cite{zeng2009comparing}
\cite{Zheng:2013hh}, thus, we can assume that $\hat{\tau}$ is a constant with regard to  $n$ when the graph is sufficiently large.  Moreover, most real-world graphs studied in related works \cite{riesen2009approximate} \cite{riesen2015approximate} 
are \emph{scale-free graphs} \cite{wiki:scale-free-graphs}, and we prove that the average degree $d=O(\log{n})$ for scale-free graphs.
Therefore, under the assumptions above, the time complexity of our approach is $O(nd+\hat{\tau}^3) = O(n\log{n})$ for comparing two scale-free graphs, and $O(|D|n\log{n})$ for searching similar graphs in the graph database $D$, where $|D|$ is the number of graphs in database $D$.

Our method is mainly inspired by probabilistic modeling approaches which are broadly utilized in similarity searches on various types of data such as text and images \cite{hu2009latent}. The basic idea of these methods is to model the formation of an object as a stochastic process, and to establish probabilistic connections between objects. In this paper, we follow this idea and model the formation of graph branch distances (GBDs) as the results of random graph editing processes, where GBD is defined in Section \ref{sec-branch-distance-definition}. By doing so, we prove the probabilistic relationship between GED and GBD, which is finally utilized to estimate GED by graph branch distance (GBD). 

To summarize, we make the following contributions.

\begin{itemize}[leftmargin=*]
    \item We adopt \emph{branches} \cite{Zheng:2013hh} as elementary structures in graphs, and define a novel graph similarity measure by comparing branches 
    between graphs, i.e., \emph{Graph Branch Distance} (GBD), which can be efficiently calculated in $O(nd)$ time, where $n$ is the number of vertices, and $d$ is the average degree of the graphs involved.
    \item We build a probabilistic model which reveals the relationship between GED and GBD 
    by considering branch variations as the result ascribed to graph edit operations.
    By applying our model, the GED between any two graphs can be estimated by their GBD in $O(\hat\tau^3)$ time, where $\hat\tau$ is the similarity 
    threshold in the graph similarity search problem.
    \item We conduct extensive experiments to show our approach has
    better accuracy, efficiency and scalability compared with the related
    approximation methods over real and synthetic data.
\end{itemize}

The paper is organized as follows. In Section
\ref{sec-preliminaries}, we formally define the symbols and
concepts which are used in this paper. In Section
\ref{sec-branch-distance-definition}, we give definitions of branches and \emph{Graph Branch Distance} (GBD).
In Section \ref{sec-extended-graphs}, we introduce the \emph{extended graphs},
which are exploited to simplify our model.
In Section \ref{sec-prob-relation-ged-gbd}, we derive the probabilistic
relation between GBD and GED, which is leveraged in Section
\ref{sec-search-framework} to perform the graph
similarity search. In Section \ref{sec-experiments}, we demonstrate the
efficiency and effectiveness of our proposed approaches through
extensive experiments. We discuss related works in Section
\ref{sec-related-works}, and we conclude the paper in Section \ref{sec-conclusions}.

\vspace{-0.5em}
\section{Preliminaries}\label{sec-preliminaries}
\vspace{-0.5em}

The \emph{graphs} discussed in this paper are restricted to
\emph{simple labeled undirected graphs}. Specifically, the $i$-th
graph in database $D$ is denoted by: $G_i\triangleq\{V_i, E_i,
\mathcal{L}\}$, where $V_i\triangleq\{v_{i,1}, v_{i,2},...,\allowbreak 
v_{i,|V_i|}\}$ is the set of vertices, $E_i\triangleq\{e_{i,1}, e_{i,2},...,\allowbreak
e_{i,|E_i|}\}$ is the set of edges, while $\mathcal{L}$ is a 
general labelling function. For any vertex $v_{i,j}\in V_i$, its label is given by
$\mathcal{L}(v_{i,j})$. Similarly, for any edge $e_{i,j}\in E_i$, its
label is given by $\mathcal{L}(e_{i,j})$. In addition, $\mathcal{L}_V$ and $\mathcal{L}_E$
are defined as the sets of all possible labels for vertices and edges, respectively. 
We also define $\varepsilon$ as a \emph{virtual}
label, which will be used later in our approach. When the label of a
vertex (or edge) is $\varepsilon$, the vertex (or edge) is said to
be \emph{virtual} and does not actually exist.  Particularly, we have 
$\varepsilon \notin \mathcal{L}_V$ and $\varepsilon \notin \mathcal{L}_E$. 
Note that our method can also handle directed and weighted graphs by considering edge directions and weights as special labels.

In this paper, we take Graph Edit Distance (GED)\cite{sanfeliu1983distance} as the graph similarity measure, which is defined as follows.

\begin{definition}[{Graph Edit Distance}]\label{def-ged}
    The edit distance between graphs $G_1$ and $G_2$, denoted by
    $GED(G_1,G_2)$, is the minimal number of graph edit operations which
    are necessary to transform $G_1$ into $G_2$, where the graph edit
    operations (GEO) are restricted to the following six types:
    
    \begin{itemize}
        \item \textbf{AV}: \underline{A}dd one isolated \underline{v}ertex with a non-virtual label;
        \item \textbf{DV}: \underline{D}elete one isolated \underline{v}ertex;
        \item \textbf{RV}: \underline{R}elabel one \underline{v}ertex;
        \item \textbf{AE}: \underline{A}dd one \underline{e}dge with a non-virtual label;
        \item \textbf{DE}: \underline{D}elete one \underline{e}dge;
        \item \textbf{RE}: \underline{R}elabel one \underline{e}dge.
    \end{itemize}
\end{definition}

Assume that a particular graph edit operation sequence which
transforms graph $G_1$ into $G_2$ is $seq_i$, where $i$ is the unique ID of this sequence.
Then, according to Definition \ref{def-ged}, $GED(G_1, G_2)$ is the
minimal length for all possible operation sequences, that is,
$GED(G_1, G_2)=\min_i\{|seq_i|\} $, where $|seq_i|$ is the length of the sequence
$seq_i$. The set of all operation sequences from
$G_1$ to $G_2$ of the minimal length is defined as $SEQ=\{seq_i \vert \forall i, GED(G_1,
G_2)=|seq_i| \}$.

\vspace{-1.5ex}
\section{Branch Distance Between Graphs}\label{sec-branch-distance-definition}
\vspace{-0.5em}

To reduce the high cost of exact GED computations (NP-hard \cite{zeng2009comparing}) in
the graph similarity search, one widely-applied strategy for pruning search results
\cite{zeng2009comparing} \cite{wang2012efficiently}\cite{zhao2012efficient} \cite{Zheng:2013hh} is to exploit the
differences between graph sub-structures as the bounds of exact GED values. 
In this paper, we consider the branches \cite{Zheng:2013hh} as elementary graph units, which are defined as:

\begin{definition}[{Branches}]\label{def-branch}
    The branch rooted at vertex $v$ is defined as $B(v)=\{\mathcal{L}(v), N(v)\}$, where $\mathcal{L}(v)$ is the label of vertex $v$, and $N(v)$ is the sorted multiset containing all labels of edges adjacent to $v$. The sorted multiset of all branches in $G_i$ is denoted by $B_{G_i}=\{B(v_{i,j}), \forall v_{i,j} \in V_i\}$.
\end{definition}

In practice, each branch $B(v)$ is stored as a list of 
strings whose first element is $\mathcal{L}(v)$ and the following elements
are strings in the sorted multiset $N(v)$. In addition, $B_{G_i}$ for each
graph $G_i$ is stored in a sorted multiset, whose elements are essentially lists of strings (i.e., branches) and are always sorted ascendingly by the ordering algorithm in \cite{cpp:algorithm-sort-multiset}.
For a fair comparison of the computational efficiency,
we assume that all auxiliary data structures in different methods,
such as the cost matrix \cite{riesen2009approximate}, adjacency matrix \cite{robles2005graph},
and our branches, are pre-computed and stored with graphs.

To define the equality between two branches, we introduce the concept of \emph{branch isomorphism} as follows.

\begin{definition}[Branch Isomorphism] \label{def-branch-iso}
    Two branches $B(v)=\{\mathcal{L}(v), \allowbreak N(v)\}$ and $B(u)=\{\mathcal{L}(u), N(u)\}$ are isomorphic if and only if $\mathcal{L}(v)=\mathcal{L}(u)$ and $N(v)=N(u)$, which is denoted by $B(v) \simeq B(u)$.
\end{definition}

Suppose that we have two branches $B(v)$ and $B(u)$ where the degrees of
vertices $v$ and $u$ are $d_v$ and $d_u$, respectively. From previous discussions, the 
branch $B(v)$ and $B(u)$ are stored as lists of lengths $d_v+1$ and $d_u+1$, respectively. Therefore, checking whether $B(v)$ and $B(u)$ are isomorphic is essentially judging whether two lists of lengths $d_v+1$ and $d_u+1$ are identical, which can be done in $d_v$ time when $d_v = d_u$, and otherwise in one unit time since the length of a list can be obtained in one unit time.

Finally, we define the Graph Branch Distance (GBD).

\begin{definition}[{Graph Branch Distance}]\label{def-gbd}
    The branch distance between graphs $G_1$ and $G_2$, denoted by
    $GBD(G_1,G_2)$, is defined as:
    \begingroup
    \setlength{\abovedisplayskip}{2pt}
    \setlength{\belowdisplayskip}{2pt}
    \setlength{\abovedisplayshortskip}{2pt}
    \setlength{\belowdisplayshortskip}{2pt}
    \begin{align} \label{equal-def-gbd}
    GBD(G_1,G_2)&=\max\{|B_{G_1}|, |B_{G_2}|\}-|B_{G_1}\cap B_{G_2}| \nonumber \\
    &=\max\{|V_1|, |V_2|\}-|B_{G_1}\cap B_{G_2}|
    \end{align}
    \endgroup
    where $B_{G_1}$ and $B_{G_2}$ are the multisets of all branches in graphs $G_1$ and $G_2$, respectively.
\end{definition}

\begin{table}[!t]
    \vspace{-10pt}
    \centering
    \caption{Table of Notations}
    \vspace{-10pt}
    \begin{tabular}{p{0.7cm} p{0.1cm} p{6.6cm} } \hline
        $D$ & $\triangleq$ & $\{G_1, G_2, ... G_{|D|}\}$, the graph database\\
        $G_i$ & $\triangleq$ & $\{V_i, E_i, \mathcal{L}\}$, $i$-th graph in database\\
        $V_i$ & $\triangleq$ & $\{v_{i,1}, v_{i,2},..., v_{i,|V_i|}\}$, the vertices in $G_i$\\
        $E_i$ & $\triangleq$ &$\{e_{i,1}, e_{i,2},..., e_{i,|E_i|}\}$, the edges in $G_i$\\
        $Q$ & $\triangleq$ & $\{V_Q, E_Q, \mathcal{L}\}$, the query graph\\
        $\mathcal{L}$ & $\triangleq$ &labelling function for vertices and edges\\
        $\mathcal{L}_V$ & $\triangleq$ &the set of all possible vertex labels\\
        $\mathcal{L}_E$ & $\triangleq$ &the set of all possible edge labels\\
        $\varepsilon$ & $\triangleq$ &the \emph{virtual} label\\
        $\hat\tau$ & $\triangleq$ &the similarity threshold\\
        
        \multicolumn{3}{c}{}\\[-4mm]
        \hline
        \multicolumn{3}{c}{}\\[-3mm]
        $seq$ & $\triangleq$ & a graph edit operation (GEO) sequence\\
        $SEQ$ & $\triangleq$ & set of all GEO sequences of minimal lengths \\
        $GED$ & $\triangleq$ & Graph Edit Distance\\  
        
        \multicolumn{3}{c}{}\\[-4mm]
        \hline
        \multicolumn{3}{c}{}\\[-3mm]
        $B(v)$ & $\triangleq$ & the branch rooted at vertex $v$\\
        $\mathcal{L}(v)$ & $\triangleq$ & the label of vertex $v$ \\
        $N(v)$ & $\triangleq$ & sorted multiset of labels of edges adjacent to $v$ \\
        $B_{G_i}$ & $\triangleq$ &sorted multiset of all branches in graph $G_i$\\
        $GBD$ & $\triangleq$ & Graph Branch Distance\\
        
        \multicolumn{3}{c}{}\\[-4mm]
        \hline
        \multicolumn{3}{c}{}\\[-3mm]
        $G^{\{k\}}$ & $\triangleq$ &extended graph of $G$ with extension factor $k$\\
        $G_1',G_2'$ & $\triangleq$ &$G_1^{\{|V_2|-|V_1|\}}$, $G_2^{\{0\}}$ when $|V_1| \le |V_2|$\\
        \hline\end{tabular}
    \vspace{-25pt}
\end{table}

The intuition of introducing GBD is as follows. The state-of-the-art method \cite{Zheng:2013hh} for pruning graph similarity search results assumes that the difference between branches of two graphs has a close relation to their GED. Therefore, in this paper, we aim to use GBD to closely estimate the GED of two graphs.

Example \ref{example-gbd-computation} below illustrates the process of computing GBD.

\begin{example} \label{example-gbd-computation}
    Assume that we have two graphs $G_1$ and $G_2$ as shown in Figure \ref{fig-intro-ged-example}. The branches rooted at the vertices in graphs $G_1$ and $G_2$ are as follows.
    \begingroup
    \setlength{\abovedisplayskip}{2pt}
    \setlength{\belowdisplayskip}{2pt}
    \setlength{\abovedisplayshortskip}{2pt}
    \setlength{\belowdisplayshortskip}{2pt}
    \begin{align*}
    &B(v_1)=\{A; y,y\},  B(v_2)=\{C; y,z\},  B(v_3)=\{B; y,z\}; \\
    &B(u_1)=\{B; x,z\},  B(u_2)=\{A; y\};\\
    &B(u_3)=\{A; x\}, B(u_4)=\{C; y,z\}.
    \end{align*}
    \endgroup
    
    The sorted multisets of branches in $G_1$ and $G_2$ are:
    \begingroup
    \setlength{\abovedisplayskip}{2pt}
    \setlength{\belowdisplayskip}{2pt}
    \setlength{\abovedisplayshortskip}{2pt}
    \setlength{\belowdisplayshortskip}{2pt}
    \begin{align*}
    {B}_{G_1}&=\{{B}(v_1), {B}(v_3), {B}(v_2)\}; \\
    {B}_{G_2}&=\{{B}(u_3), {B}(u_2), {B}(u_1), {B}(u_4)\}.
    \end{align*}
    \endgroup
    
    Therefore, according to Definition \ref{def-gbd}, we can obtain the graph branch distance (GBD) between graphs $G_1$ and $G_2$ by applying the Equation (\ref{equal-def-gbd}), which is:
    \begingroup
    \setlength{\abovedisplayskip}{2pt}
    \setlength{\belowdisplayskip}{2pt}
    \setlength{\abovedisplayshortskip}{2pt}
    \setlength{\belowdisplayshortskip}{2pt}
    \begin{align*}
    GBD(G_1,G_2)=\max\{|B_{G_1}|, |B_{G_2}|\}-|{B}_{G_1}\cap {B}_{G_2}|=3,
    \end{align*}
    \endgroup
    where $|B_{G_1}|=3$ and $|B_{G_2}|=4$ are the sizes of multisets $B_{G_1}$ and $B_{G_2}$, respectively. In addition, according to Definition \ref{def-branch-iso}, the only pair of isomorphic branches between multisets $B_{G_1}$ and $B_{G_2}$ is $B(v_2) \simeq B(u_4)$. Therefore, the intersection set of multisets $B_{G_1}$ and $B_{G_2}$ is $\{B(v_2)\}$, whose size is $|{B}_{G_1}\cap {B}_{G_2}| = |\{B(v_2)\}| = 1$.
\end{example}

Note that the time cost of computing the \emph{size} of intersection of
two \emph{sorted} multisets is $\max\{m_1, m_2\}$ \cite{zeng2009comparing},
where $m_1$ and $m_2$ are the sizes of two multisets, respectively.
Therefore, the GBD between query graph $Q$ and any 
graph $G \in D$ can be computed in time:
\begingroup
\setlength{\abovedisplayskip}{2pt}
\setlength{\belowdisplayskip}{2pt}
\setlength{\abovedisplayshortskip}{2pt}
\setlength{\belowdisplayshortskip}{2pt}
\begin{align}
\textstyle \sum_i^n{d_i} = O(nd),
\end{align}
\endgroup
where $n=\max\{|V_{Q}|, |V_{G}|\}$, $d_i$ is the degree of $i$-th compared
vertex in $G$, and $d$ is the average degree of graph $G$.

The GBD defined in this section is utilized to model the graph edit process
and further leveraged for estimating the graph edit distance (GED) in Section \ref{sec-prob-relation-ged-gbd}. 

\vspace{-2ex}
\section{Extended Graphs}\label{sec-extended-graphs}
\vspace{-0.4em}

In this section, we reduce the number of graph edit operation types that need to be considered by extending graphs with \emph{virtual} vertices and edges, which helps to simplify our probabilistic model in Section \ref{sec-prob-relation-ged-gbd}. Moreover, we show that the GED and the GBD between the extended graphs stay the same as the GED and the GBD between the original ones, respectively.

\begin{figure}[!t]
    \centering
    \vspace{-25pt}
    \captionsetup[subfigure]{labelformat=empty}
    \scalebox{0.95}{
        \subfloat[$G_1^{\{1\}}$]{
            \begin{tikzpicture}[
            vertex/.style={draw,circle,text width=5pt,align=center},
            virtual/.style={draw,circle,densely dashed,text width=6pt,align=center},
            tag/.style={text width=5pt,align=center}
            ]
            \node[vertex] (v1) at (0,1.41) {A};
            \node[tag] (t1) at (-0.7,1.41) {$v_1$};
            \node[vertex] (v2) at (0,0) {C};
            \node[tag] (t2) at (-0.7,0) {$v_2$};
            \node[vertex] (v3) at (2,0) {B};
            \node[tag] (t3) at (2.6,0) {$v_3$};
            \node[virtual] (v4) at (2,1.41) {$\varepsilon$};
            \node[tag] (t3) at (2.57,1.41) {$v_4$};
            \path (v1) edge[-] (v2);
            \path (v3) edge[-] (v2);
            \path (v1) edge[-] (v3);
            \path (v1) edge[densely dashed] (v4);
            \path (v2) edge[densely dashed] (v4);
            \path (v3) edge[densely dashed] (v4);
            \node[tag] (te1) at (-0.2,0.7) {$y$};
            \node[tag] (te2) at (0.7,0.7) {$y$};
            \node[tag] (te3) at (1,0.15) {$z$};
            \end{tikzpicture}
        }
        \hspace{2em}
        \subfloat[$G_2^{\{0\}}$]{
            \begin{tikzpicture}[
            vertex/.style={draw,circle,text width=5pt,align=center},
            tag/.style={text width=5pt,align=center}
            ]
            \node[vertex] (v1) at (0,1.41) {B};
            \node[tag] (t2) at (-0.7,1.41) {$u_1$};
            \node[vertex] (v2) at (2,1.41) {A};
            \node[tag] (t3) at (2.6,1.41) {$u_2$};
            \node[vertex] (v3) at (0,0) {A};
            \node[tag] (t2) at (-0.7,0) {$u_3$};
            \node[vertex] (v4) at (2,0) {C};
            \node[tag] (t3) at (2.6,0) {$u_4$};
            \path (v1) edge[-] (v3);
            \path (v1) edge[-] (v4);
            \path (v2) edge[-] (v4);
            \path (v1) edge[densely dashed] (v2);
            \path (v3) edge[densely dashed] (v2);
            \path (v3) edge[densely dashed] (v4);
            \node[tag] (te1) at (-0.2,0.7) {$x$};
            \node[tag] (te3) at (0.7,0.7) {$z$};
            \node[tag] (te3) at (2.2,0.7) {$y$};
            \end{tikzpicture}
    }}
    \vspace{-5pt}
    \caption{Extended Graphs for Example \ref{example-extended-graph}}
    \label{fig-extended-graph-example}
    \vspace{-4ex}
\end{figure}
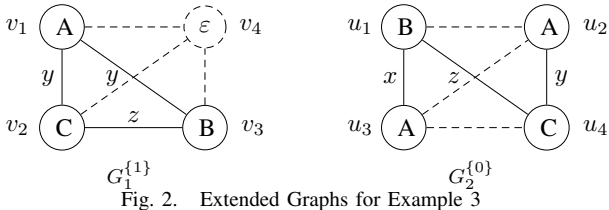

The definition of extended graphs is as follows.

\begin{definition}[Extended Graphs]
    For graph $G$, its extended graph, denoted by $G^{\{k\}}$,
    is generated by first inserting $k$ isolated virtual vertices
    into $G$, and then inserting a virtual edge between each pair
    of non-adjacent vertices in $G$ (including virtual vertices),
    where $k$ is the extension factor.
\end{definition}

\begin{example}\label{example-extended-graph}
    For graphs $G_1$ and $G_2$ in Figure \ref{fig-intro-ged-example},
    their extended graphs $G_1^{\{1\}}$ and $G_2^{\{0\}}$ are shown in Figure \ref{fig-extended-graph-example}.
    The virtual vertices are labelled by $\varepsilon$, while virtual edges are represented by dashed lines. 
    Note that when the extension factor is $0$, no virtual vertex will be inserted.
\end{example}

In particular, for any graph pair $(G_1,G_2)$ where $|V_1| \le |V_2|$,
we define $G_1'=G_1^{\{|V_2|-|V_1|\}}$ and $G_2'=G_2^{\{0\}}$ by
extending $G_1$ and $G_2$ with extension factor $|V_2|-|V_1|$ and $0$, respectively. Previous studies \cite{justice2006binary} \cite{serratosa2014fast} have shown that, for any graph edit operation sequence $seq$ which transforms the extended graph $G_1'$ into $G_2'$ and has the minimal length, every operation in $seq$
is equivalent to a relabelling operation on $G_1'$. Therefore, we only
need to consider graph edit operations of types \textbf{RV} and \textbf{RE} when modeling the graph edit
process of transforming the extended graph $G_1'$ into $G_2'$. 

Finally, given the graphs $G_1$ and $G_2$ (for $|V_1| \le |V_2|$), and their extended graphs $G_1'$and $G_2'$, we have the following Theorems \ref{theorem-preserve-ged} and \ref{theorem-preserve-gbd}, which are utilized in the Section \ref{sec-prob-relation-ged-gbd}.

\begin{theorem}\label{theorem-preserve-ged}
    $GED(G_1,G_2)= GED(G_1',G_2')$
\end{theorem}

\begin{proof}
    Please refer to Appendix \ref{Appendix-Proof-Theorem-1}.
\end{proof}

\begin{theorem}\label{theorem-preserve-gbd}
    $GBD(G_1,G_2)= GBD(G_1',G_2')$
\end{theorem}

\begin{proof}
    Please refer to Appendix \ref{Appendix-Proof-Theorem-2}.
\end{proof}

Note that the extended graph is only a conceptual model for reducing the number of graph edit operation types that need to be considered. According to Theorems \ref{theorem-preserve-ged} and \ref{theorem-preserve-gbd}, whenever the values of $GED(G_1',G_2')$ and $GBD(G_1',G_2')$ are required, we can instead calculate $GED(G_1,G_2)$ and $GBD(G_1,G_2)$. Therefore, in practice, we do not actually convert graphs into their extended versions, and the calculations of GEDs and GBDs are still conducted on original graphs rather than the extended ones, which means that there is no overhead for creating and maintaining extended graphs.

\vspace{-0.4em}
\section{Our Probabilistic Model} \label{sec-prob-relation-ged-gbd}
\vspace{-0.4em}

In this section, we aim to solve the stated graph similarity search problem
by estimating GED from GBD in a 
probabilistic manner. According to Theorems \ref{theorem-preserve-ged} and \ref{theorem-preserve-gbd}, the relation between original graphs' $GED$ and $GBD$ must be the same with the relation between extended graphs' $GED$ and $GBD$. Therefore, as discussed in Section \ref{sec-extended-graphs}, we only
need to consider graph edit operations of types \textbf{RV} and \textbf{RE} when building our probabilistic model.

To be more specific, we consider two given graphs $G_1$ and $G_2$ where $|V_1| \le |V_2|$ and their extended graphs $G'_1$ and $G'_2$. For simplicity, we denote $GED(G_1',G_2')$ and $GBD(G_1',G_2')$ by $GED$ and $GBD$, respectively.  As mentioned in \mbox{Section \ref{sec-intro}}, we model the formation of graph branch distances (GBDs) as the results of random graph editing processes, and thus we establish a probabilistic connection between GED and GBD. 
The detailed steps to construct our model is as follows.\vspace{-1ex}

\begin{figure}[!t]
    \vspace{-30pt}
    \centering
    \scalebox{0.8}{
        \begin{tikzpicture}[
        bnnode/.style={draw,circle,align=center},
        gednode/.style={draw,ellipse,text width=1cm,align=center,fill=gray!40},
        gbdnode/.style={draw,ellipse,text width=1cm,align=center}
        ]
        \node[gednode] (ged) at (-1.7,0) {$GED$};
        \node[bnnode] (s) at (0,0) {$S$};
        \node[bnnode] (x) at (0.7,1) {$X$};
        \node[bnnode] (y) at (0.7,-1) {$Y$};
        \node[bnnode] (z) at (2,-1) {$Z$};
        \node[bnnode] (r) at (2.7,0) {$R$};
        \node[gbdnode] (gbd) at (4.3,0) {$GBD$};
        \path (s) edge[-latex] (x)
        (s) edge[-latex] (y)
        (x) edge[-latex] (y)
        (x) edge[-latex] (r)
        (y) edge[-latex] (z)
        (z) edge[-latex] (r)
        (ged) edge[-latex] (s)
        (r) edge[-latex] (gbd);
        \end{tikzpicture}}
    \vspace{-6pt}
    \caption{\small{Bayesian Network of Random Variables}}
    \label{fig-bayes-net}
    \vspace{-3ex}
\end{figure}
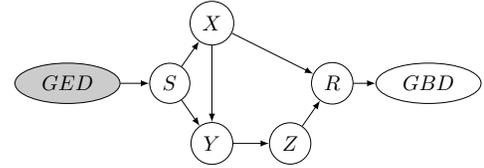

\begin{description}[leftmargin=0pt]
    \item[Step 1:] We consider $GED$ as an observed random variable.
    
    \item[Step 2:] We randomly choose one graph edit operation (GEO) sequence from all possible GEO sequences whose lengths are equal to $GED$. We define a random variable $X(\omega)$, where $\omega$ is a particular choice of operation sequence from $SEQ$, and $S(\omega)=s$ iff $\omega$ choose the sequence with ID $s$, that is, $seq_s$.
    
    \item[Step 3:] We model the numbers of \textbf{RV} and \textbf{RE} operations in the sequence chosen in Step 2 as random variables $X(s)$ and $Y(s)$, respectively, where $X(s)=x$ iff the number of operations with type \textbf{RV} in $seq_s$ is $x$, and $Y(s)=y$ iff the number of operations with type \textbf{RE} in $seq_s$ is $y$. Note that, when given $GED=\tau$ and $X=x$, we always have $Y=\tau-x$.
    
    \item[Step 4:] We define random variables $Z(y)$ as the random variable where $y$ is a particular value of $Y$, and $Z(y)=m$ iff $Y(s)=y$ and the number of vertices covered by relabeled edges is $m$ when conducting $seq_s$. In addition, we define $R(x,m)$ as the random variable where $x$ and $m$ are particular values of $X$ and $Z$, respectively. That is, $R(x,m)=r$ iff $X=x$, $Z=m$ and the number of vertices either relabelled or covered by relabelled edges is $r$.
    
    \hspace*{1em} The reason we conduct Step 4 is because we want to model branch variations by the number of branches rooted at the vertices either relabelled or covered by relabelled edges, i.e., the random variable $R$.
    
    \item[Step 5:] We consider $GBD$ as the random variable dependent on $R$, where their relation is proved in Appendix \ref{proof-lemma-cal-omega3}.
\end{description}

The random variables defined in the five-step model above are listed in Table \ref{table-notation-random}, and their relations among are represented by a Bayesian Network, as shown in Figure \ref{fig-bayes-net}.
We use Example \ref{example-prob-model-idea} below to better illustrate the key idea of our model.

\begin{table}[!t]
    \caption{Notation of Random Variables}
    \label{table-notation-random}
    \vspace{-10pt}
    \centering
    \begin{tabular}{|p{0.3em}p{0.3em}p{0.85\linewidth}|}
        \hline
        $S$ & $\triangleq$ & Choice of operation sequence \\
        $X$ & $\triangleq$ & Number of relabeled vertices \\
        $Y$ & $\triangleq$ & Number of relabeled edges \\
        $Z$ & $\triangleq$ & Number of vertices covered by relabeled edges \\
        $R$ & $\triangleq$ & Number of vertices either relabeled or covered by relabeled edges \\
        \hline
    \end{tabular}
    \vspace{-20pt}
\end{table}

\begin{example} \label{example-prob-model-idea}
    Given two extended graphs $G_1'$ and $G_2'$, as shown in Figure \ref{fig-prob-model-idea}, where the virtual edges are represented by dashed lines. 
    The minimal number of graph edit operations 
    to transform $G_1'$ into $G_2'$ is $2$, and the set of all possible graph edit operation sequences with the minimal length $2$ is:
    \begingroup
    \setlength{\abovedisplayskip}{2pt}
    \setlength{\belowdisplayskip}{2pt}
    \setlength{\abovedisplayshortskip}{2pt}
    \setlength{\belowdisplayshortskip}{2pt}
    $$SEQ=\{\{ op_1,op_2 \}_1, \{ op_2,op_1 \}_2, \{ op_3,op_4 \}_3, \{ op_4,op_3 \}_4\}$$
    \endgroup
    
    where the subscript of each sequence is its ID, and

    \vspace{-13pt}
    \begingroup
    \setlength{\abovedisplayskip}{2pt}
    \setlength{\belowdisplayskip}{2pt}
    \setlength{\abovedisplayshortskip}{2pt}
    \setlength{\belowdisplayshortskip}{2pt}
    \begin{align*}
    op_1&=\text{Relabelling the edge $\langle v_1, v_2\rangle$ to label $y$;} \\
    op_2&=\text{Relabelling the edge $\langle v_1, v_3\rangle$ to label $x$;} \\
    op_3&=\text{Relabelling the vertex $v_2$ to label $C$.} \\
    op_4&=\text{Relabelling the vertex $v_3$ to label $B$.}
    \end{align*}
    \endgroup
    Then, the values of random variables in our model could be:
    
    \begin{enumerate}[leftmargin=*]
        \item First, we consider $GED$ as a random variable, which has the value $2$ in this example.  ($GED=2$)
        \item Second, we randomly choose one sequence from $SEQ$, which is $seq_2=\{op_2,op_1\}$. Therefore, in this case the random variable $S=2$.
        \item Third, the numbers of \textbf{RV} and \textbf{RE} operations in $seq_2$ are $0$ and $2$, respectively. Therefore, the random variables $X=0$ and $Y=2$ in this example. 
        \item Then, after conducting operations in $seq_2$, the number of vertices covered by relabelled edges is $3$, and the number of vertices either relabelled or covered by relabelled edges is $3$. Therefore,  the random variables $Z=3$ and $R=3$.
        \item Finally, $GBD$ is considered to be the random variable, where $GBD=2$ in this example.
    \end{enumerate}
\end{example}

\begin{figure}[!t]
    \centerfloat
    \captionsetup[subfigure]{labelformat=empty}
    \scalebox{0.8}{
        \subfloat[$G_1'$]{
            \begin{tikzpicture}[
            vertex/.style={draw,circle,text width=5pt,align=center},
            tag/.style={text width=5pt,align=center}
            ]
            \node[vertex] (v1) at (1,1.41) {A};
            \node[tag] (t1) at (0.3,1.41) {$v_1$};
            \node[vertex] (v2) at (0,0) {B};
            \node[tag] (t2) at (-0.7,0) {$v_2$};
            \node[vertex] (v3) at (2,0) {C};
            \node[tag] (t3) at (2.6,0) {$v_3$};
            \path (v1) edge[-] (v2);
            \path (v1) edge[-] (v3);
            \path (v2) edge[densely dashed] (v3);
            \node[tag] (te1) at (0.2,0.74) {$x$};
            \node[tag] (te2) at (1.72,0.74) {$y$};
            \end{tikzpicture}
        }
        \hspace{2em}
        \subfloat[$G_2'$]{
            \begin{tikzpicture}[
            vertex/.style={draw,circle,text width=5pt,align=center},
            tag/.style={text width=5pt,align=center}
            ]
            \node[vertex] (v1) at (1,1.41) {A};
            \node[tag] (t1) at (0.3,1.41) {$u_1$};
            \node[vertex] (v2) at (0,0) {B};
            \node[tag] (t2) at (-0.7,0) {$u_2$};
            \node[vertex] (v3) at (2,0) {C};
            \node[tag] (t3) at (2.6,0) {$u_3$};
            \path (v1) edge[-] (v2);
            \path (v1) edge[-] (v3);
            \path (v2) edge[densely dashed] (v3);
            \node[tag] (te1) at (0.2,0.74) {$y$};
            \node[tag] (te2) at (1.72,0.74) {$x$};
            \end{tikzpicture}
    }}
    \vspace{-1ex}
    \caption{\small{Extended Graphs for Example \ref{example-prob-model-idea}}}
    \label{fig-prob-model-idea}
\end{figure}
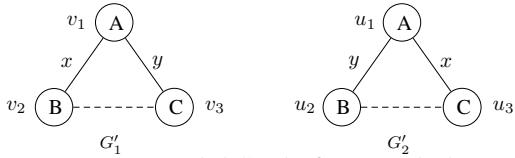

In this paper, we aim to infer the probability distribution of $GED$ when given $GBD$, which is essentially to calculate the following probability for given constants $\hat{\tau}$ and $\varphi$:
\begingroup
\setlength{\abovedisplayskip}{2pt}
\setlength{\belowdisplayskip}{2pt}
\setlength{\abovedisplayshortskip}{2pt}
\setlength{\belowdisplayshortskip}{2pt}
\begin{align}\label{equal-ged-condon-gbd-origin}
&Pr[GED \le \hat{\tau} \mid GBD=\varphi] \nonumber \\
=&\textstyle \sum_{\tau=0}^{\hat{\tau}}Pr[GED = {\tau} \mid GBD=\varphi]
\end{align}
\endgroup

By applying Bayes' Rule, we have:
\begingroup
\setlength{\abovedisplayskip}{2pt}
\setlength{\belowdisplayskip}{2pt}
\setlength{\abovedisplayshortskip}{2pt}
\setlength{\belowdisplayshortskip}{2pt}
\begin{align}\label{equal-ged-condon-gbd-bayes-formula}
Pr[GED \le \hat{\tau} \mid GBD=\varphi]
=\textstyle \sum_{\tau=0}^{\hat{\tau}} \Lambda_1 \cdot \frac{\Lambda_3}{\Lambda_2},
\end{align}
\endgroup
where
\begingroup
\setlength{\abovedisplayskip}{2pt}
\setlength{\belowdisplayskip}{2pt}
\setlength{\abovedisplayshortskip}{2pt}
\setlength{\belowdisplayshortskip}{2pt}
\begin{align}\label{equal-lambdas-origin}
&\Lambda_1(G_1',G_2';\tau, \varphi) = Pr[GBD=\varphi \mid GED = \tau] \\
&\Lambda_2(G_1',G_2';\varphi) = Pr[GBD = \varphi] \\
&\Lambda_3(G_1',G_2';\tau) = Pr[GED = \tau]
\end{align}
\endgroup

Therefore, the problem to solve becomes calculating the values of  $\Lambda_1$, $\Lambda_2$ and $\Lambda_3$ when given extended graphs $G_1'$ and $G_2'$, and constants $\tau$ and $\varphi$.
In this following three subsections \ref{section-lambda-1}, \ref{section-lambda-2} and \ref{section-lambda-3}, we illustrate how to utilize our model to calculate $\Lambda_1$, $\Lambda_2$ and $\Lambda_3$, respectively. 

\vspace{-0.3em}
\subsection{Calculating $\mathbf{\Lambda_1}$} \label{section-lambda-1}
\vspace{-0.3em}

In this subsection, we aim to calculate $\Lambda_1=Pr[GBD=\varphi \mid GED = \tau]$. The key idea to calculate $\Lambda_1$ is to expand its formula by applying Chain Rule \cite{wiki:chain-rule} according to the dependency relationship between random variables in our model, where the expanded formula is:
\begingroup
\setlength{\abovedisplayskip}{2pt}
\setlength{\belowdisplayskip}{2pt}
\setlength{\abovedisplayshortskip}{2pt}
\setlength{\belowdisplayshortskip}{2pt}
\begin{align} \label{equal-lambda-1-expansion}
\textstyle \Lambda_1 = \sum_{x}\Omega_1 \sum_{m}\Omega_2 \sum_{r} \Omega_3 \cdot \Omega_4
\end{align}
\endgroup

where
\begingroup
\setlength{\abovedisplayskip}{2pt}
\setlength{\belowdisplayskip}{2pt}
\setlength{\abovedisplayshortskip}{2pt}
\setlength{\belowdisplayshortskip}{2pt}
\begin{align} \label{equal-lambda-1-expansion-omegas}
\Omega_1 & \textstyle =\frac{1}{N} \sum_{s}Pr[X=x, Y=\tau-x \mid S=s] \\
\Omega_2 & \textstyle =Pr\left[Z=m \mid Y=\tau-x\right] \\
\Omega_3 & \textstyle = Pr\left[GBD=\varphi \mid R=r\right] \\
\Omega_4 &\textstyle=Pr\left[R=r \mid X=x, Z=m\right] 
\end{align}
\endgroup

Please refer to Appendix \ref{Appendix-Closed-Forms-of-Equations} for the closed forms of $\Omega_1$, $\Omega_2$, $\Omega_3$, and $\Omega_4$. 
Due to the space limitation, please refer to Appendix \ref{proof-theo-all-trans} for the proof of Equation (\ref{equal-lambda-1-expansion}).

\vspace{-2ex}
\subsection{Calculating $\mathbf{\Lambda_2}$} \label{section-lambda-2}
\vspace{-0.4em}

In this subsection, we aim to calculate $\Lambda_2$,
which is essentially to infer the \emph{prior distributions} of GBDs. In practice, we pre-compute the prior distribution of GBDs without knowing the query graph $Q$ in the graph similarity search problem.
This is because in most real-world scenarios \cite{zeng2009comparing} \cite{Zheng:2013hh}, 
the query graph $Q$ ofter comes from the same population as graphs in database $D$. As a counter example, it is 
unusual to use a query graph of protein structure to search for similar graphs
in social networks, and vice versa.
Therefore, we assume that GBDs between $Q$ and
each graph $G$ in database $D$, follow the same prior distributions as those among all graph pairs in $D$.

\begin{figure}[!t]
    \centerfloat
    \minipage{0.24\textwidth}
    \includegraphics[width=\linewidth]{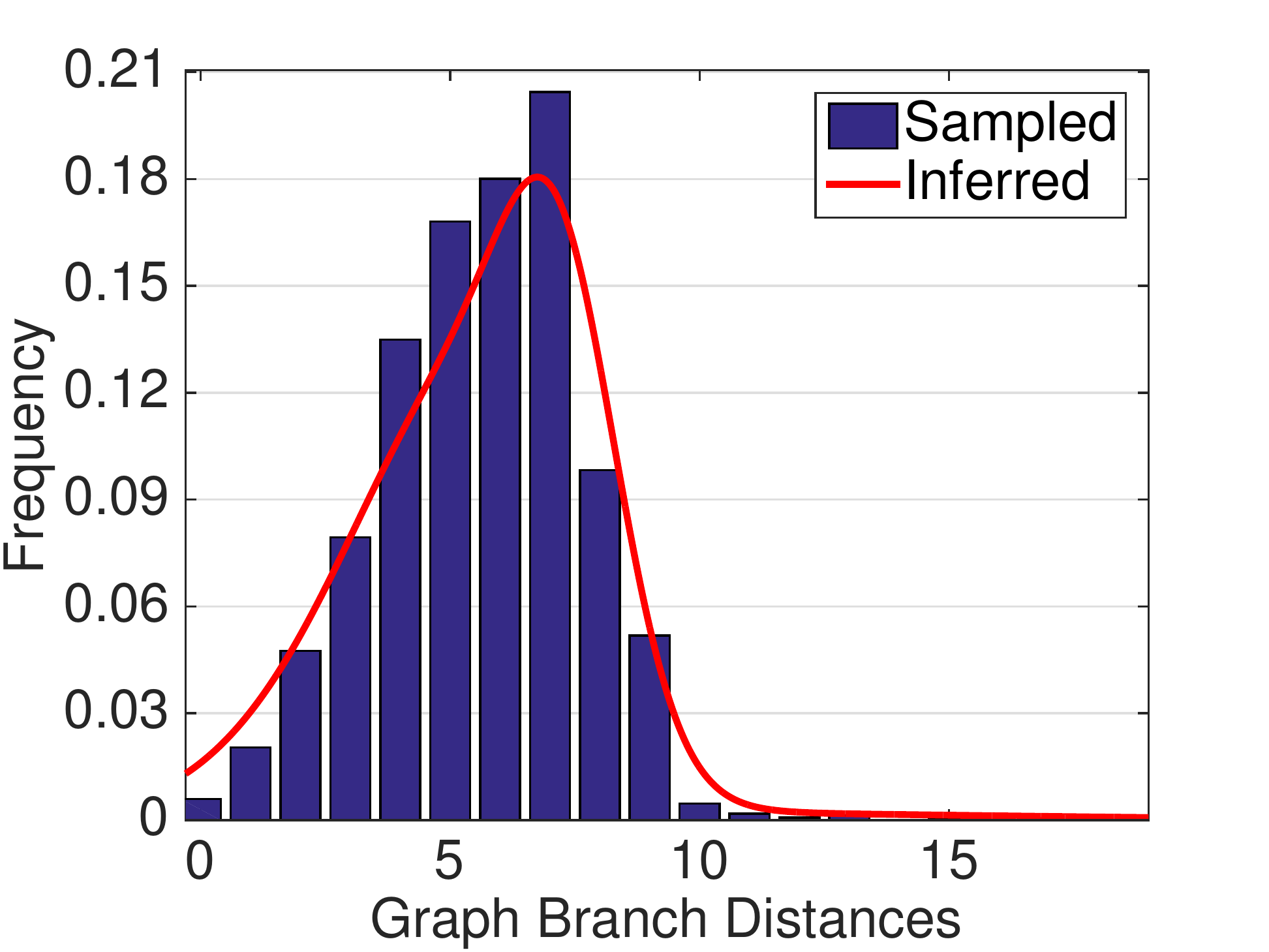}
    \vspace{-20pt}
    \caption{\small{Inferred prior distribution of GBDs on Fingerprint data set}}\label{fig-prior-gbd}
    \endminipage
    \hspace{10pt}
    \minipage{0.24\textwidth}
    \includegraphics[width=\linewidth]{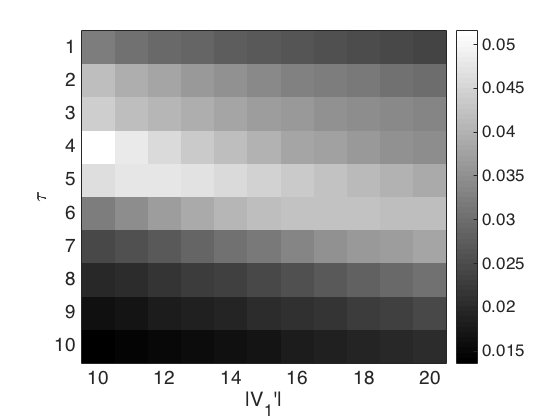}
    \vspace{-20pt}
    \caption{\small{Inferred prior distribution of GEDs on Fingerprint data set}}\label{fig-prior-ged}
    \endminipage\hfill  
    \vspace{-20pt}
\end{figure}

To calculate the prior distribution of GBDs, we first randomly sample $\alpha\%$ of graph pairs from the database $D$, and calculate the GBD between each pair of sampled graphs. Then, we utilize the \emph{Gaussian Mixture Model} (GMM) \cite{day1969estimating} to approximate the distribution of GBDs between all pairs of sampled graphs, whose probability density function is:
\begingroup
\setlength{\abovedisplayskip}{2pt}
\setlength{\belowdisplayskip}{2pt}
\setlength{\abovedisplayshortskip}{2pt}
\setlength{\belowdisplayshortskip}{2pt}
\begin{align}\label{equal-prior-gbd-pdf}
\textstyle f(\phi)=\sum_{i=1}^{K}\pi_{i}\cdot \mathcal{N}(\phi; \mu_i, \sigma_i)
\end{align}
\endgroup
where $K$ is the number of components in the mixture model defined by user, and
$\mathcal{N}$ is the probability density function of the normal distribution. Here, $\pi_{i}$, $\mu_i$ and $\sigma_i$ are parameters of the $i$-th component in the mixture model, which can be inferred from the GBDs over sampled graph pairs. Please refer to \cite{day1969estimating} for the process of inferring parameters in GMM.

Finally, we can compute the prior probability $Pr[{GBD}=\varphi]$ by integrating the probability density function $f(\phi)$ on the adjacent interval of $\varphi$, i.e., $[\varphi-0.5,\varphi+0.5]$.
\begingroup
\setlength{\abovedisplayskip}{2pt}
\setlength{\belowdisplayskip}{2pt}
\setlength{\abovedisplayshortskip}{2pt}
\setlength{\belowdisplayshortskip}{2pt}
\begin{align}\label{equal-prior-gbd}
\textstyle Pr[{GBD}=\varphi]=        \int_{\varphi-0.5}^{\varphi+0.5}\sum_{i=1}^{K}\pi_{i}\cdot \mathcal{N}(\phi; \mu_i, \sigma_i) \text{ } d\phi
\end{align}
\endgroup

Note that this integration technique is commonly used in the field called continuity correction \cite{wiki:contiunecorrection}, which aims to approximate discrete distributions by continuous distributions. In addition, the choice of integral interval $[\varphi-0.5,\varphi+0.5]$ is a common practice in the continuity correction field \cite{degroot2013probability}. 

\begin{example}\label{example-prior-gbd}
    We randomly sample 60,000 pairs of graphs from Fingerprint dataset of IAM Graph Database\cite{riesen2008iam}, where the distribution of GBDs between all pairs of sampled graphs is represented by the blue histogram in Figure \ref{fig-prior-gbd}. We infer the GBD prior distribution of sampled graph pairs, which is represented by the red line in Figure \ref{fig-prior-gbd}. This way, we can
    compute and store the probability $Pr[GBD=\varphi]$ for each possible value of variable $\varphi$ by Equation (\ref{equal-prior-gbd}), where the range of variable $\varphi$ is analyzed in Section \ref{sec-search-framework}.\vspace{-2ex}
\end{example}

\subsection{Calculating $\mathbf{\Lambda_3}$} \label{section-lambda-3}
\vspace{-0.4em}

In this subsection, we focus on calculating $\Lambda_3$, i.e., the prior distribution of GEDs.
Recall that the GED computation is NP-hard \cite{zeng2009comparing}. Therefore, sampling some graph pairs and calculating the GED between each pair is infeasible especially for large graphs, which means that we cannot simply infer the prior distribution of GEDs from sampled graph pairs.

To address this problem, we utilize the Jeffreys prior \cite{jeffreys1946invariant} as the prior distribution of GEDs. The Jeffreys prior is well-known for its non-informative property \cite{wiki:jeffreys}, which means that it provides very little additional information to the probabilistic model.
Therefore, Jeffreys prior is a common choice in Bayesian methods when we know little about the actual prior distribution. Then, according to the definition of Jeffreys prior \cite{jeffreys1946invariant}, the value of probability $Pr[GED=\tau]$ is calculated by:

\begingroup
\setlength{\abovedisplayskip}{-10pt}
\setlength{\belowdisplayskip}{2pt}
\setlength{\abovedisplayshortskip}{2pt}
\setlength{\belowdisplayshortskip}{2pt}
\begin{align}
& \hspace*{-1em} Pr[GED=\tau] \nonumber \\
\label{equal-prior-ged-jeff-def}&\hspace*{-2.3em} \textstyle = \frac{1}{C}  \sqrt{\sum_{\varphi=0}^{2\tau} Pr[GBD=\varphi|GED=\tau] \cdot \mathcal{Z}^2(G_1', G_2'; \tau, \varphi)} \\
\label{equal-prior-ged-jeff-final}&\hspace*{-2.3em} \textstyle = \frac{1}{C}  \sqrt{\sum_{\varphi=0}^{2\tau} \Lambda_1(G_1', G_2'; \tau, \varphi) \cdot \mathcal{Z}^2(G_1', G_2'; \tau, \varphi)},
\end{align}
\endgroup
where Equation (\ref{equal-prior-ged-jeff-def}) comes from the definition of the Jeffreys prior \cite{jeffreys1946invariant}, which is the expected value of function $\mathcal{Z}^2$ with respect to the conditional distribution of variable $GBD$ when given $GED=\tau$. Here, $C$ is a constant for normalization, and the function $\mathcal{Z}$ is defined in the following Equation (\ref{equal-prior-ged-zfunc}).

\begingroup
\setlength{\abovedisplayskip}{-6pt}
\setlength{\belowdisplayskip}{2pt}
\setlength{\abovedisplayshortskip}{2pt}
\setlength{\belowdisplayshortskip}{2pt}

\begin{align}\label{equal-prior-ged-zfunc}
\hspace*{-2.6em} \mathcal{Z}(G_1', G_2'; \tau, \varphi) = \left. \frac{\partial \{\log{Pr[GBD|GED]}\}}{\partial\{GED\}} \right|_{GED=\tau, GBD=\varphi}
\end{align} 
\endgroup
where the symbol $\partial$ represents the partial derivative of a function, and the symbol $F(x)|_{x=k}$ means to substitute value $k$ for variable $x$ in the function $F(x)$. Please refer to Appendix \ref{Appendix-Closed-Forms-of-Equations} for the closed forms of Equation (\ref{equal-prior-ged-jeff-final}).

According to the closed form of Equation (\ref{equal-prior-ged-jeff-final}), we find that the value of probability $Pr[{GED}=\tau]$ only depends on the constant $\tau$ and the size of the extended graph $G'_1$, i.e., ($|V_1'|$).
Therefore, for each data set, we pre-compute the value of function 
$Pr[{GED}=\tau]$ for each possible value of $\tau$ and $|V_1'|$, and store these pre-computed values in a matrix for quick looking up when searching similar graphs. In Section \ref{sec-search-framework}, we discuss the ranges of $\tau$ and $|V_1'|$ in detail. Note that, if there are $k_1$ possible values of $\tau$ and $k_2$ possible values of $|V_1'|$, then the normalization constant $C$ in Equations (\ref{equal-prior-ged-jeff-def}) and (\ref{equal-prior-ged-jeff-final}) will be $C=1/(k_1\cdot k_2)$.

\begin{example}\label{example-prior-ged}
    We show part of the Jeffreys prior distribution of GEDs on Fingerprint data set \cite{riesen2008iam} in Figure \ref{fig-prior-ged} as an example, where the x-axis represents values of $|V_1'|$,
    and the y-axis represents values of $\tau$. The gray scale of each $1\times 1$ block in Figure \ref{fig-prior-ged} represents the corresponding value of $Pr[{GED}=\tau]$.
\end{example}

\newlength{\textfloatsepsave} 
\setlength{\textfloatsepsave}{\textfloatsep} 
\setlength{\textfloatsep}{0pt}
\begin{algorithm}[!t]
    \caption{Graph Similarity Search with Graph Branch Distance Approximation (GBDA)}\label{algo-similarity-search}
    \begin{algorithmic}
        \Require a query graph $Q$, a graph database $D$,\\ 
        a similarity threshold $\hat{\tau}$, and a probability threshold $\gamma$
        \Ensure the search result $D_0$
        \For{\textbf{each} graph $G \in D$}
        \State{\textbf{Step 1}*: Pre-compute $\Lambda_2$ and $\Lambda_3$ for all possible inputs}
        \State{\textbf{Step 2}: Calculate $GBD(Q',G')$ by Definition \ref{def-gbd}.}
        \State{\textbf{Step 3}: Given $GBD(Q',G') = \varphi$, we calculate}
        \begingroup
        \setlength{\abovedisplayskip}{1pt}
        \setlength{\belowdisplayskip}{1pt}
        \setlength{\abovedisplayshortskip}{1pt}
        \setlength{\belowdisplayshortskip}{1pt}
        \begin{align*}
        \Phi=&Pr[GED(Q, G) \le \hat{\tau} \mid GBD(Q, G)=\varphi] \nonumber \\
        =&\sum_{\tau=0}^{\hat{\tau}}\Lambda_1(Q',G';\tau, \varphi) \cdot \frac{ \Lambda_3(Q',G';\tau)}{\Lambda_2(Q',G';\varphi)}
        \end{align*}
        \endgroup
        \State{where $\Lambda_1(Q',G';\tau, \varphi)$ is calculated by Equation (\ref{equal-lambda-1-expansion}).}
        \State{\textbf{Step 4}: Insert $G$ into $D_0$ if $\Phi \ge \gamma$}
        \EndFor
    \end{algorithmic}
\end{algorithm}

\vspace{-0.4em}
\section{Graph Similarity Search with the Probabilistic Model} \label{sec-search-framework}
\vspace{-0.4em}

In this section, we first elaborate on our graph similarity search algorithm (i.e., GBDA)
based on the model derived in the previous section, which consists of 
two stages: offline pre-processing and online querying. Then, we study the
time and space complexity of these two stages in detail.

\vspace{-0.8em}
\subsection{Graph Similarity Search Algorithm}
\vspace{-0.4em}

Given a query graph $Q$, a graph database $D$, a similarity threshold $\hat{\tau}$, and a probability threshold $\gamma$, the search results $D_0$
is achieved by Algorithm \ref{algo-similarity-search}, where $Q'$ and $G'$ are the extended graphs of 
$Q$ and $G$, respectively. Step 1 tagged with symbol * is
pre-computed in the offline stage, as discussed in Sections \ref{section-lambda-2} and \ref{section-lambda-3}. We give the Example \ref{example-algo-1} below to better illustrate the process of Algorithm 1.

\begin{example}\label{example-algo-1}
    Assume that the graph $G_1$ in Figure \ref{fig-intro-ged-example} is the query graph $Q$, and $G_2$ in Figure \ref{fig-intro-ged-example} is a graph in database $D$. Given the similarity threshold $\hat{\tau}=3$ and the probability threshold $\gamma=0.8$, the process of determining whether $G_2$ should be in the search result $D_0$ is as follows.
    
    \begin{enumerate}[leftmargin=*]
        \item First, we pre-compute $\Lambda_2(Q',G_2'; \varphi)$ and $\Lambda_3(Q',G_2'; \tau)$ by inferring the prior distributions of GBDs and GEDs on database $D$, respectively. Since this is a simulated example and there is no concrete database $D$, we assume that $\Lambda_3(Q',G_2'; \tau) / \Lambda_2(Q',G_2';\varphi) \equiv 0.8$ for all possible values of $\tau$ and $\varphi$. 
        \item Second, from Example \ref{example-gbd-computation}, we know $GBD(Q, G_2) = 3$.
        \item Then, according to Equation (\ref{equal-lambda-1-expansion}), we can calculate
        \begingroup
        \setlength{\abovedisplayskip}{1pt}
        \setlength{\belowdisplayskip}{1pt}
        \setlength{\abovedisplayshortskip}{1pt}
        \setlength{\belowdisplayshortskip}{1pt}
        \begin{align*}
        \Phi &\textstyle = \sum_{\tau=0}^{\hat{\tau}}\Lambda_1(Q',G_2';\tau, \varphi) \cdot \frac{ \Lambda_3(Q',G_2';\tau)}{\Lambda_2(Q',G_2';\varphi)} \\
        & = (0 + 0 + 0.5113 + 0.5631) \times 0.8 = 0.8595 > \gamma = 0.8
        \end{align*}
        \endgroup
        \item Therefore, $G_2$ is inserted into the search result $D_0$.
    \end{enumerate}
\end{example}

\vspace{-0.8em}
\subsection{Complexity Analysis of Online Stage} \label{sec-complexity-online}
\vspace{-0.4em}

The online querying stage in our approach includes Steps 2, 3 and 4 in Algorithm \ref{algo-similarity-search}, where Step 4 clearly costs $O(1)$ time for each graph $G$. In addition, from the discussions in Section \ref{sec-branch-distance-definition}, Step 2 costs $O(nd)$ time, where $n=\max\{|V_G|,|V_Q|\}$, and
$d$ is the average degree of graph $G$.

In Step 3, since $\Lambda_2$ and $\Lambda_3$ have already been pre-computed in the offline stage (i.e., Step 1), their values can be obtained in $O(1)$ time for each $\tau \in [0,\hat{\tau}]$ and graph $G \in D$.

Now we focus on analyzing the time complexity of computing $\Lambda_1$ in Step 3 of Algorithm 1.
Let the time for computing $\Omega_1$, $\Omega_2$, $\Omega_3$ and $\Omega_4$ in Equation (\ref{equal-lambda-1-expansion})
be $C_1, C_2, C_3$ and $C_4$, respectively. 
According to Equation (\ref{equal-lambda-1-expansion}), the time for computing $\Lambda_1$ for each $\tau \in [0,\hat{\tau}]$ and graph $G \in D$ is:
\begingroup
\setlength{\abovedisplayskip}{0pt}
\setlength{\belowdisplayskip}{0pt}
\setlength{\abovedisplayshortskip}{0pt}
\setlength{\belowdisplayshortskip}{0pt}
\begin{align} \label{equal-online-complexity-lambda-1-orgin}
\overbrace{xC_1}^{\sum \Omega_1} + \overbrace{xmC_2}^{\sum \Omega_2} + \overbrace{xmr(C_3\cdot C_4)}^{\sum \Omega_3 \cdot \Omega_4}
\end{align}
\endgroup
where $x,m$ and $r$ are the summation subscripts in Equation (\ref{equal-lambda-1-expansion}), and
the ranges of $x,m$ and $r$ are:

\begin{itemize}[leftmargin=*]
    \item $x\in [0, \tau]$. Since $x$ is the number of \textbf{RV} operations, $x$ must not be larger than the number of graph edit operations $\tau$.
    \item $m\in [0, 2\tau]$. Since $m$ is the number of vertices covered by relabelled edges given the relabelled edge number $Y = \tau - x$,
    and each edge can cover at most two vertices, we have $0 \le m \le 2(\tau - x) \le 2\tau$.
    \item $r\in [0, 3\tau]$. Note that $r$ is the number of vertices either relabelled or covered by relabelled edges when
    the relabelled vertex number is $X=x$ and the number of vertices covered by relabelled edges is $Z=m$. Therefore, we have
    $0 \le r \le x + m \le \tau + 2\tau = 3\tau$. 
\end{itemize}

In addition, according to the closed form of Equation (\ref{equal-lambda-1-expansion}) in Appendix \ref{Appendix-Closed-Forms-of-Equations}, 
it is clear that $C_1 = C_3 = C_4 = O(1)$ and $C_2 = O(m) = O(\tau)$. According to Equation (\ref{equal-online-complexity-lambda-1-orgin}), the time of computing $\Lambda_1$ for each $\tau \in [0,\hat{\tau}]$ and graph $G \in D$ is:
\begingroup
\setlength{\abovedisplayskip}{0pt}
\setlength{\belowdisplayskip}{0pt}
\setlength{\abovedisplayshortskip}{0pt}
\setlength{\belowdisplayshortskip}{0pt}
\begin{align}
\overbrace{O(\tau)\vphantom{x^3}}^{\sum \Omega_1} + \overbrace{O(\tau^3)}^{\sum \Omega_2} + \overbrace{O(\tau^3)}^{\sum\Omega_3\cdot \Omega_4} = O(\tau^3)
\end{align}
\endgroup

Moreover, from the above discussions about the ranges of summation subscripts, for any $\tau \in [0, \hat{\tau})$, we have:

\begingroup
\setlength{\abovedisplayskip}{-10pt}
\setlength{\belowdisplayskip}{0pt}
\setlength{\abovedisplayshortskip}{0pt}
\setlength{\belowdisplayshortskip}{0pt}
\begin{align}
&\textstyle\sum_{x=0}^{\hat{\tau}}\sum_{m=0}^{2\hat{\tau}}\Omega_2(m,x,\hat{\tau}) \nonumber \\=&\textstyle\sum_{x=0}^{\hat{\tau}}\sum_{m=0}^{2\hat{\tau}}Pr[Z=m\mid Y =\hat{\tau}-x] \label{equal-omega2-before-4terms} \\
=&\textstyle \sum_{x=0}^{\tau}\sum_{m=0}^{2\tau}Pr[Z=m\mid Y =\tau-x] \nonumber \\
&\textstyle +\sum_{x=0}^{\hat{\tau}-\tau-1}\sum_{m=0}^{2\tau}Pr[Z=m\mid Y =\hat\tau-x] \nonumber \\
&\textstyle +\sum_{x=0}^{\tau}\sum_{m=2\tau+1}^{2\hat{\tau}}Pr[Z=m\mid Y =\hat\tau-x] \nonumber \\
&\textstyle +\sum_{x=0}^{\hat{\tau}-\tau-1}\sum_{m=2\tau+1}^{2\hat{\tau}}Pr[Z=m\mid Y =\hat\tau-x] \label{equal-omega2-4terms} \\
=&\textstyle f(m,x,\hat{\tau}) + \sum_{x=0}^{{\tau}}\sum_{m=0}^{2{\tau}}\Omega_2(m,x,{\tau})
\label{equal-omega2-reduce}
\end{align}
\endgroup
where $f(m,x,\hat{\tau})$ is sum of last three terms in Equation (\ref{equal-omega2-4terms}). Note that Equation (\ref{equal-omega2-4terms}) is a sum on four disjoint two-dimensional intervals whose combination is the sum interval of Equation (\ref{equal-omega2-before-4terms}).

Equation (\ref{equal-omega2-reduce}) means, the value of $\sum_{x,m}\Omega_2(m,x,\tau)$ where $\tau < \hat{\tau}$ have already
been calculated in the process of computing $\sum_{x,m}\Omega_2(m,x,\hat{\tau})$.
Therefore, we can reduce redundant computations by 
only computing $\sum_{x,m}\Omega_2(m,x,\hat{\tau})$ once to obtain
values of $\sum_{x,m}\Omega_2(m,x,\tau)$ for all $\tau < \hat{\tau}$.
Similar conclusions can also be derived for $\sum_{x,m,r}\Omega_3(r,\varphi)\cdot \Omega_4(x,r,m)$, where the detailed proofs are omitted here.

Therefore, the time cost of Step 3 in Algorithm \ref{algo-similarity-search} is:
\begingroup
\setlength{\abovedisplayskip}{0pt}
\setlength{\belowdisplayskip}{0pt}
\setlength{\abovedisplayshortskip}{0pt}
\setlength{\belowdisplayshortskip}{0pt}
\begin{align}
\textstyle \overbrace{O(\hat{\tau}^3)}^{\sum\Omega_2\sum\Omega_3\cdot\Omega_4} + \sum_{\tau=0}^{\hat{\tau}}\{\overbrace{O(\tau)\vphantom{\hat{\tau}^2}}^{\sum\Omega_1}
+\overbrace{O(1)\vphantom{\hat{\tau}^2}}^{\Lambda_2}\}=O(\hat{\tau}^3)
\end{align}
\endgroup
for each graph $G$ in database $D$.

Finally, we can obtain Theorem \ref{theorem-online-time-complexity}.

\begin{theorem} \label{theorem-online-time-complexity}
    The time of the whole online stage is:
    \begingroup
    \setlength{\abovedisplayskip}{0pt}
    \setlength{\belowdisplayskip}{0pt}
    \setlength{\abovedisplayshortskip}{0pt}
    \setlength{\belowdisplayshortskip}{0pt}
    \begin{align}
    \overbrace{O(nd\vphantom{d^2})}^{\text{Step 2}} + \overbrace{O(\hat{\tau}^3)}^{\text{Step 3}} + \overbrace{O(1\vphantom{d^2})}^{\text{Step 4}} = O(nd + \hat{\tau}^3)
    \end{align}
    \endgroup
    for each graph $G$ in database $D$, where $n=\max\{|V_G|,|V_Q|\}$,
    $d$ is the average degree of graph $G$, and $\hat\tau$ is the similarity threshold in the graph similarity search problem.
\end{theorem}

\begin{proof}
    Please refer to the discussions above.
\end{proof}

Note that the similarity threshold $\hat{\tau}$ is often 
set as a small value (i.e., $\hat{\tau} \le 10$) and does not increase with the number of vertices $n$ in previous studies \cite{zeng2009comparing}
\cite{Zheng:2013hh}, thus, we can assume that $\hat{\tau}$ is a constant with regard to  $n$ when the graph is sufficiently large.
Moreover, most real-world graphs studied in related works \cite{riesen2009approximate} \cite{riesen2015approximate} 
are \emph{scale-free graphs} \cite{wiki:scale-free-graphs}, whose average degrees $d=O(\log{n})$ as proved in Appendix \ref{append-theorem-scale-free-degree}.

\vspace{-0.6em}
\subsection{Complexity Analysis of Offline Stage} \label{sec-complexity-offline}
\vspace{-0.3em}
The offline pre-processing stage in our approach is Step 1 in Algorithm \ref{algo-similarity-search}, which is essentially to pre-compute the prior distributions of GEDs and GBDs respectively among all pairs of graphs involved in the graph similarity search.

\subsubsection{Complexity Analysis of Computing the Prior Distribution of GBDs}

As discussed in Section 5.1, the prior distributions of GBDs can be pre-computed by the following four steps:
\begin{description}[leftmargin=0pt]
    \item[Step 1.1:] Sample $N$ graph pairs from the database $D$.
    \item[Step 1.2:] Calculate GBD between each sampled graph pairs. 
    \item[Step 1.3:] Learn the Gaussian Mixture Model (GMM) of the GBDs between sampled graph pairs.
    \item[Step 1.4:] Calculate $Pr[GBD=\varphi]$ for each possible value of $\varphi$ by using Equation (\ref{equal-prior-gbd}).
\end{description}

It is clear that Step 1.1 costs $O(N)$ time. From the discussions in Section \ref{sec-branch-distance-definition}, Step 1.2 costs $O(N\cdot nd)$ time, where
$n$ is the maximal number of vertices among the sampled graphs, and $d$ is the average degree of the sampled graphs. The learning process of GMM in Step 1.3 costs $O(N\cdot K\epsilon)$ time\cite{day1969estimating}, where $K$ is the number of components in GMM, and $\epsilon$ is the maximal learning iterations for learning GMM.

As for Step 1.4, since $\varphi$ is the value of GBD, from the definition of GBD, the possible values of $\varphi$ are essentially $\{0,1,2...,n\}$, where $n$ is the maximal number of vertices among the sampled graphs. According to Equation (\ref{equal-prior-gbd}), computing $Pr[GBD=\varphi]$ for each $\varphi$ costs $O(K)$ time, where $K$ is the number of components in GMM derived in Step 1.3. Thus, Step 1.4 costs $O(nK)$ time.

Note that in the Gaussian Mixture Model, the component number $K$ and the maximal learning iterations $\epsilon$ are fixed constants. Therefore, the prior distributions of GBDs can be calculated in time
\begingroup
\setlength{\abovedisplayskip}{0pt}
\setlength{\belowdisplayskip}{0pt}
\setlength{\abovedisplayshortskip}{0pt}
\setlength{\belowdisplayshortskip}{0pt}
\begin{align}
\overbrace{O(N)}^{\text{Step 1.1}} + \overbrace{O(Nnd)}^{\text{Step 1.2}} + \overbrace{O(NK\epsilon)}^{\text{Step 1.3}} + \overbrace{O(nK)}^{\text{Step 1.4}} = O(Nnd)
\end{align}
\endgroup
where $N$ is the number of graph pairs sampled in Step 1.1,
$n$ is the maximal number of vertices among sampled graphs, and $d$ is the average degree of sampled graphs.

Based on the discussions above, the number of pre-computed values of $Pr[GBD=\varphi]$ is at most $n$. Thus, the space cost of storing the prior distribution of GBDs is $O(n)$.

\subsubsection{Complexity Analysis of Computing the Prior Distribution of GEDs}

According to the discussions in Section \ref{section-lambda-3}, computing the prior distribution of GEDs is essentially calculating Equation (\ref{equal-prior-ged-jeff-final}) for each possible values of $\tau$ and $|V_1'|$.

First, from the closed form of Equation (\ref{equal-prior-ged-jeff-final}) in Appendix \ref{Appendix-Closed-Forms-of-Equations}, when $\tau$ and $|V_1'|$ are fixed values, it is clear that computing $\frac{d}{d\tau} \log{\Lambda_1}$ costs the same time as computing ${\Lambda_1}$, which is $O({\hat{\tau}^3})$, where $\hat{\tau}$ is the user-defined similarity threshold. $n$ and $d$ are the maximal number of vertices and the average degree among all graphs, respectively. 

Then, since one graph edit operation can at most change two branches,  when the GED between two graphs is $\tau$, 
the possible GBD values (i.e., $\varphi$ in Equation (\ref{equal-prior-ged-jeff-final})) between these two graphs are $\{0,1,2...,2{\tau}\}$. Therefore, according to Equation (\ref{equal-prior-ged-jeff-final}), the prior probability value of $Pr[GED=\tau]$ can be calculated in time complexity $O(2{\tau}\cdot {{\tau}^3}) = O({\tau}^4)$ when $\tau$ and $|V_1'|$ are fixed values.

Finally, recall that computing the GED prior distribution is essentially calculating Equation (\ref{equal-prior-ged-jeff-final}) for all possible values of $\tau$ and $|V_1'|$, and it is clear that the possible values of $\tau$ are $\{0,1,2...,\hat{\tau}\}$, where $\hat{\tau}$ is the user-defined similarity threshold.
In addition, the possible values of $|V_1'|$ are essentially $\{1,2,...n\}$,  where $n$ is the maximal number of vertices among all graphs involved the graph similarity search. Therefore, the time of calculating the GED prior distribution is:
\begingroup
\setlength{\abovedisplayskip}{2pt}
\setlength{\belowdisplayskip}{2pt}
\setlength{\abovedisplayshortskip}{2pt}
\setlength{\belowdisplayshortskip}{2pt}
\begin{align*}
O(\hat{\tau}\cdot n \cdot {\hat{\tau}^4}) =O(n\hat{\tau}^5)
\end{align*}
\endgroup

According to Section \ref{section-lambda-3}, 
we need to store a matrix whose rows represent possible values of $\tau$,
and columns represent possible values of $|V_1'|$. Therefore, the space cost of storing the prior distribution of GEDs is $ O(\hat{\tau}\cdot n )$.

Finally, we have Theorem \ref{theorem-offline-complexity}.

\begin{theorem} \label{theorem-offline-complexity}
    The time complexity of the offline stage is:
    \begingroup
    \setlength{\abovedisplayskip}{2pt}
    \setlength{\belowdisplayskip}{2pt}
    \setlength{\abovedisplayshortskip}{2pt}
    \setlength{\belowdisplayshortskip}{2pt}
    \begin{align*}
    \overbrace{O(Nnd)\vphantom{\hat{\tau}^5}}^{\text{GBD Prior}}+\overbrace{O(n\hat{\tau}^5)}^{\text{GED Prior}} = O(Nnd+n\hat{\tau}^5)
    \end{align*}
    \endgroup
    and the space complexity of the offline stage is:
    \begingroup
    \setlength{\abovedisplayskip}{2pt}
    \setlength{\belowdisplayskip}{2pt}
    \setlength{\abovedisplayshortskip}{2pt}
    \setlength{\belowdisplayshortskip}{2pt}
    \begin{align*}
    \overbrace{O(n)\vphantom{\hat{\tau}}}^{\text{GBD Prior}} + \overbrace{O(\hat{\tau}\cdot n )}^{\text{GED Prior}} = O(n(1+\hat{\tau}))
    \end{align*}
    \endgroup
    where $N$ is the number of graph pairs sampled in Step 1.1,
    $n$ and $d$ are the maximal number of vertices and the average degree among all the graphs involved in the graph simialrity search, respetively. $\hat{\tau}$ is the user-defined similarity threshold.
\end{theorem}

\begin{proof}
    Please refer to the discussions above.
\end{proof}

\begin{table}[!t]
    \vspace{-25pt}
    \caption{Statistics of Data Sets}
    \vspace{-15pt}
    \label{table-stat-dataset}
    \begin{center}
        \begin{tabular}{|c|c|c|c|c|c|c|c|c|c|}
            \hline
            {Data Set} & $|D|$ & $|\mathcal{Q}|$ & $\mathcal{V}_m$ &$\mathcal{E}_m$& $d$ & Scale-free\\
            \hline
            {AIDS} & 1896 & 100 & 95 & 103 & 2.1  & Yes\\
            \hline
            {Finger} & 2159 & 114 & 26 & 26 & 1.7 & Yes\\
            \hline
            {GREC} & 1045 & 55 & 24 & 29 & 2.1 & Yes\\
            \hline
            {AASD} & 37995 & 100 & 93 & 99 & 2.1 & Yes\\
            \hline
            {Syn-1} & 3430 & 70 & 100K & 1M & 9.6 & Yes\\
            \hline
            {Syn-2} & 3430 & 70 & 100K & 1M & 9.4 & No\\
            \hline
        \end{tabular}
    \end{center}
    \small 
    \begin{minipage}{\linewidth}
        Note: $|D|$ is the number of graphs in database $D$. $|\mathcal{Q}|$ is the number of query graphs.
        $\mathcal{V}_m$ and $\mathcal{E}_m$ are the maximal numbers of vertices and edges, respectively, while $d$ is
        the average degree. $K$ means thousand and $M$ means million.
    \end{minipage}
    \vspace{2pt}
\end{table}

\vspace{-0.3em}
\section{Experiments} \label{sec-experiments}
\vspace{-0.3em}

\subsection{Data Sets and Settings}
\vspace{-0.3em}
We first present the experimental data sets for evaluating our approaches. There are 4 real-world data sets (i.e., AIDS,
Fingerprint and GREC from the IAM Graph
Database\cite{riesen2008iam}, and AIDS Antiviral Screen Data (AASD) \cite{data:aads}), and 2 synthetic data sets (i.e., Syn-1 and Syn-2).
The details about the data sets are listed in Table
\ref{table-stat-dataset}. 

The 4 real-world data sets are widely-used for evaluating the performance of GED estimation methods in previous works \cite{riesen2009approximate} \cite{riesen2015approximate} \cite{robles2005graph}. In order to evaluate how well the GED is approximated, we must know the exact value of GED, which is NP-hard to compute \cite{zeng2009comparing}. Specifically, even the state-of-the-art method \cite{DBLP:conf/icde/GoudaH16} cannot compute an exact GED for graphs with 100 vertices within 48 hours on our machine (with 32 Intel E5 2-core, 2.40 GHz CPUs  and 128GB DDR3 RAMs). 

\begin{table}[!t]
    \vspace{-25pt}
    \caption{Costs of computing GBD prior distribution}
    \vspace{-10pt}
    \label{table-offline-gbd-cost}
    \begin{tabular}{|c|c|c|c|c|c|c|c|c|c|}
        \hline
        {Data Set} & AIDS  & Finger &GREC& AASD& Syn-1 & Syn-2\\
        \hline
        Time Costs & 11.1s & 7.5s & 20.6s & 232.4s & 3.8h & 3.2h\\
        \hline
        Space Costs & 0.06kb & 0.04kb & 0.10kb & 1.21kb & 13.3gb & 0.3gb\\
        \hline
    \end{tabular}
\end{table}

\begin{table}[!t]
    \vspace{-10pt}
    \caption{Costs of computing GED prior distribution}
    \vspace{-10pt}
    \label{table-offline-ged-cost}
    \begin{tabular}{|c|c|c|c|c|c|c|c|c|c|}
        \hline
        {Data Set} & AIDS & Finger &GREC& AASD& Syn-1 & Syn-2\\
        \hline
        Time Costs & 70.32h & 16.91h & 15.40h &69.16h& 6.31h  & 6.31h\\
        \hline
        Space Costs & 1.5kb & 0.4kb & 0.4kb & 1.4kb & 0.1kb & 0.1kb\\
        \hline
    \end{tabular}
    \small 
    \begin{minipage}{0.5\textwidth}
        Note: $h$ means hours and $s$ means seconds. $kb$ means KBytes and $gb$ means GBytes.
    \end{minipage}
\end{table}

However, we still manage to evaluate our proposed method on large graphs. To address the problem above, we generate 2 sets of large random graphs (i.e., Syn-1  and Syn-2), where the GED between each pair of graphs is known. Both data sets Syn-1 and Syn-2 contain 7 subsets of graphs, where each subset contains 500 graphs whose numbers of vertices are 1K, 2K, 5K, 10K, 20K, 50K, and 100K, respectively. 
The difference between data sets Syn-1 and Syn-2 is that the graphs in \mbox{Syn-1} satisfy the \emph{scale-free} property \cite{wiki:scale-free-graphs} while graphs in \mbox{Syn-2} are not. The algorithm of generating synthetic graphs with known GEDs is described in Appendix \ref{append-graph-gen-algo}.

Note that, the scale-free property of real data sets is testified by checking whether the degree distributions of the vertices in real data sets follow the power-law distribution, while the scale-free property of Syn-1 data set and the non-scale-free property of Syn-2 data sets are ensured by our algorithm of generating synthetic graphs in Appendix \ref{append-graph-gen-algo}.

For each real data set, we randomly select 5\% graphs as query graphs, while the remaining 95\% graphs constitute the graph database $D$. For each synthetic data set, we randomly select 10 graphs from each of its subset as query graphs.

On real data sets, we evaluate our method with the similarity thresholds $\hat\tau = \{1,2,\allowbreak...,10\}$, which are commonly-used values of the similarity thresholds in previous studies \cite{zeng2009comparing} \cite{Zheng:2013hh}. On the synthetic data sets, we test our method with larger similarity thresholds $\hat\tau = \{10,11,12,\allowbreak...,30\}$ to show that, when the similarity threshold is larger than the commonly-used values (i.e., $\{1,2,\allowbreak...,10\}$), our GBDA method is still more efficient than the competitors on large graphs.

\vspace{-0.6em}
\subsection{Evaluating Offline Stage}
\vspace{-0.3em}

In this subsection, we evaluate the time and space costs of the offline stage in our GBDA approach, which is essentially to pre-compute the prior distributions of GEDs and GBDs, on both real and synthetic data sets.
Tables \ref{table-offline-gbd-cost} and \ref{table-offline-ged-cost} present the time and space costs of estimating GBD and GED prior distributions on different data sets, respectively, where the number of graph pairs sampled to estimate the GBD prior distribution is set to $N=100,000$.

The experimental results generally confirm the complexity analysis in Section \ref{sec-complexity-offline}. Specifically, since the number of sampled graph pairs $N$ and the similarity threshold $\hat{\tau}$ are fixed values in our experiments, the cost of inferring the GBD prior distribution grows with $n$ and $d$, while the cost of estimating the GED prior distribution depends on $n$, where $n$ and $d$ are the maximal number of vertices and the average degree among all the graphs involved in the graph simialrity search, respetively.

\begin{figure*}[!t]
    \vspace{-25pt}
    \minipage{0.24\textwidth}
    \includegraphics[width=\linewidth]{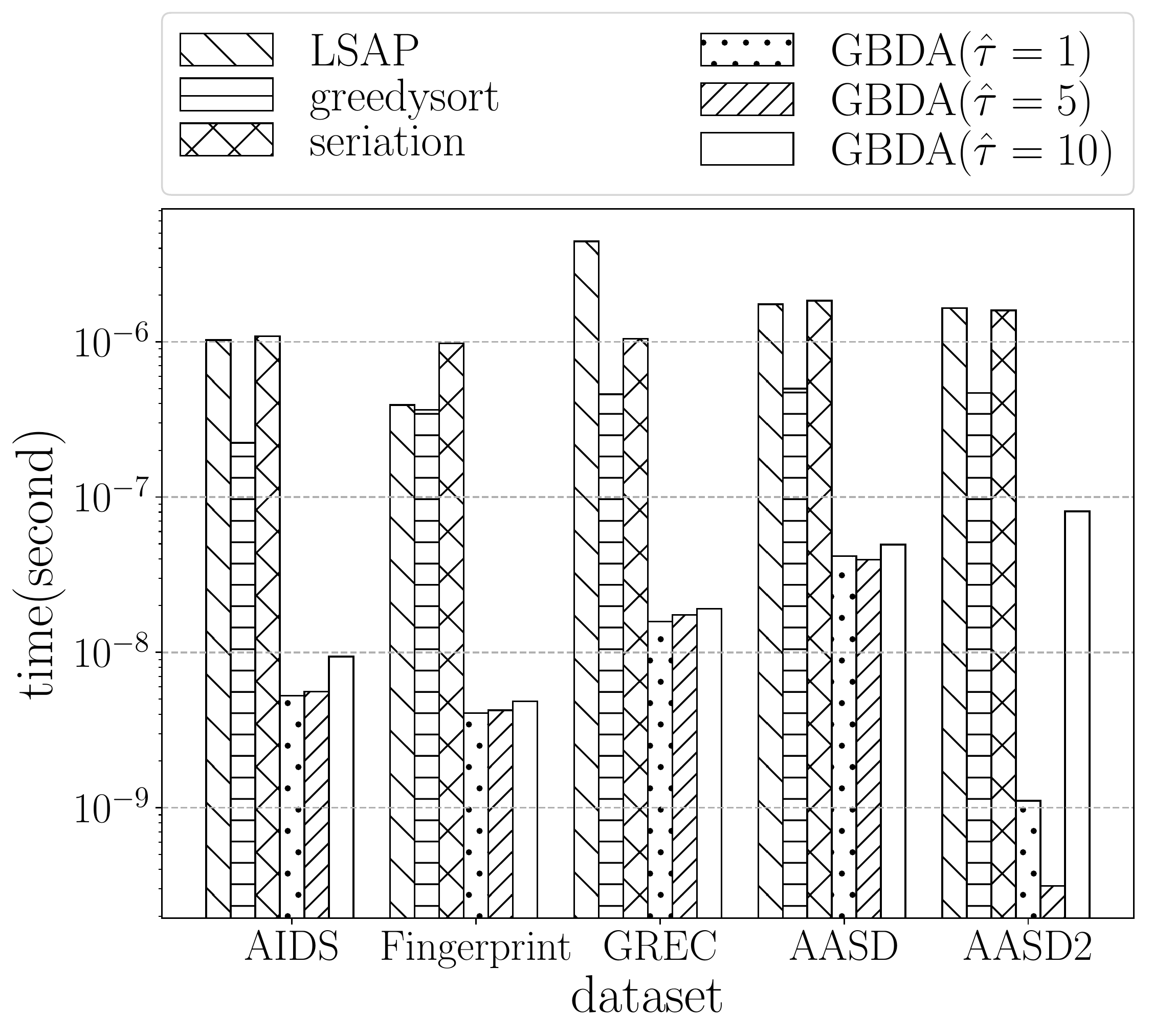}
    \vspace{-25pt}
    \caption{\footnotesize{Time Costs on Real Graphs}}\label{fig-online-time-real}
    \endminipage
    \hfill
    \minipage{0.24\textwidth}
    \includegraphics[width=\linewidth]{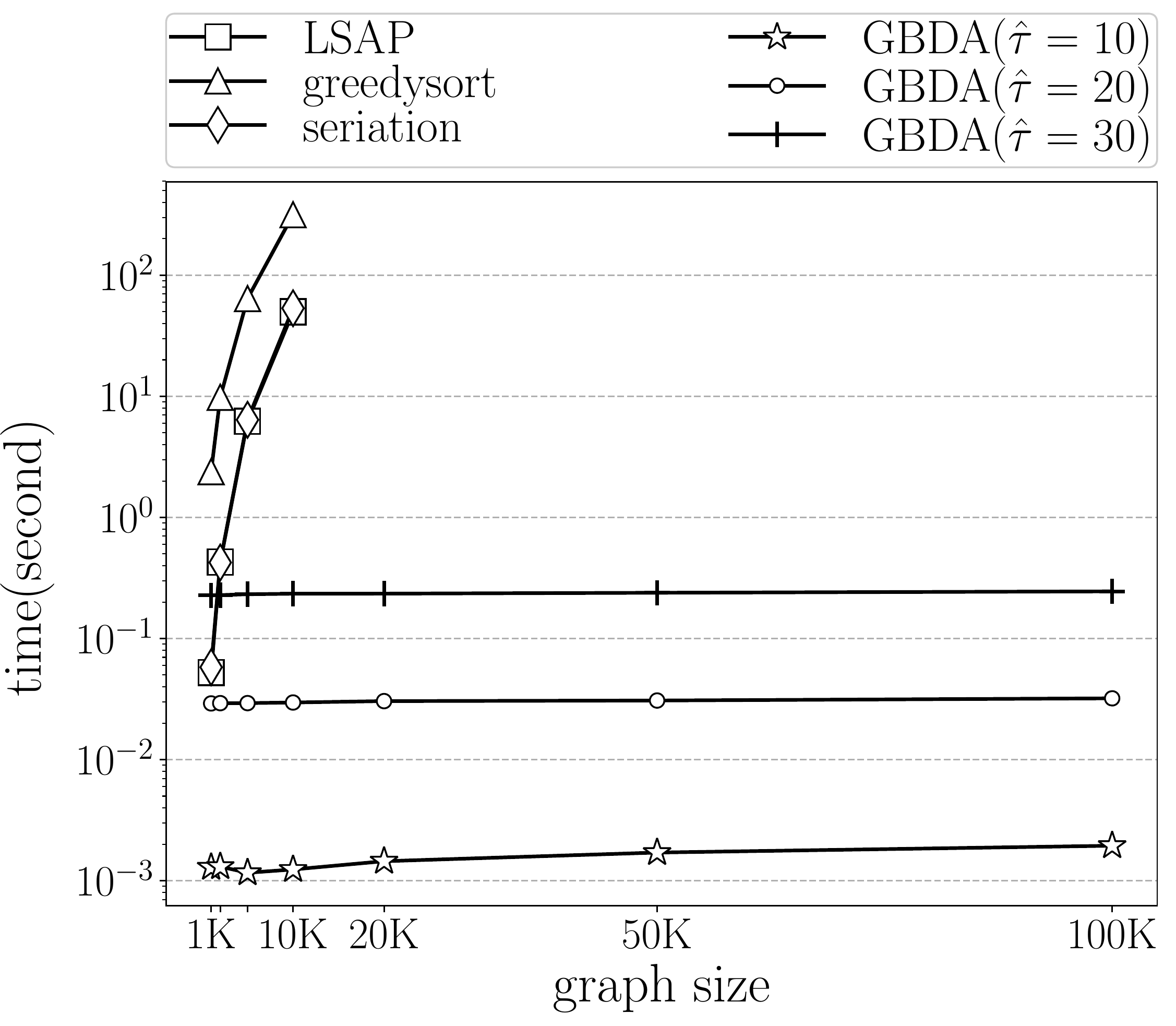}
    \vspace{-25pt}
    \caption{\footnotesize{Time Costs vs. $n$ on \mbox{Syn-1}}}\label{fig-online-time-syn1}
    \endminipage
    \hfill
    \minipage{0.24\textwidth}
    \includegraphics[width=\linewidth]{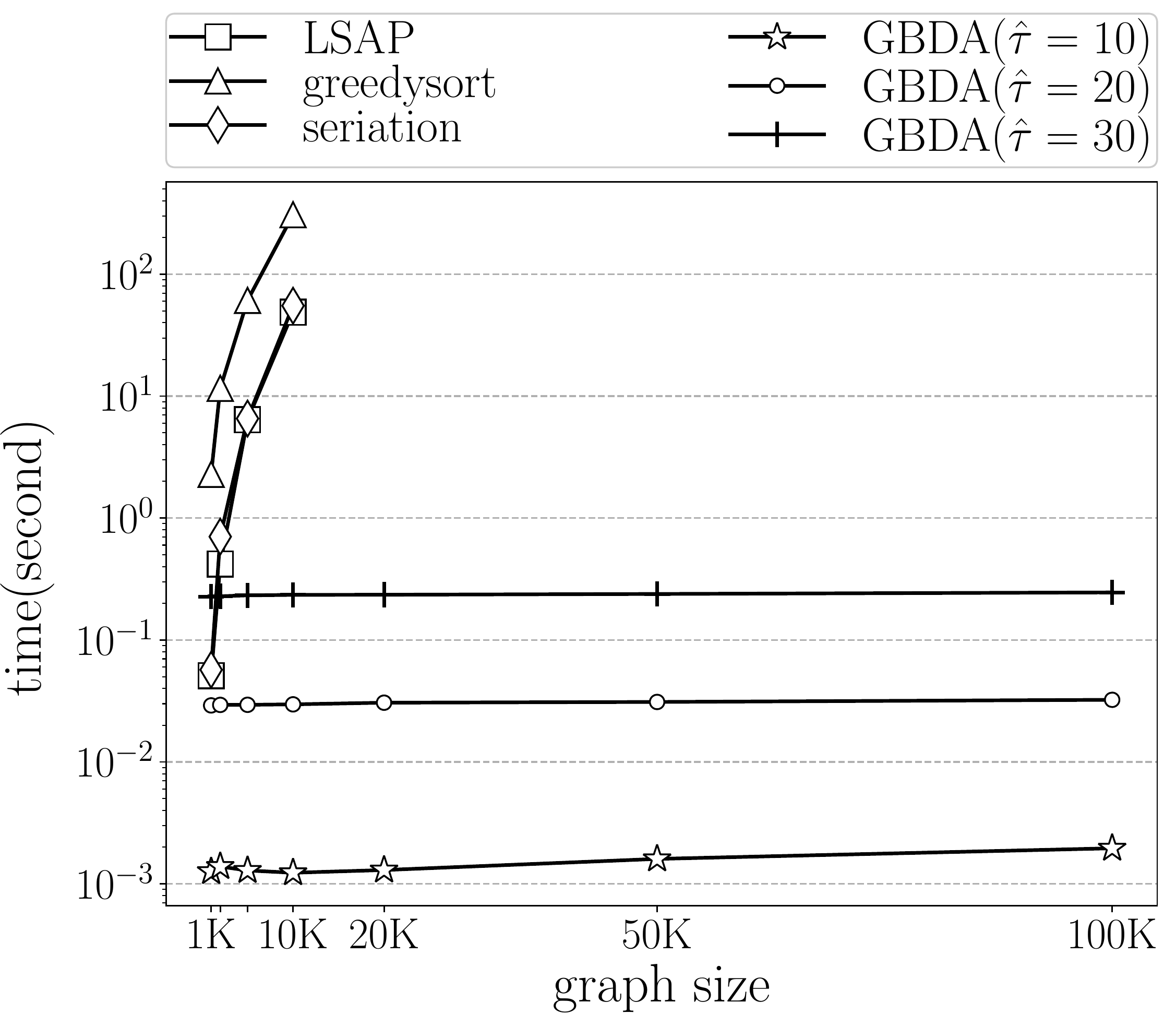}
    \vspace{-25pt}
    \caption{\footnotesize{Time Costs vs. $n$ on \mbox{Syn-2}}}\label{fig-online-time-syn2}
    \endminipage\hspace{10pt}
\end{figure*}

\begin{figure*}[!t]
    \vspace{-10pt}
    \centerfloat
    \minipage{0.24\textwidth}
    \includegraphics[width=\linewidth]{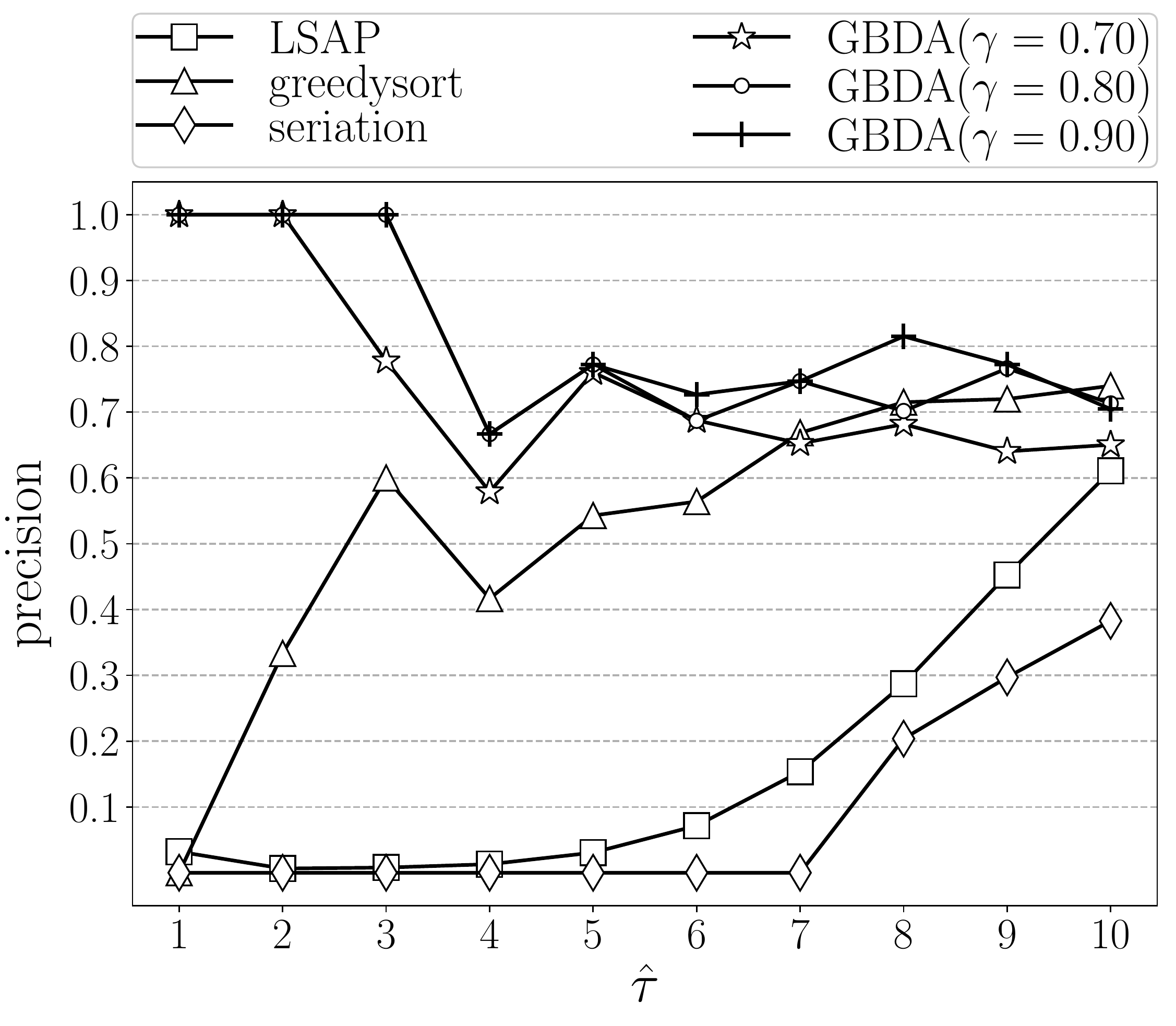}
    \vspace{-25pt}
    \caption{\footnotesize{Precision vs. $\hat\tau$ on AIDS}}\label{fig-online-acc-aids}
    \endminipage
    \hspace{10pt}
    \minipage{0.24\textwidth}
    \includegraphics[width=\linewidth]{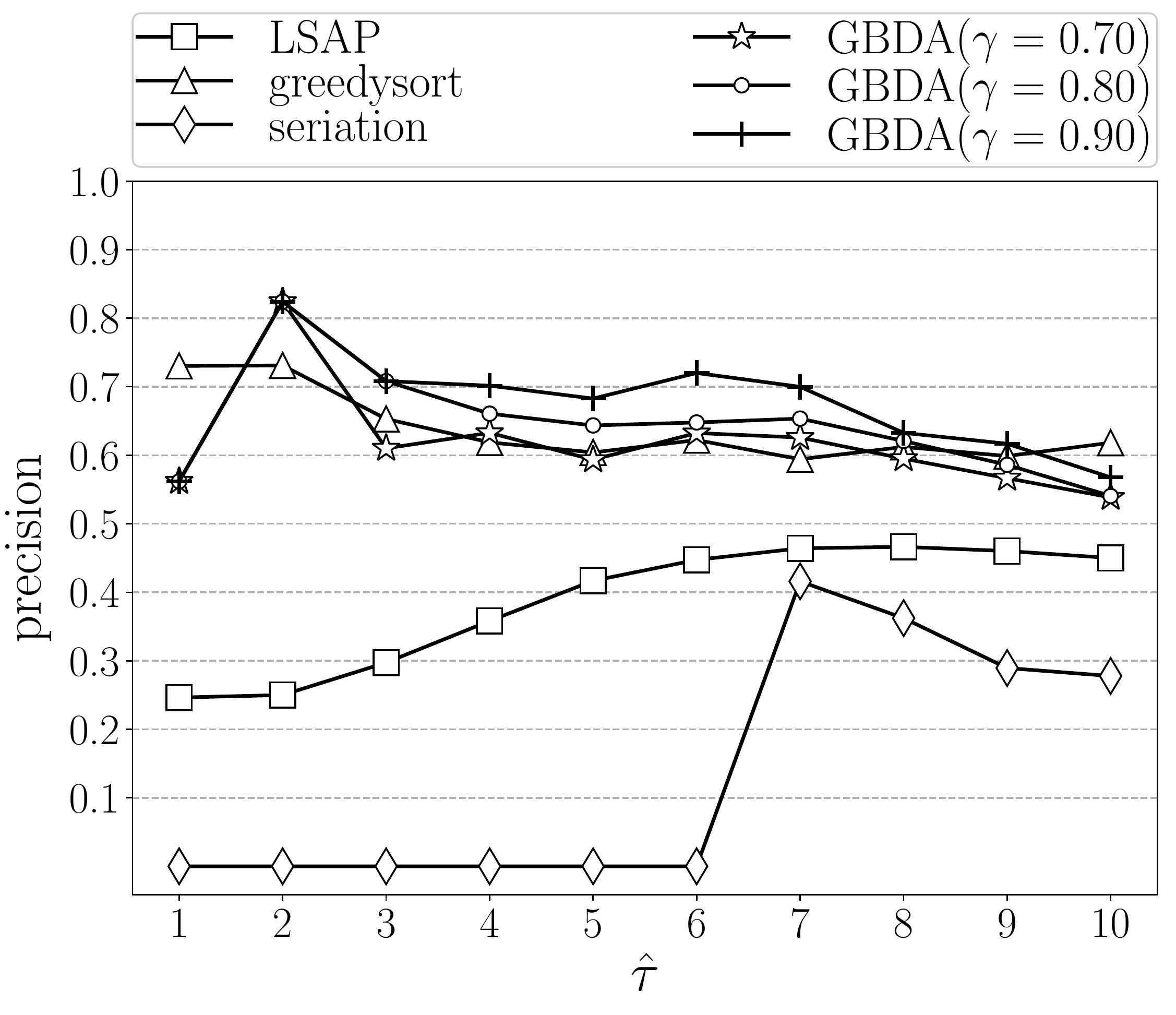}
    \vspace{-25pt}
    \caption{\hspace*{-0.8em}\footnotesize{Precision vs. $\hat\tau$ on Fingerprint}}\label{fig-online-acc-finger}
    \endminipage
    \hspace{10pt}
    \minipage{0.24\textwidth}
    \includegraphics[width=\linewidth]{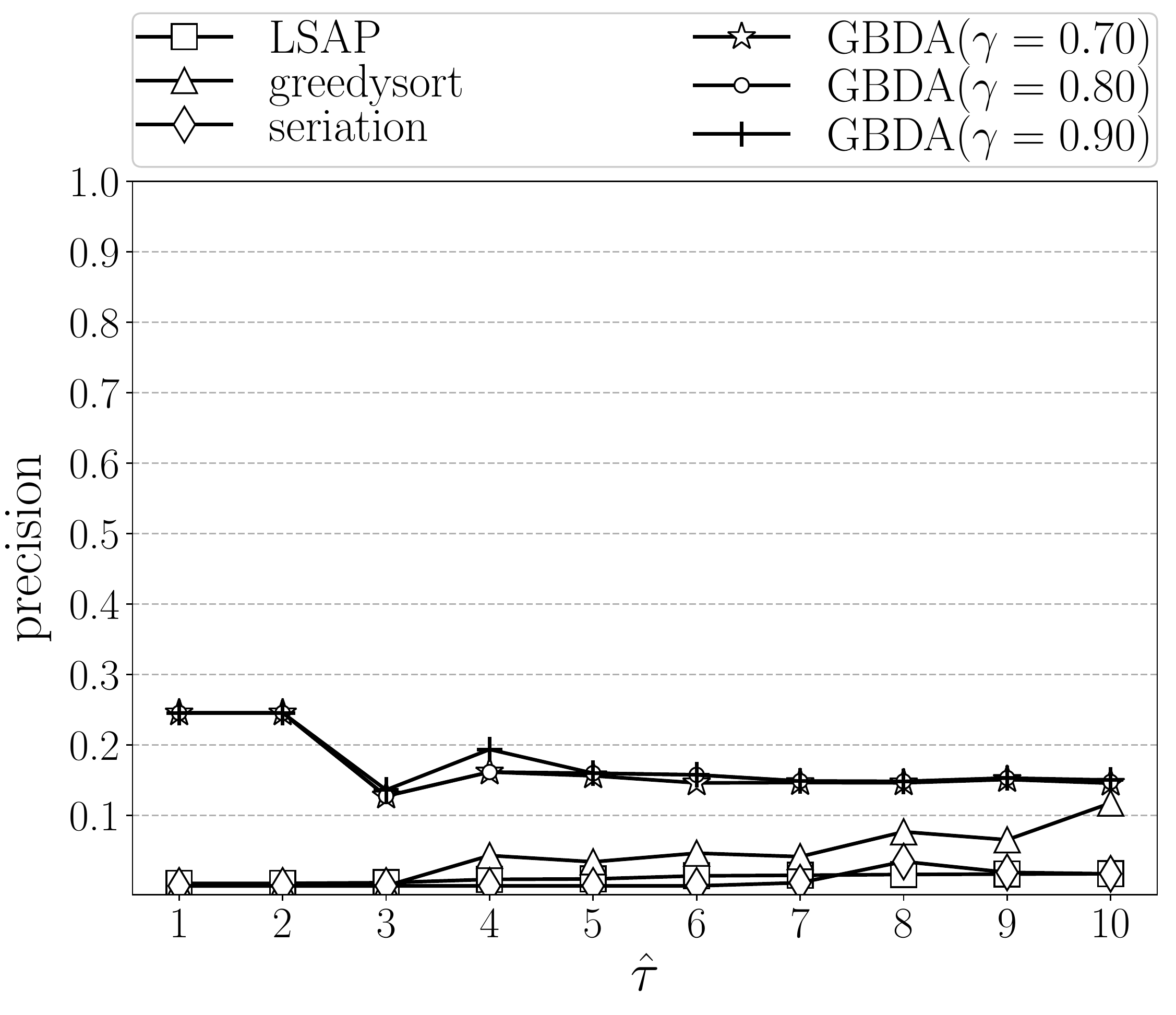}
    \vspace{-25pt}
    \caption{\footnotesize{Precision vs. $\hat\tau$ on GREC}}\label{fig-online-acc-grec}
    \endminipage
    \hspace{10pt}
    \minipage{0.24\textwidth}
    \includegraphics[width=\linewidth]{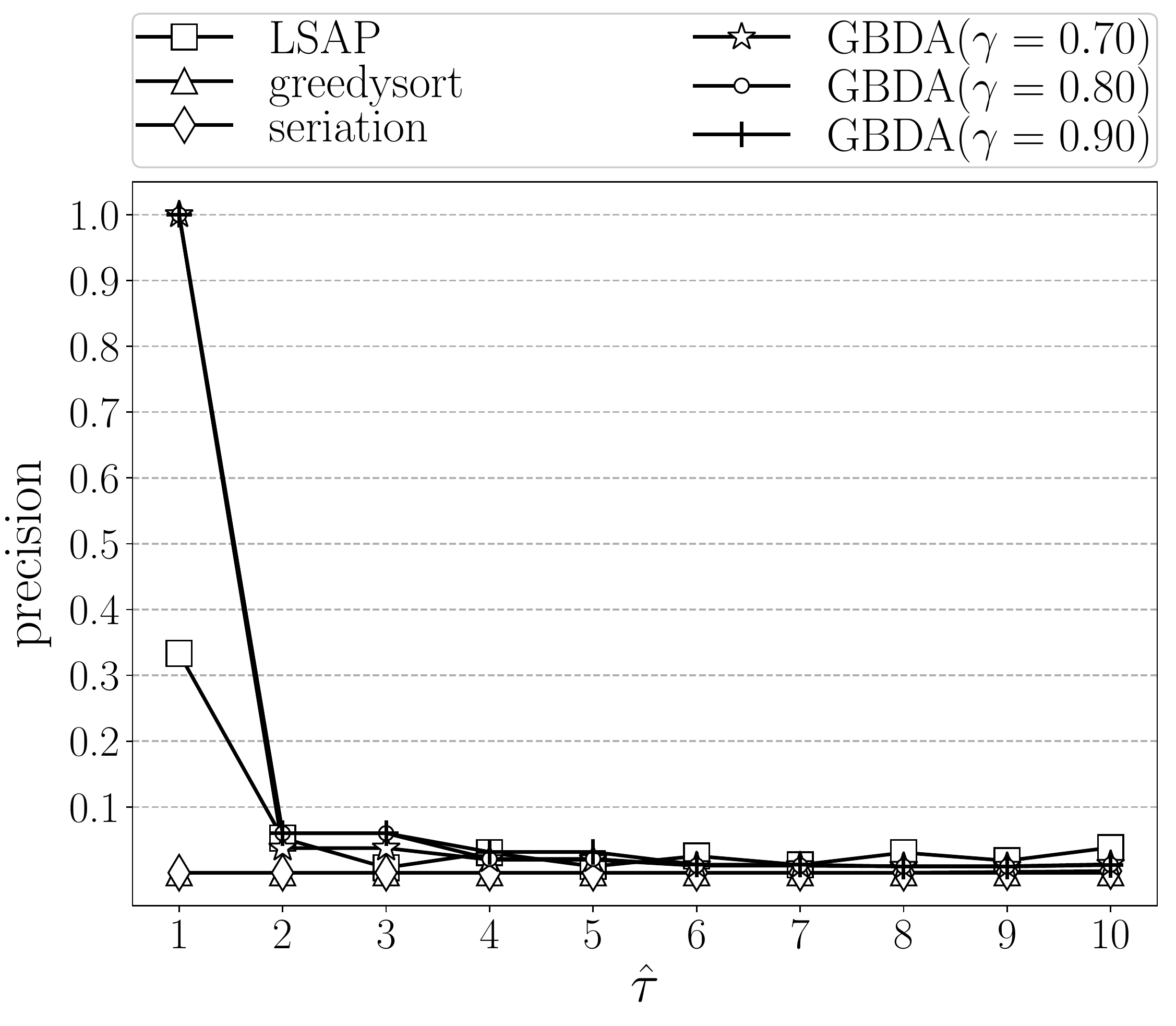}
    \vspace{-25pt}
    \caption{\footnotesize{Precision vs. $\hat\tau$ on AASD}}\label{fig-online-acc-aasd}
    \endminipage
\end{figure*}

\begin{figure*}[!t]    
    \vspace{-10pt}
    \centerfloat    
    \minipage{0.24\textwidth}
    \includegraphics[width=\linewidth]{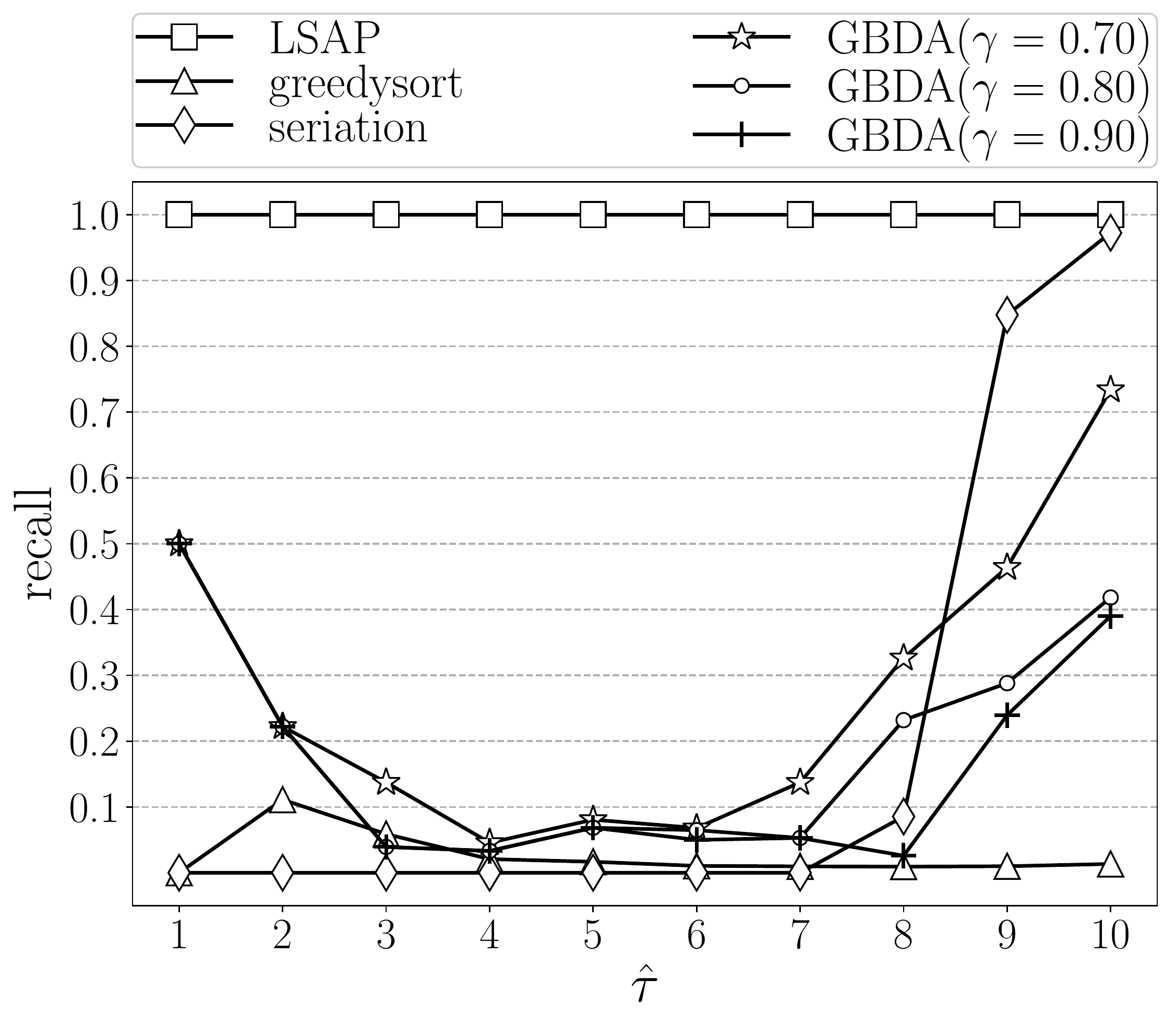}
    \vspace{-25pt}
    \caption{\footnotesize{Recall vs. $\hat\tau$ on AIDS}}\label{fig-online-recall-aids}
    \endminipage
    \hspace{10pt}
    \minipage{0.24\textwidth}
    \includegraphics[width=\linewidth]{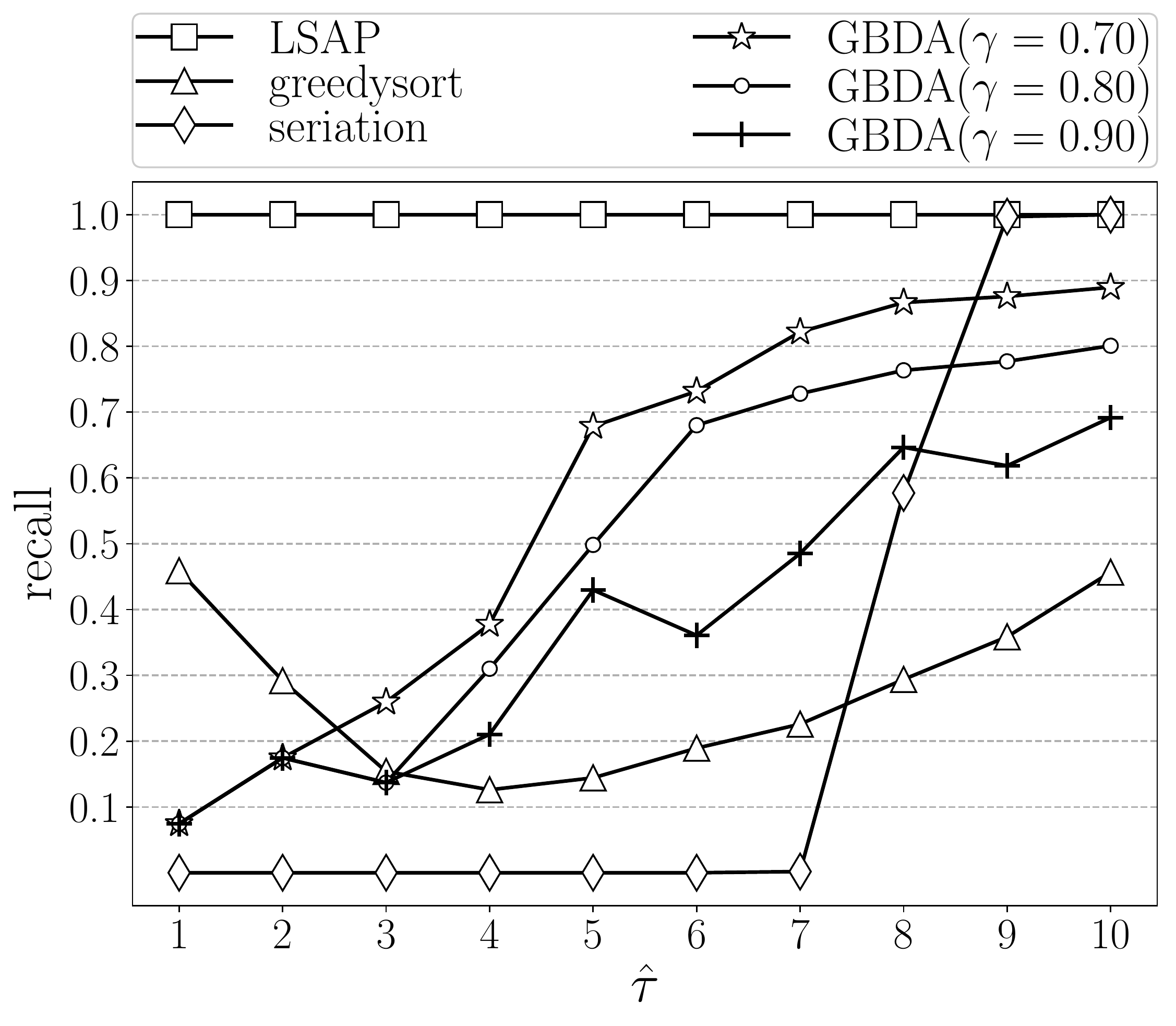}
    \vspace{-25pt}
    \caption{\footnotesize{Recall vs. $\hat\tau$ on Fingerprint}}\label{fig-online-recall-finger}
    \endminipage
    \hspace{10pt}
    \minipage{0.24\textwidth}
    \includegraphics[width=\linewidth]{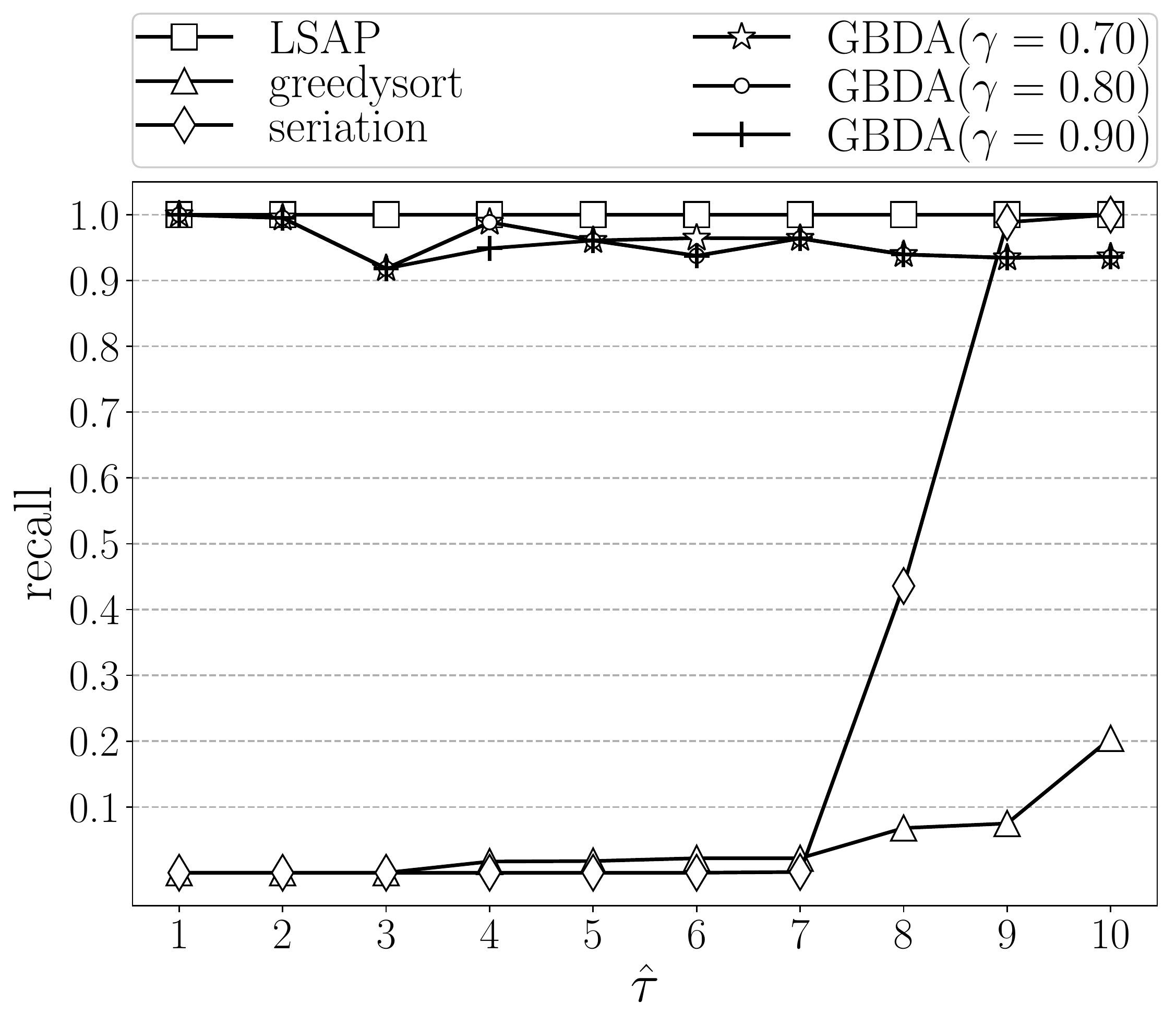}
    \vspace{-25pt}
    \caption{\footnotesize{Recall vs. $\hat\tau$ on GREC}}\label{fig-online-recall-grec}
    \endminipage
    \hspace{10pt}
    \minipage{0.24\textwidth}
    \includegraphics[width=\linewidth]{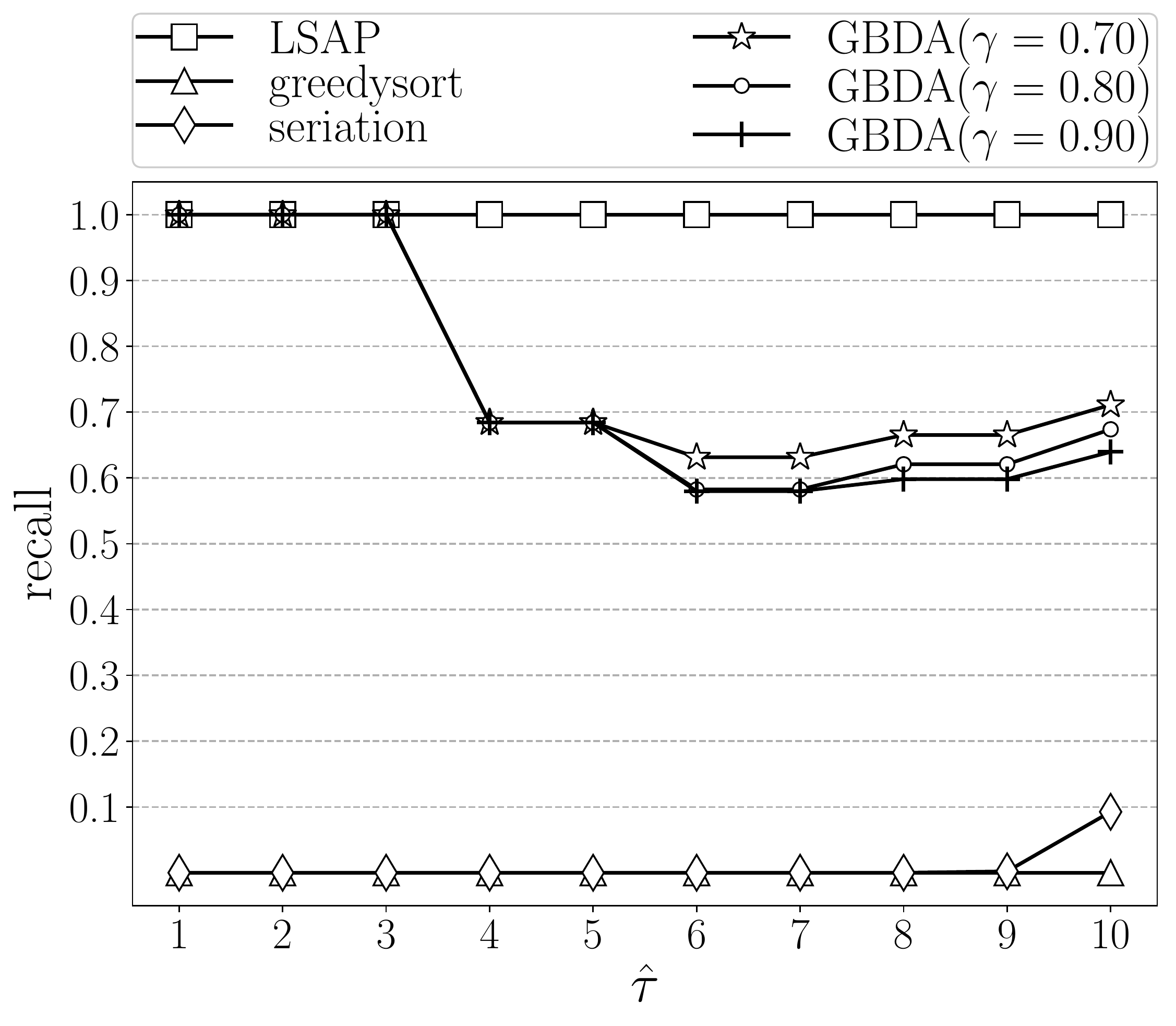}
    \vspace{-25pt}
    \caption{\footnotesize{Recall vs. $\hat\tau$ on AASD}}\label{fig-online-recall-aasd}
    \endminipage
\end{figure*}

\begin{figure*}[!t]
    \vspace{-10pt}
    \centerfloat
    \minipage{0.24\textwidth}
    \includegraphics[width=\linewidth]{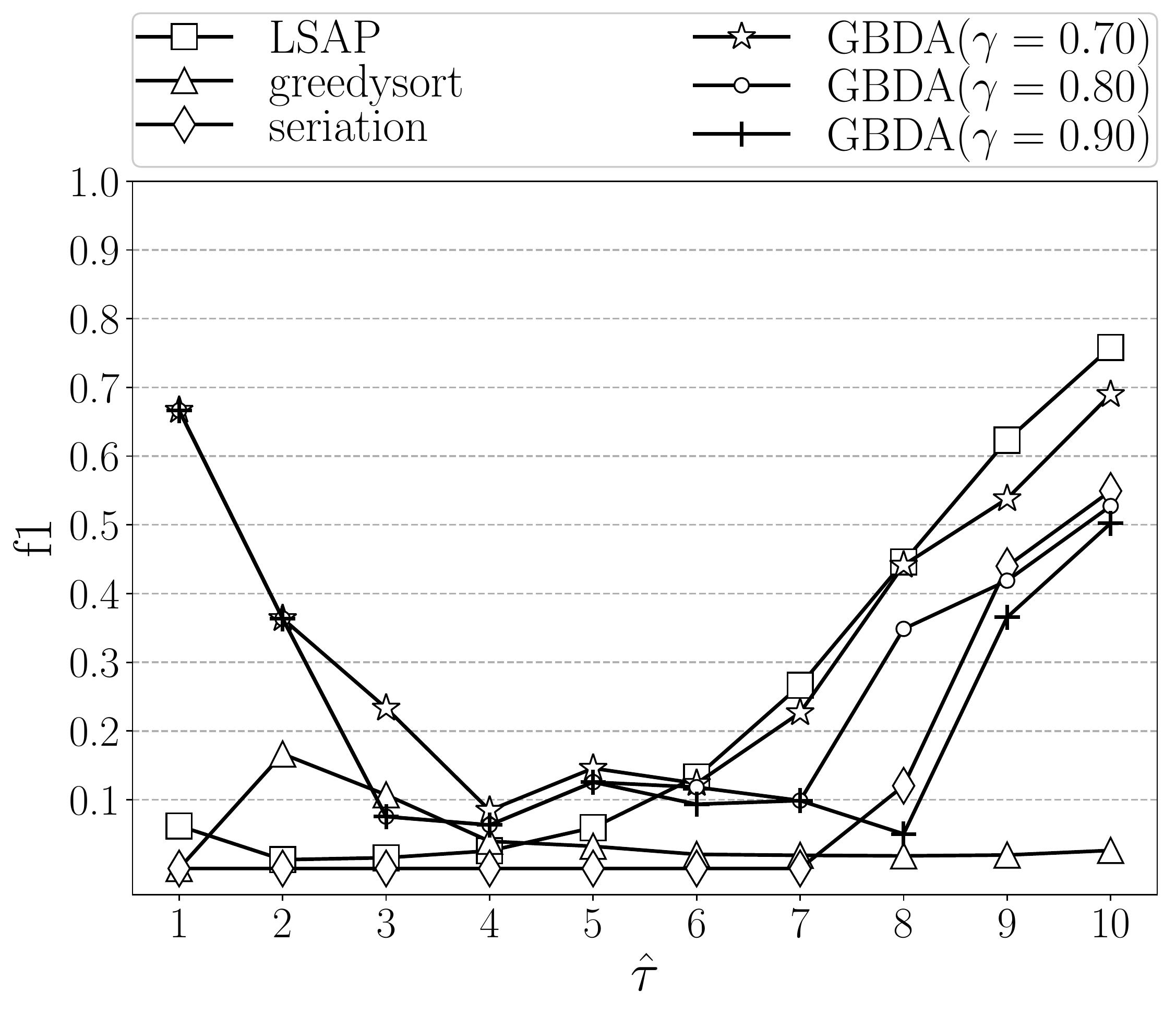}
    \vspace{-25pt}
    \caption{\footnotesize{F1-Score vs. $\hat\tau$ on AIDS}}\label{fig-online-f1-aids}
    \endminipage
    \hspace{10pt}
    \minipage{0.24\textwidth}
    \includegraphics[width=\linewidth]{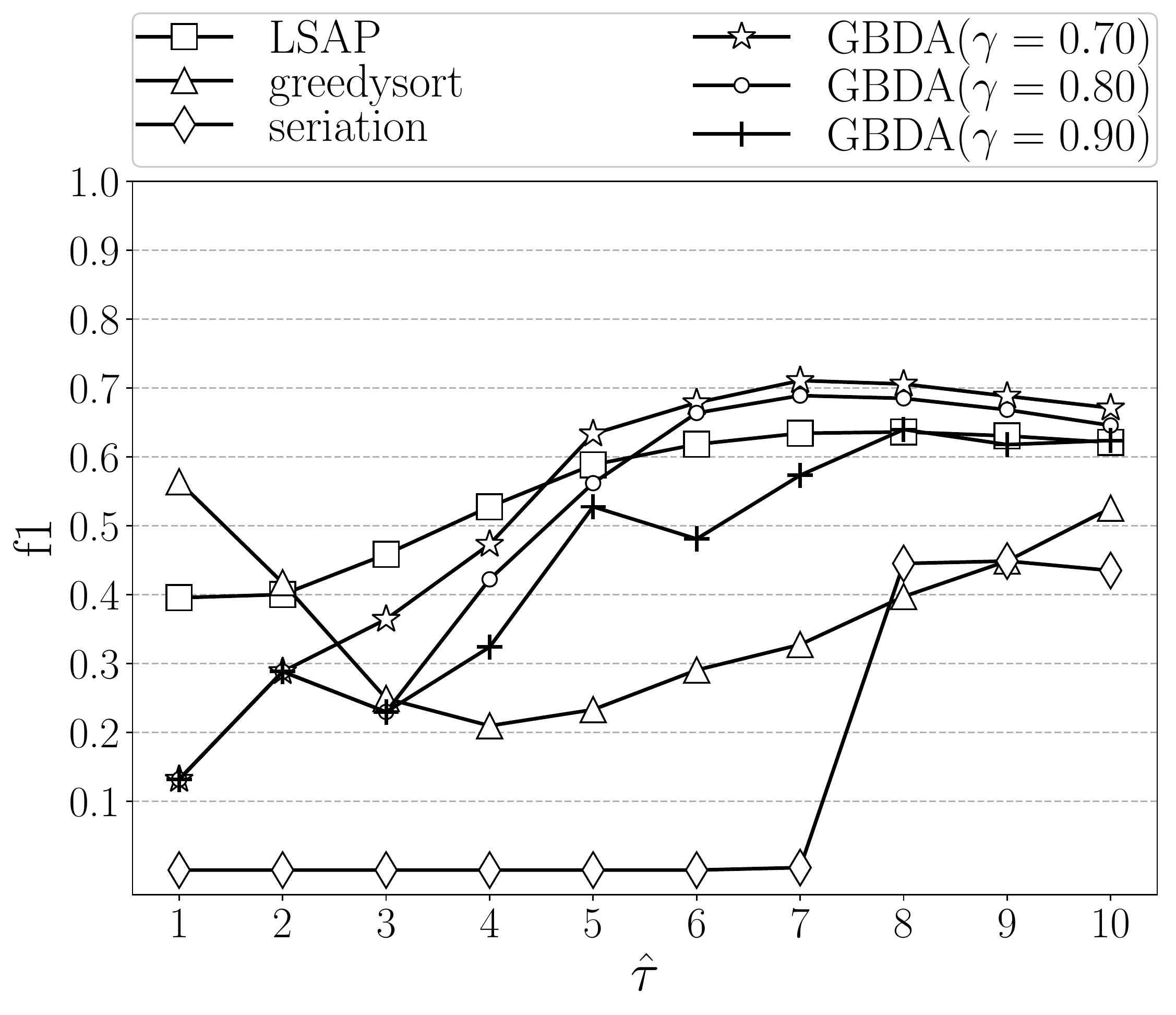}
    \vspace{-25pt}
    \caption{\hspace{-0.8em}\footnotesize{F1-Score vs. $\hat\tau$ on Fingerprint}}\label{fig-online-f1-finger}
    \endminipage
    \hspace{10pt}
    \minipage{0.24\textwidth}
    \includegraphics[width=\linewidth]{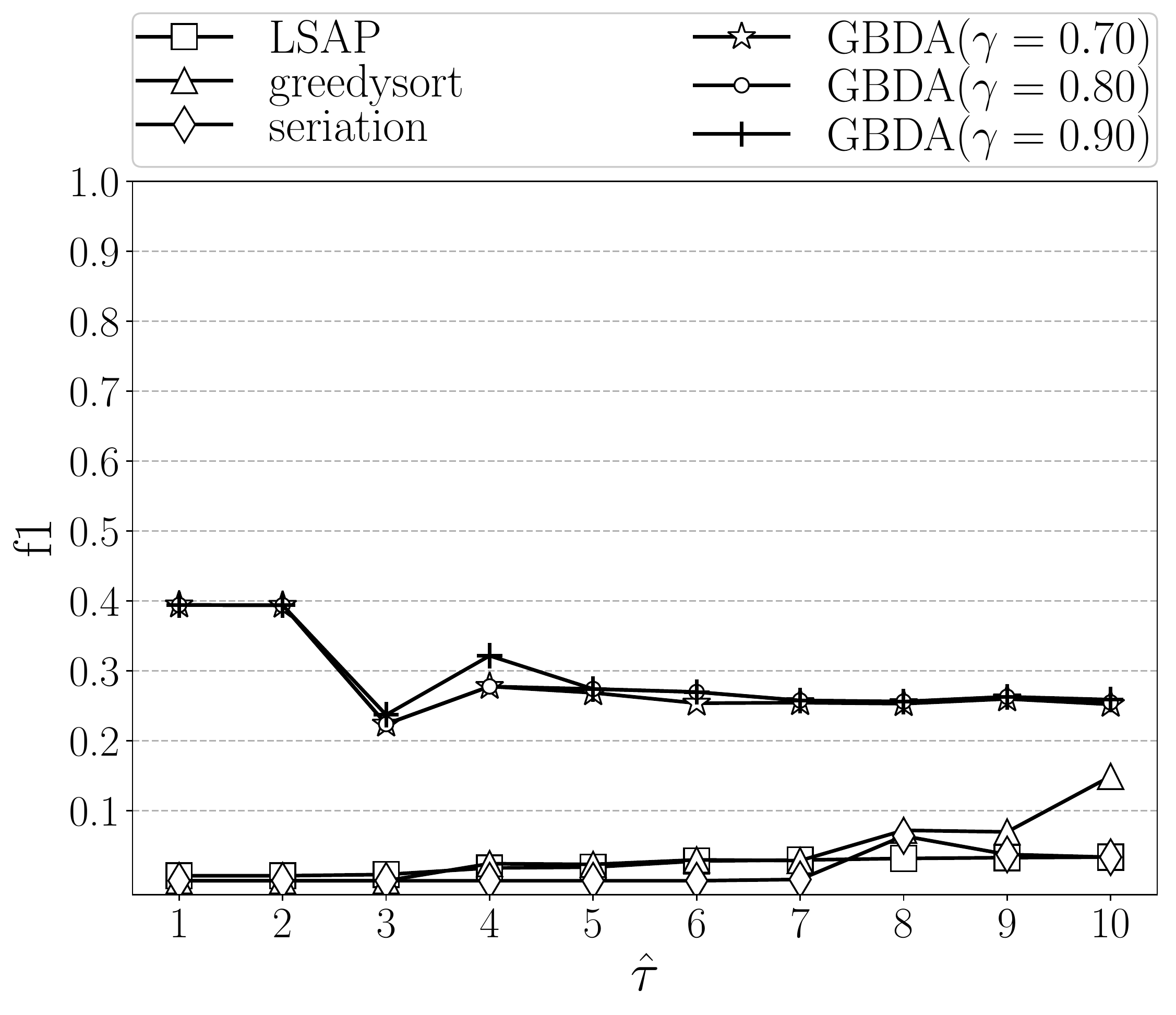}
    \vspace{-25pt}
    \caption{\footnotesize{F1-Score vs. $\hat\tau$ on GREC}}\label{fig-online-f1-grec}
    \endminipage
    \hspace{10pt}
    \minipage{0.24\textwidth}                                                          
    \includegraphics[width=\linewidth]{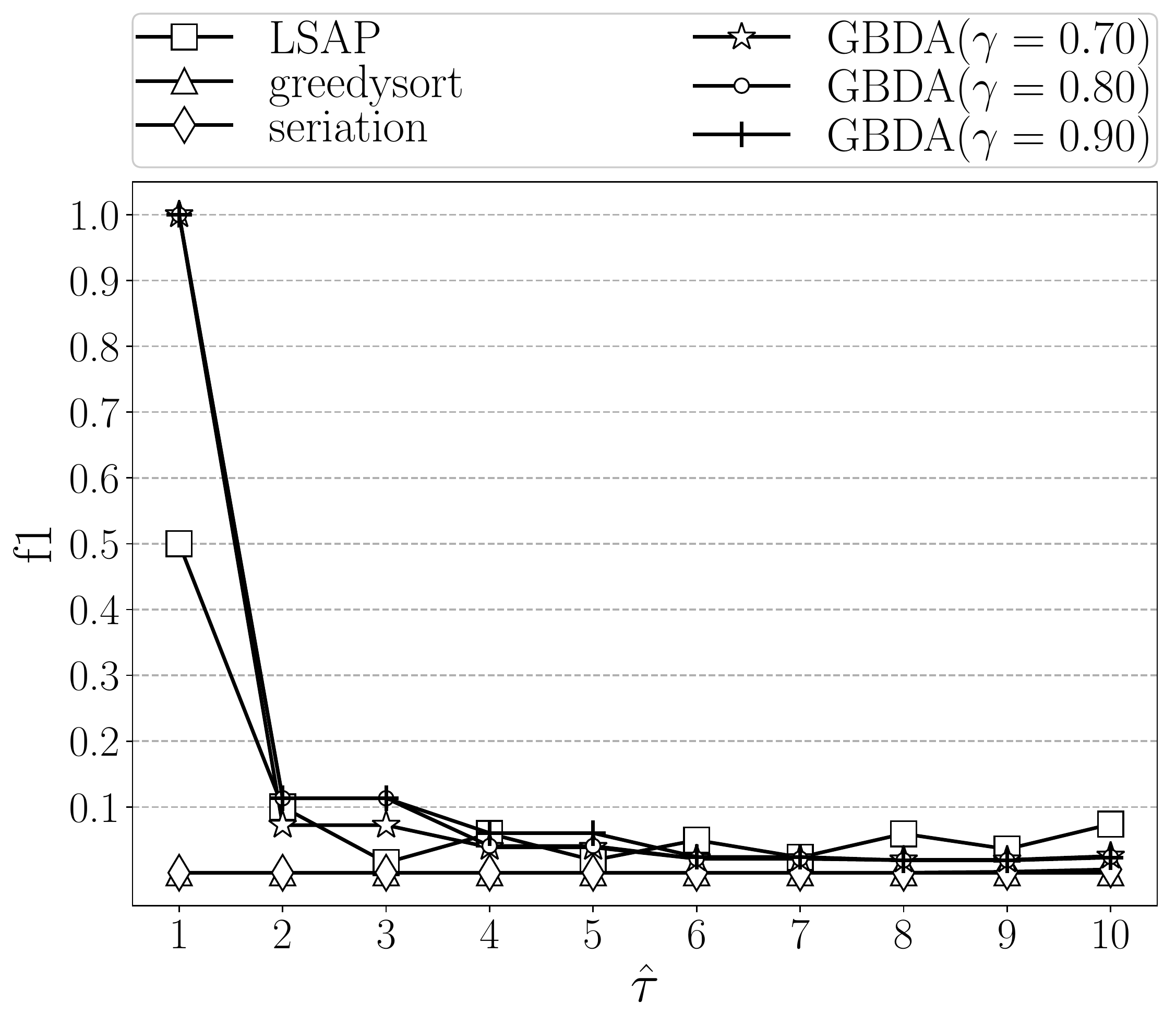}
    \vspace{-25pt}
    \caption{\footnotesize{F1-Score vs. $\hat\tau$ on AASD}}\label{fig-online-f1-aasd}
    \endminipage
    \vspace{-2em}
\end{figure*}

Note that, the costs of computing the GED prior distribution on \emph{synthetic} graphs do not exactly follow the theoretical analysis. This is because the numbers of vertices in synthetic graphs have only 7 possible values (i.e., 1K, 2K, 5K, 10K, 20K, 50K, and 100K), instead of the worst-case range $1\sim 100K$ as discussed in Section \ref{sec-complexity-offline}. In addition, although the synthetic graphs are larger than the real ones, the number of possible values of $n$ on synthetic graphs are smaller than that on real graphs, where $n$ is the number of vertices in graphs.
Therefore, the costs of computing the GED prior distribution on synthetic graphs are smaller than the costs on real data sets.

\vspace{-0.6em}
\subsection{Evaluating Online Stage}
\vspace{-0.3em}

In this subsection, we compare the efficiency and accuracy of the online stage of our GBDA approach with three
competitors (i.e., LSAP \cite{riesen2009approximate},
Greedy-Sort-GED \cite{riesen2015approximate} and Graph Seriation
\cite{robles2005graph}), by conducting graph similarity search tasks
over both real and synthetic data sets. In addition, we analyze how the efficiency and effectiveness (i.e., accuracy, recall and F1-score) of our method are influenced by the parameters, such as the similarity threshold $\hat{\tau}$, the probability threshold $\gamma$, and the number, $n$, of vertices.

\subsubsection{The Efficiency Evaluation}

In this subsection, we evaluate the query efficiency of our GBDA
approach and three competitors on real and synthetic data sets. Note that the time cost of our GBDA methods depends on both $n$ and $\hat{\tau}$, while the competitors' time costs only depend on $n$, where $n$ is the number of vertices in graphs and $\hat{\tau}$ is the similarity threshold. Therefore, the experiments in this subsection are conducted under a fixed probability threshold $\gamma=0.9$, since $\gamma$ does not affect the time costs of all the methods.

Specifically, given a specific method and its parameters (e.g., $\hat{\tau}$ and $\gamma$), for each query graph $Q$, we utilize this method to obtain a set of graphs similar to graph $Q$ from each data set, and we record the average query response time for each data set, which are presented in Figures \ref{fig-online-time-real}$\sim$\ref{fig-online-time-syn2}. Particularly, for a specific method under one specific parameter set (e.g., $\hat{\tau}$ and $\gamma$), each query's response time is recorded and counted only once for each data set, and we present the average of response times for all queries in the experimental figures.

The result in Figure \ref{fig-online-time-real} shows that our GBDA approach is more efficient than the three competitors on all real data sets where $\hat{\tau}$ is set to 1, 5 and 10, respectively. 
In addition, we studied how the number, $n$, of vertices in graphs influences the efficiency of our GBDA approach by comparing the query response time on synthetic data sets with various similarity thresholds $\hat{\tau}$, where the results are shown in Figures \ref{fig-online-time-syn1} and \ref{fig-online-time-syn2}. The results show that our GBDA approach is more efficient than the competitors on both scale-free and non-scale-free graphs, where the similarity threshold $\hat{\tau}\le 20$. Particularly, when the similarity threshold $\hat{\tau}= 30$, although our GBDA approach costs more time than the other methods on graphs with 1,000 vertices, our approach is faster than the competitors on larger graphs with more than 2,000 vertices. Therefore, when the similarity threshold is larger than the commonly-used values (i.e., $\{1,2,\allowbreak...,10\}$), the time cost of our algorithm is still smaller than the compared methods on large graphs (i.e., graphs with more than 2,000 vertices).

Note that the competitors (i.e., LSAP, Greedy-Sort-GED and Graph Seriation) can handle graphs with at most 20K vertices on our machine. Specifically, when the graphs have more than 20K vertices, all the competitors consume more than 128 GB memory on our machine, which exceeds the capacity of the physical memory. However, our GBDA method can handle graphs with 100K vertices efficiently, which confirms that our method has better scalability (with respect to the number, $n$, of vertices) than the competitors.

\subsubsection{The Effectiveness Evaluation}

We evaluate the effectiveness of our GBDA approach and three competitors by comparing the precision, recall, and F1-score \cite{wiki:f1-score} of the query results on each real data set with probability thresholds $\gamma=0.7,0.8$ and $0.9$, respectively.

The results in Figures \ref{fig-online-acc-aids}$\sim$\ref{fig-online-acc-aasd} show that our approach always outperforms the other three competitors in precision on AIDS, Fingerprint and GREC data sets, and achieves the highest precisions among all methods where $\hat{\tau}=1,2,3,4,5$ and $7$ on AASD data set.
The results in Figures \ref{fig-online-recall-aids}$\sim$\ref{fig-online-recall-aasd} show that our method has the second highest recalls under most parameter settings. Note that, LSAP method returns a lower bound of GED \cite{riesen2009approximate}, and therefore the recall of its search result is always 100\%. However, if we evaluate the methods by the F1-score, our method always outperforms the other three competitors on AIDS, Fingerprint and GREC data sets, and achieves the highest F1-scores among all methods when $\hat{\tau}=1,2,3,4,5$ and $7$ on AASD data set. Therefore, the experimental results confirm that the effectiveness of our method outperforms the competitors' under most parameter settings on the real data sets in our experiments.

In addition, we study how the number, $n$, of vertices in graphs influences the effectiveness (i.e., precision, recall and F1-score) of our approach on the Syn-1 data set with various probability thresholds $\gamma$ and similarity thresholds $\hat{\tau}$ in Appendix \ref{append-exp-online-syn}.

The results in Appendix \ref{append-exp-online-syn} show that, the precision of our method outperforms the competitors on \mbox{Syn-1} data set, where the similarity threshold $\hat{\tau}=15,20,25$ and $30$. Moreover, there is no significant difference between the precision of our method under various settings of the probability threshold $\gamma$, which demonstrates the robustness of our method under different settings of parameters. The recalls of our method are slightly lower than the LSAP method, but much higher than the Greedy-Sort-GED and Graph Seriation methods, where the similarity threshold $\hat{\tau}=20,25$ and $30$. Finally, the F1-Scores of our method are mostly higher than the competitors. Therefore, the experimental results on synthetic graphs demonstrate that our method is more effective than the competitors on large graphs.

\vspace{-1.0em}
\subsection{Comparing with Alternatives}
\vspace{-0.5em}

In this subsection, we compare the effectiveness of our GBDA approach and two variants of our method by comparing the F1-score \cite{wiki:f1-score} of the query results on each real data set with probability thresholds $\gamma=0.9$. The variants of our GBDA method, i.e., GBDA-V1 and GBDA-V2, are considered as alternatives of GBDA method, which are illustrated as follows.

\begin{description}[leftmargin=0pt,topsep=0pt]
    \item[GBDA-V1:] The method GBDA-V1 utilizes the average number of vertices among a sample of graphs from the graph database as the parameter $|V_1'|$ when computing $\Lambda_1$ and $\Lambda_3$ in Algorithm \ref{algo-similarity-search}, instead of using the number of vertices in the extended query graph $Q'$ as the parameter $|V_1'|$. 
    
    \item[GBDA-V2:] The method GBDA-V2 exploits the variant GBD (VGBD) instead of the original GBD value when computing $\Lambda_1$ and $\Lambda_2$ in Algorithm \ref{algo-similarity-search}, where VGBD is defined as:
    \begingroup
    \setlength{\abovedisplayskip}{2pt}
    \setlength{\belowdisplayskip}{2pt}
    \setlength{\abovedisplayshortskip}{2pt}
    \setlength{\belowdisplayshortskip}{2pt}
    \begin{align} \label{equal-def-vgbd}
    \hspace*{-0.5em}VGBD(G_1,G_2)&=\max\{|V_1|, |V_2|\}-w\cdot |B_{G_1}\cap B_{G_2}|
    \end{align}
    \endgroup
    where $B_{G_1}$ and $B_{G_2}$ are the multisets of all branches in graphs $G_1$ and $G_2$, respectively, $|V_1|$ and $|V_2|$ are numbers of vertices in graphs $G_1$ and $G_2$, respectively, and $w$ is a user-defined constant.
\end{description}

\begin{figure*}[!t]  
    \vspace{-30pt}
    \centerfloat
    \minipage{0.24\textwidth}
    \includegraphics[width=\linewidth]{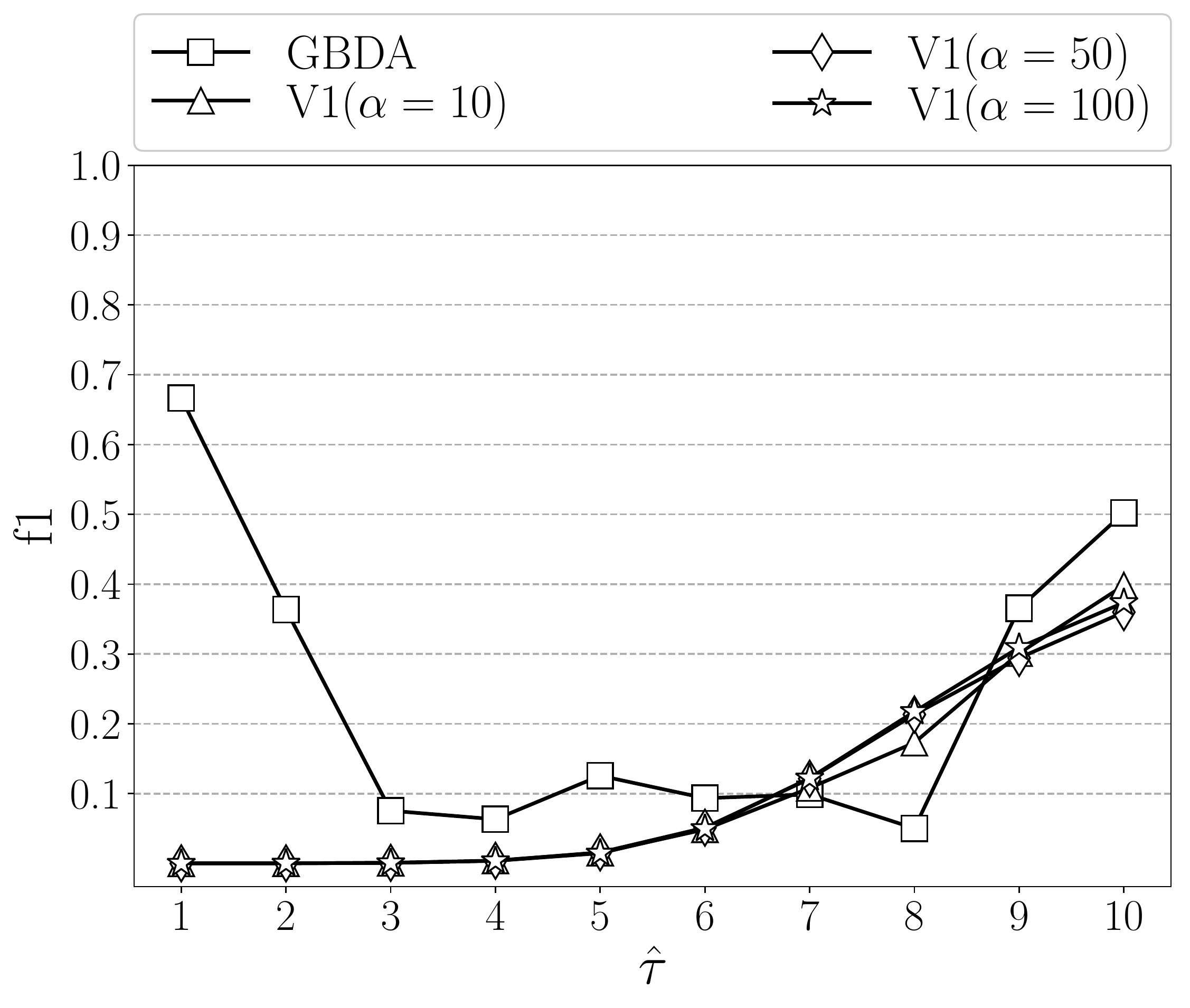}
    \vspace{-25pt}
    \caption{\footnotesize{F1-Score vs. $\hat\tau$ on AIDS\newline (Compared with GBDA-V1)}}\label{fig-v1-f1-aids}
    \endminipage
    \hspace{10pt}
    \minipage{0.24\textwidth}
    \includegraphics[width=\linewidth]{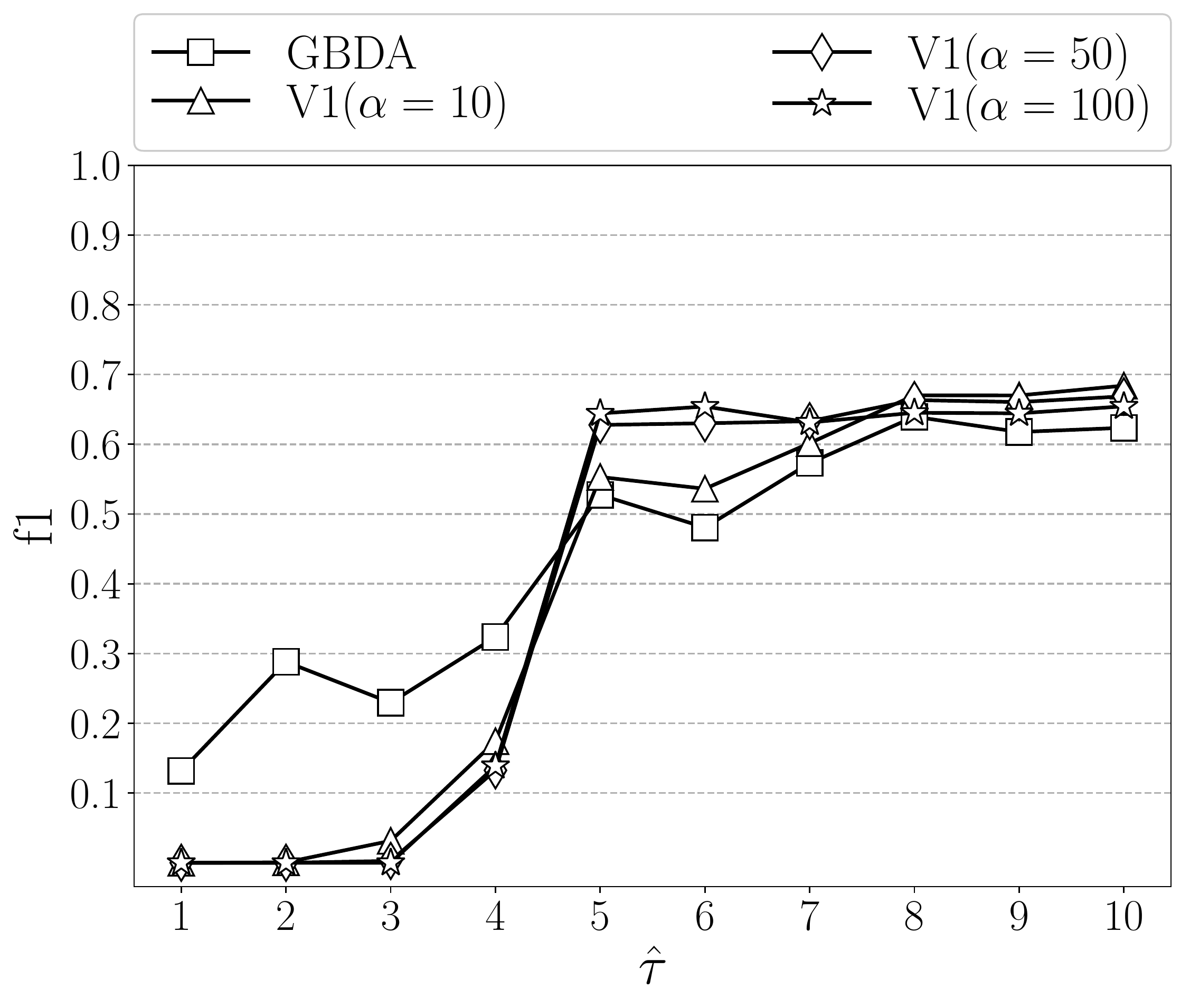}
    \vspace{-25pt}
    \caption{\hspace{-0.8em}\footnotesize{F1-Score vs. $\hat\tau$ on Fingerprint \newline (Compared with GBDA-V1)}}\label{fig-v1-f1-finger}
    \endminipage
    \hspace{10pt}
    \minipage{0.24\textwidth}
    \includegraphics[width=\linewidth]{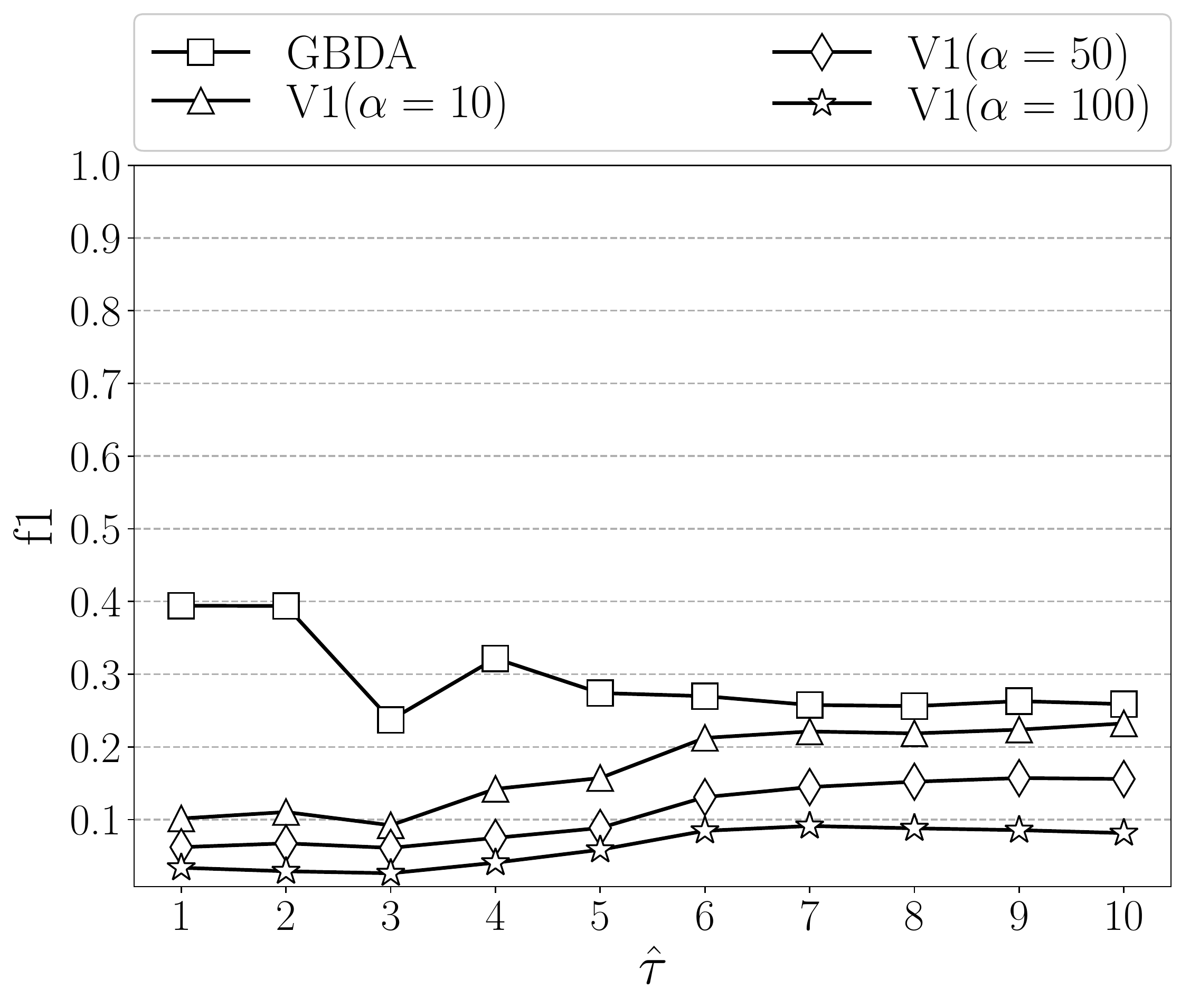}
    \vspace{-25pt}
    \caption{\footnotesize{F1-Score vs. $\hat\tau$ on GREC \newline (Compared with GBDA-V1)}}\label{fig-v1-f1-grec}
    \endminipage
    \hspace{10pt}
    \minipage{0.24\textwidth}                                                          
    \includegraphics[width=\linewidth]{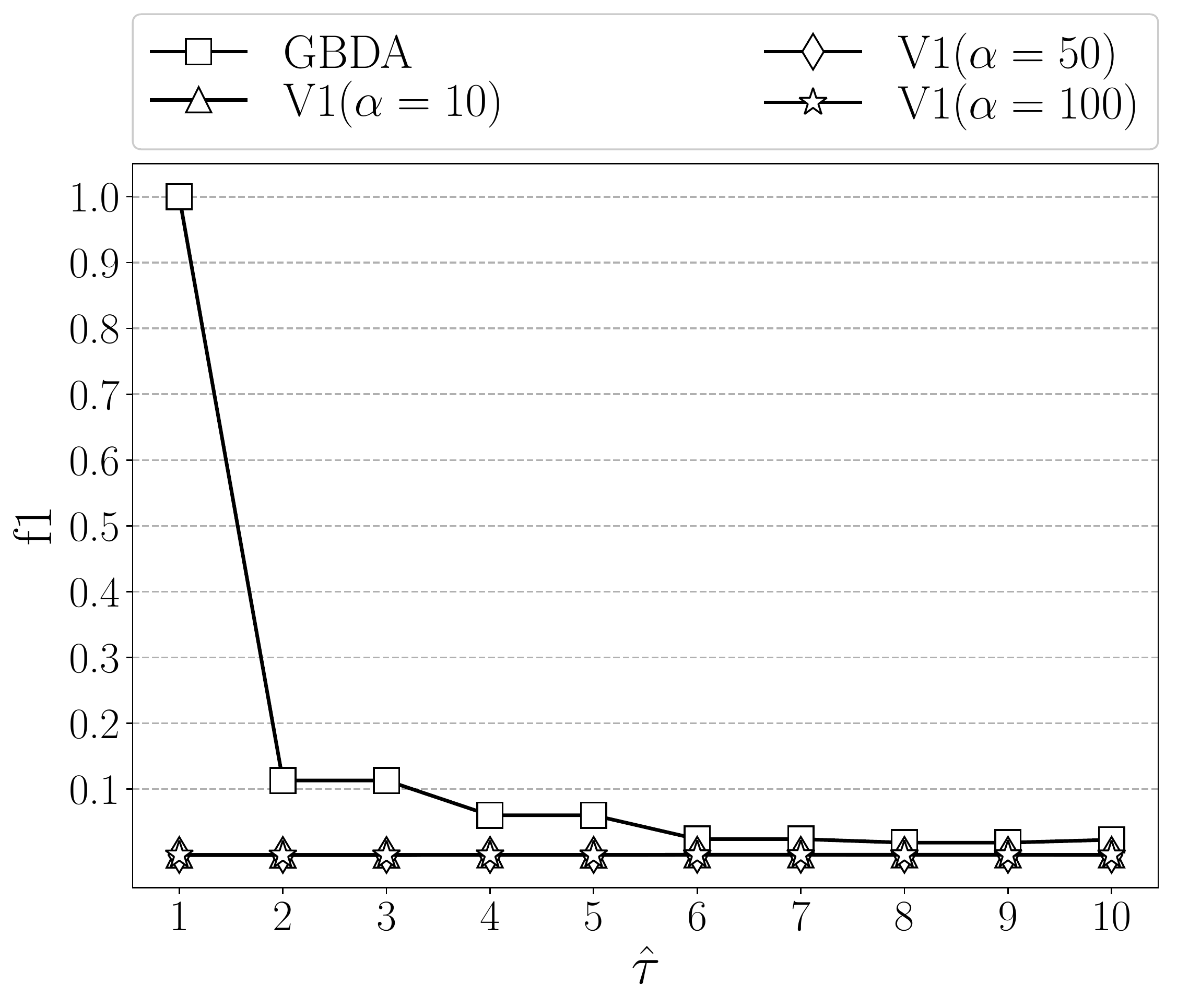}
    \vspace{-25pt}
    \caption{\footnotesize{F1-Score vs. $\hat\tau$ on AASD \newline (Compared with GBDA-V1)}}\label{fig-v1-f1-aasd}
    \endminipage
    \vspace{-10pt}
\end{figure*}

\begin{figure*}[!t]  
    \centerfloat
    \minipage{0.24\textwidth}
    \includegraphics[width=\linewidth]{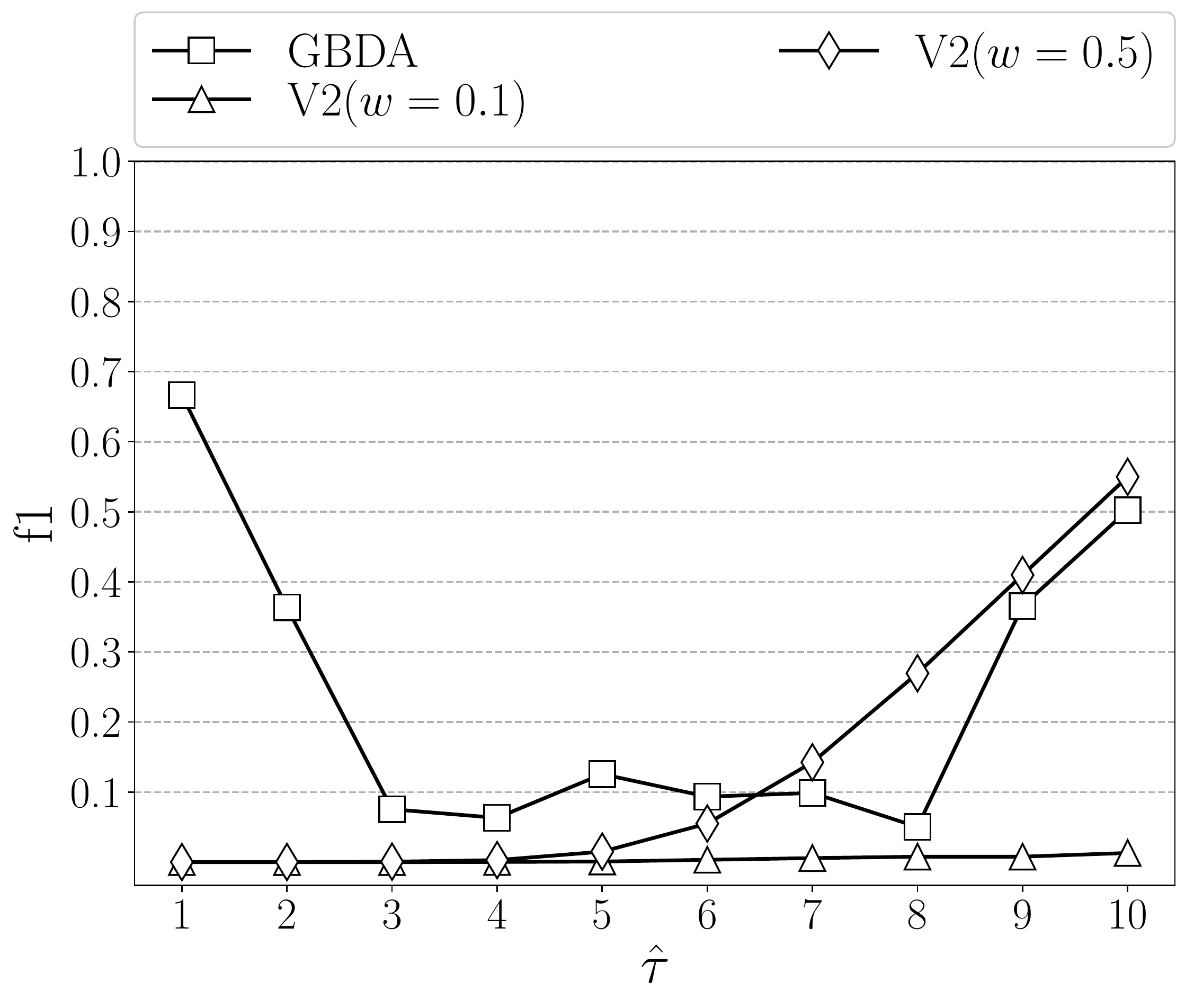}
    \vspace{-25pt}
    \caption{\footnotesize{F1-Score vs. $\hat\tau$ on AIDS \newline (Compared with GBDA-V2)}}\label{fig-v2-f1-aids}
    \endminipage
    \hspace{10pt}
    \minipage{0.24\textwidth}
    \includegraphics[width=\linewidth]{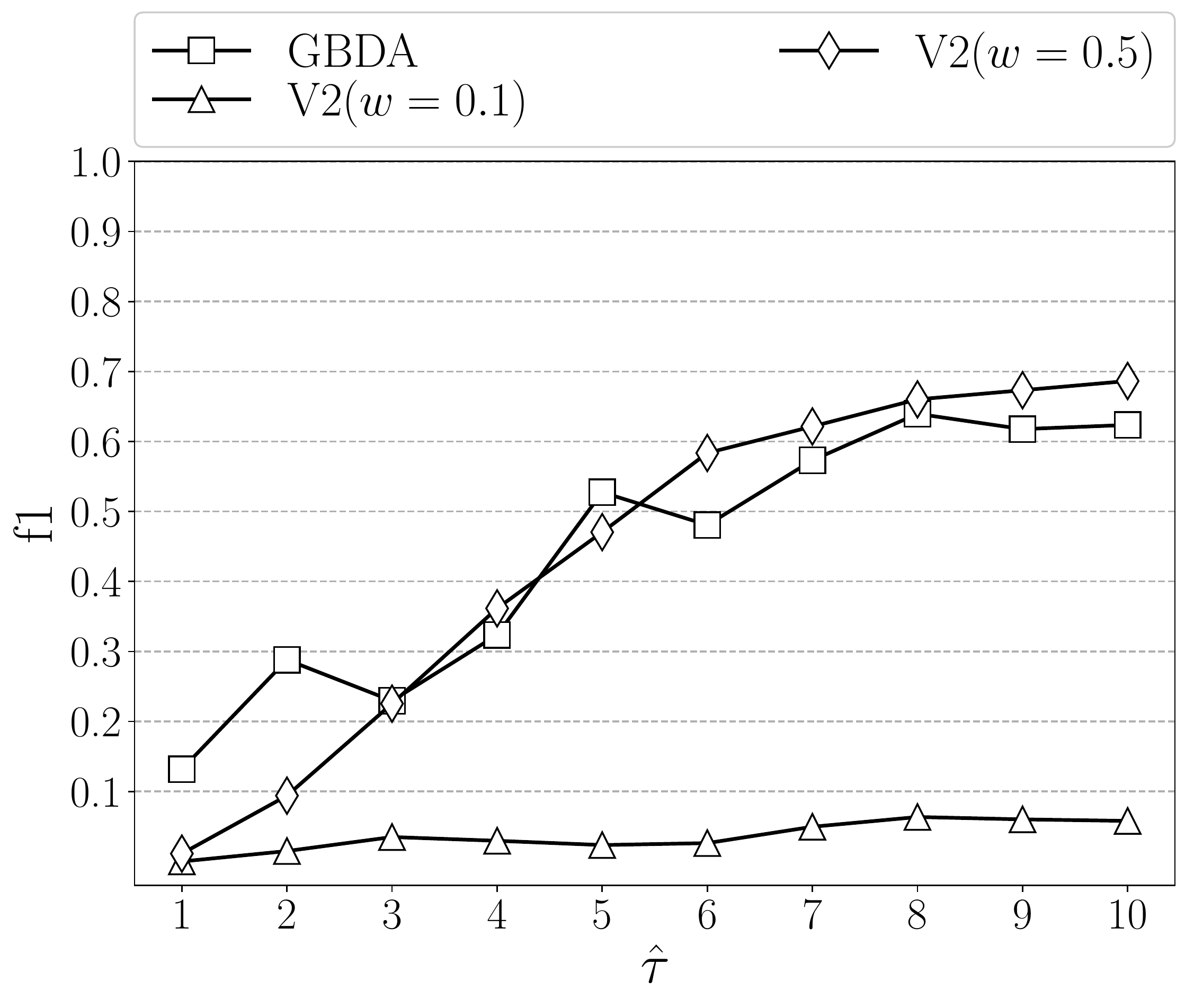}
    \vspace{-25pt}
    \caption{\hspace{-0.8em}\footnotesize{F1-Score vs. $\hat\tau$ on Fingerprint \newline (Compared with GBDA-V2)}}\label{fig-v2-f1-finger}
    \endminipage
    \hspace{10pt}
    \minipage{0.24\textwidth}
    \includegraphics[width=\linewidth]{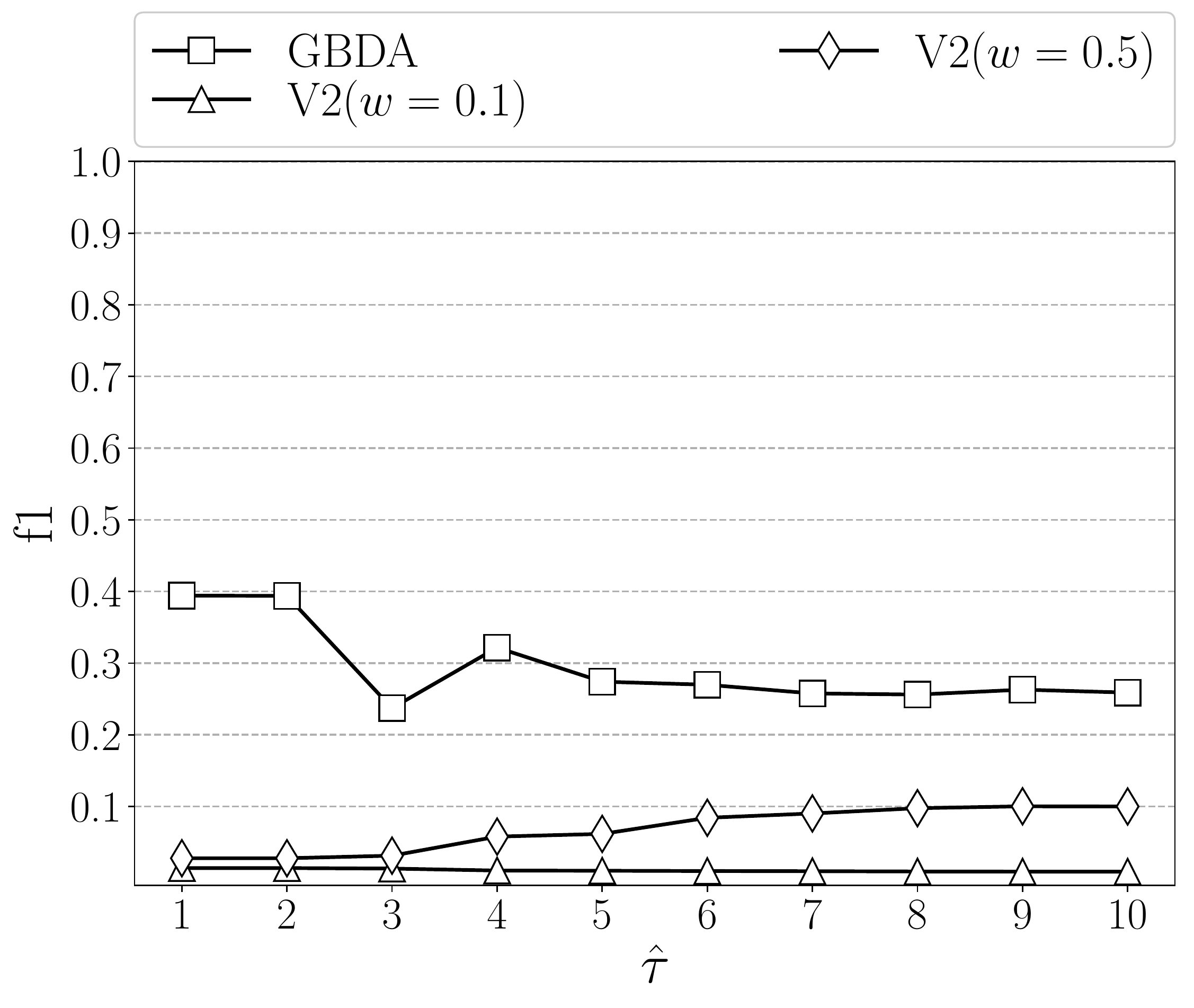}
    \vspace{-25pt}
    \caption{\footnotesize{F1-Score vs. $\hat\tau$ on GREC \newline (Compared with GBDA-V2)}}\label{fig-v2-f1-grec}
    \endminipage
    \hspace{10pt}
    \minipage{0.24\textwidth}                                                          
    \includegraphics[width=\linewidth]{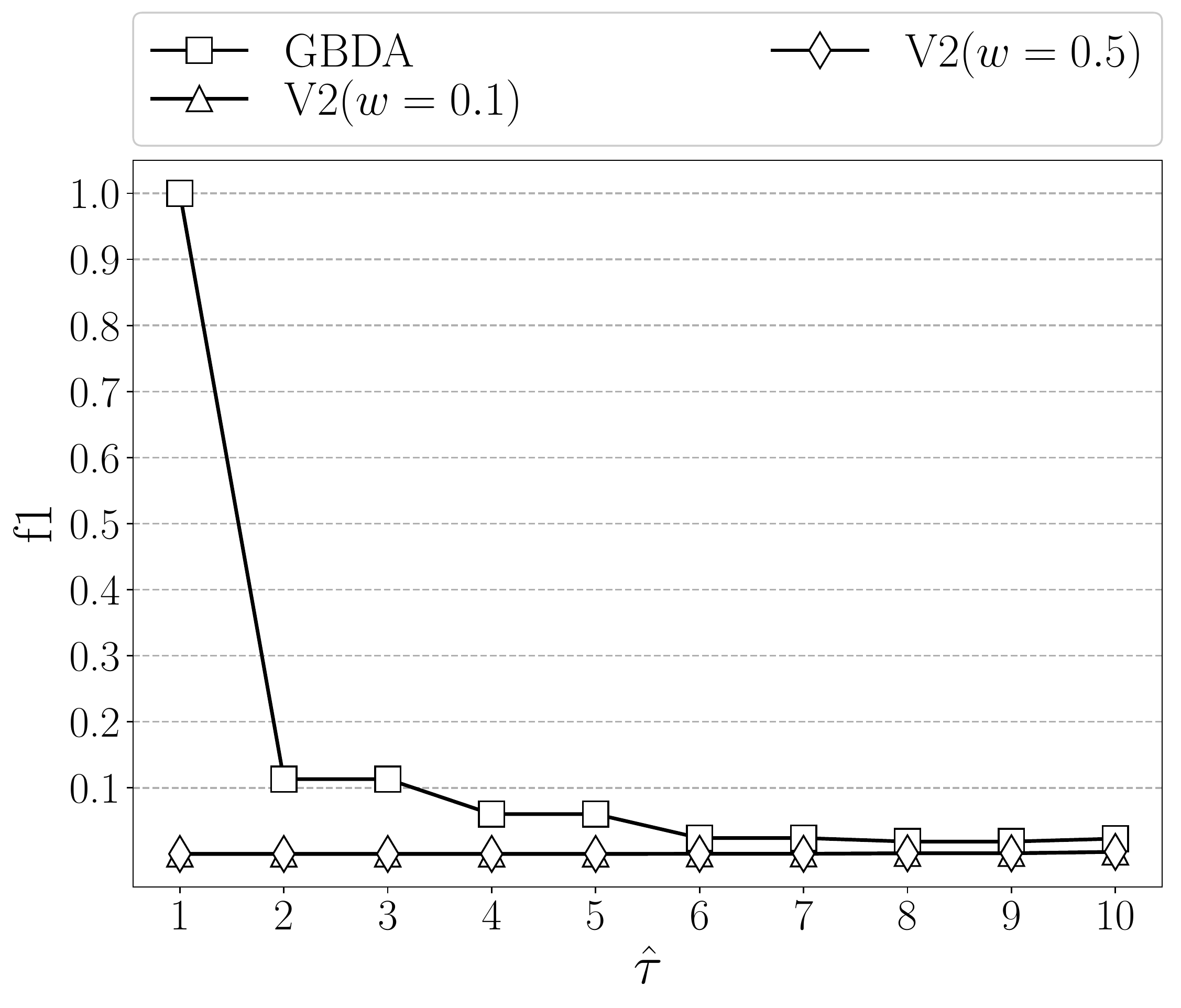}
    \vspace{-25pt}
    \caption{\footnotesize{F1-Score vs. $\hat\tau$ on AASD \newline (Compared with GBDA-V2)}}\label{fig-v2-f1-aasd}
    \endminipage
    \vspace{-20pt}
\end{figure*}

Specifically, we denote the number of sampled graphs in method GBDA-V1 by $\alpha$, and we test the method GBDA-V1 by setting $\alpha=10$, $50$, and $100$ on real data sets. In addition, we evaluate the effectiveness of method GBDA-V2 on real data sets by setting the parameter $w=0.1$ and $0.5$, where $w$ is defined in Equation (26). 

The results in Figures \ref{fig-v1-f1-aids}$\sim$\ref{fig-v1-f1-aasd} show that our GBDA method achieves higher F1-score than its variant GBDA-V1 for the similarity threshold $\hat{\tau} \le 4$, but generally has the same F1-score as GBDA-V1 for $\hat{\tau} \ge 5$. In addition, from the results in Figures \ref{fig-v2-f1-aids}$\sim$\ref{fig-v2-f1-aasd}, we can find that the F1-score of our GBDA method is higher than or almost the same as the \mbox{GBDA-V2} method for the similarity threshold $\hat{\tau} \le 2$ on all real data sets, and slightly lower than GBDA-V2 method on Fingerprint data set for $\hat{\tau} \ge 3$. In general, the experimental results show that our GBDA method outperforms methods GBDA-V1 and GBDA-V2 in most cases for similarity threshold $\hat{\tau} \le 5$, and performs better or almost the same as GBDA-V1 and GBDA-V2 methods for the similarity threshold $\hat{\tau} \ge 6$.

\vspace{-0.6em}
\section{Related Works}\label{sec-related-works}
\vspace{-0.6em}

\subsection{Exact GED Computation}
\vspace{-0.6em}

The state-of-the-art method for exact GED computation is the $A^*$ algorithm \cite{hart1968formal} and its variant \cite{DBLP:conf/icde/GoudaH16}, whose time costs are exponential with respect to graph sizes \cite{zeng2009comparing}. To address this NP-hardness problem, most graph similarity search method are based on the filter-and-verification framework \cite{zeng2009comparing}  \cite{wang2012efficiently} \cite{zhao2012efficient} \cite{Zheng:2013hh}, 
which first filters out undesirable graphs from the graph databases and then only verifies the remaining candidates. A common filtering approach is to use the distance between sub-structures of two graphs as a lower bound of their GED, which includes tree-based \cite{wang2012efficiently}, path-based \cite{zhao2013efficient} , branch-based \cite{Zheng:2013hh} and partition-based \cite{zhao2013partition} approaches. In this paper, we adopt the branch structure \cite{Zheng:2013hh} to build our model. However, we re-define the distance between branches, since the original definition \cite{Zheng:2013hh} of branch distances requires $O(n^3)$ time for computation while ours only requires $O(nd)$ time. In addition, a recent paper \cite{liang2017similarity} propose a multi-layer indexing approach to accelerate the filtering process based on their proposed partition-based filtering method.

\vspace{-0.6em}
\subsection{GED Estimation}
\vspace{-0.3em}

In this paper, we focus on GED estimation approaches. One well-studied method \cite{bougleux2016graph} \cite{riesen2009approximate}  
is to utilize the solution of a
\emph{linear sum assignment problem} (LSAP) as an estimation of GED.
The LSAP is an optimization problem which can be exactly solved by Hungarian method\cite{kuhn1955hungarian}, or be approximately solved by the greedy method\cite{riesen2015approximate} and the genetic algorithm\cite{riesen2014improving}. In our experiment, we compare our GBDA method with the exact \cite{riesen2009approximate} and greedy \cite{riesen2015approximate} solutions of LSAP. Since the exact solution of LSAP defined on two graphs is a lower bound of their GED \cite{riesen2009approximate}, the LSAP method can always obtain all graphs whose GED to the query graph is no larger than the similarity threshold and achieve 100\% recall in the similarity search tasks. On the other hand, the Greedy-Sort-GED method \cite{riesen2015approximate} solves LSAP approximately and has no bound to actual GED. However, the Greedy-Sort-GED method generally achieves better estimations of GEDs \cite{riesen2015approximate} and higher precisions in graph similarity search as shown in our experiments.
Another state-of-the-art approach compared in our experiment is \emph{graph seriation}\cite{robles2005graph},
which first converts graphs into one-dimensional vectors by extracting their leading eigenvalues of the adjacency matrix,
and then exploits a probabilistic model based on these vectors to estimate the GED.
Although both graph seriation and our approach utilize probabilistic models, the structure of our model is totally different from the prior work \cite{robles2005graph}. In addition, their model takes 
the leading eigenvalues of the adjacency matrix as the inputs, while the inputs of our model are the GBDs. Moreover, our GBDA method outperforms the competitors' (i.e., the LSAP \cite{riesen2009approximate},
Greedy-Sort-GED \cite{riesen2015approximate} and Graph Seriation
\cite{robles2005graph}) under most parameter settings on real data sets.

\vspace{-0.6em}
\section{Conclusions}  \label{sec-conclusions}
\vspace{-0.6em}
In this paper, we define the branch distance between two graphs (GBD), and further prove that the GBD has a probabilistic relationship with the GED by considering branch variations as the result ascribed to graph edit operations and modeling this process by probabilistic approaches. Furthermore, this relation between GED and GBD is leveraged to perform graph similarity searches. Experimental results demonstrate both the correctness and effectiveness of our approach, which outperforms the comparable methods.\vspace{-2ex} 


\bibliographystyle{IEEEtran}
\bibliography{aged}

\begin{appendices}
    \footnotesize
    
    \section{Proof of Theorem 1} \label{Appendix-Proof-Theorem-1}
    
    \begin{proof}
        Since GED satisfies the triangle inequality, we have
        \begingroup
        \setlength{\abovedisplayskip}{2pt}
        \setlength{\belowdisplayskip}{2pt}
        \setlength{\abovedisplayshortskip}{2pt}
        \setlength{\belowdisplayshortskip}{2pt}
        \begin{align*}
        GED(G_1,G_2)&\le GED(G_1,G_1') + GED(G_1',G_2') + GED(G_2',G_2)\\
        GED(G_1',G_2')&\le GED(G_1',G_1) + GED(G_1,G_2) + GED(G_2,G_2')
        \end{align*}
        \endgroup
        
        According to Definition \ref{def-ged}, adding a vertex or an edge with the virtual label $\varepsilon$ is not counted as a graph edit operation. Therefore, we have:
        \begingroup
        \setlength{\abovedisplayskip}{2pt}
        \setlength{\belowdisplayskip}{2pt}
        \setlength{\abovedisplayshortskip}{2pt}
        \setlength{\belowdisplayshortskip}{2pt}
        \begin{align*}
        GED(G_1,G_1')&=GED(G_1',G_1)=0\\
        GED(G_2,G_2')&=GED(G_2',G_2)=0
        \end{align*}
        \endgroup
        
        Therefore, $GED(G_1,G_2)\le GED(G_1',G_2') \le GED(G_1,G_2)$.
        
        That is, $GED(G_1,G_2)=GED(G_1',G_2')$. 
    \end{proof}
    
    \vspace{-1em}
    \section{Proof of Theorem 2} \label{Appendix-Proof-Theorem-2}
    
    \begin{proof}
        Let the sets of branches rooted at \emph{virtual} vertices
        in $G_1'$, $G_2'$ be $\Delta B_{G_1}$ and $\Delta B_{G_2}$, respectively.
        We have:
        \begingroup
        \setlength{\abovedisplayskip}{2pt}
        \setlength{\belowdisplayskip}{2pt}
        \setlength{\abovedisplayshortskip}{2pt}
        \setlength{\belowdisplayshortskip}{2pt}
        \begin{align*}
        \lvert B_{G_1'} \cap B_{G_2'} \rvert &= 
        \lvert B_{G_1} \cap B_{G_2} \rvert
        + \lvert \Delta B_{G_1} \cap B_{G_2} \rvert \\
        &+ \lvert B_{G_1} \cap \Delta B_{G_2} \rvert 
        + \lvert \Delta B_{G_1} \cap \Delta B_{G_2} \rvert 
        \end{align*}
        \endgroup
        
        Since branches rooted at \emph{virtual} vertices are not isomorphic 
        to any other branches, we also have:
        \begingroup
        \setlength{\abovedisplayskip}{2pt}
        \setlength{\belowdisplayskip}{2pt}
        \setlength{\abovedisplayshortskip}{2pt}
        \setlength{\belowdisplayshortskip}{2pt}
        $$ \lvert \Delta B_{G_1} \cap B_{G_2} \rvert 
        = \lvert B_{G_1} \cap \Delta B_{G_2} \rvert
        = \lvert \Delta B_{G_1} \cap \Delta B_{G_2} \rvert = 0 $$
        \endgroup

        Therefore, $\lvert B_{G_1'} \cap B_{G_2'} \rvert = \lvert B_{G_1} \cap B_{G_2} \rvert$.
        
        From the definitions of branches and extended graphs, we have:
        \begingroup
        \setlength{\abovedisplayskip}{2pt}
        \setlength{\belowdisplayskip}{2pt}
        \setlength{\abovedisplayshortskip}{2pt}
        \setlength{\belowdisplayshortskip}{2pt}
        $$\max\{|V_1|, |V_2|\}=\max\{|V_1'|, |V_2'|\}$$
        \endgroup
        
        From the definition of GBD, we can obtain:
        \begingroup
        \setlength{\abovedisplayskip}{2pt}
        \setlength{\belowdisplayskip}{2pt}
        \setlength{\abovedisplayshortskip}{2pt}
        \setlength{\belowdisplayshortskip}{2pt}
        \begin{align*}
        GBD&(G_1,G_2)=\max\{|V_1|,|V_2|\}-|B_{G_1}\cap B_{G_2}| \\   
        &=\max\{|V_1'|, |V_2'|\}-|B_{G_1'}\cap B_{G_2'}|=GBD(G_1',G_2')
        \end{align*}
        \endgroup
        
        That is, $GBD(G_1,G_2)= GBD(G_1',G_2')$.
    \end{proof}
    
    \vspace{-0.8em}
    \section{Closed Forms of Equations} \label{Appendix-Closed-Forms-of-Equations}
    \vspace{-0.8em}
    
    \subsection{Closed Form of Equation (\ref{equal-lambda-1-expansion})}
    
    \begingroup
    \setlength{\abovedisplayskip}{-10pt}
    \setlength{\belowdisplayskip}{2pt}
    \setlength{\abovedisplayshortskip}{2pt}
    \setlength{\belowdisplayshortskip}{2pt}
    \begin{align}\label{equal-theo-all-trans}
    & \Lambda_1(G_1',G_2';\tau, \varphi) = Pr[GBD= \varphi \mid GED=\tau] \nonumber \\
    =& \sum_{x}\Omega_1(x,\tau) \sum_{m}{\Omega_2(m,x,\tau) \sum_{r} \Omega_3(r,\varphi) \Omega_4(x,r,m) } \\        
    &\Omega_1(x,\tau)=\textstyle \mathcal{H}\left(x; |V_1'|+{|V_1'| \choose 2}, |V_1'|, \tau\right) \label{equal-omega-1}  \\
    &\Omega_2(m,x,\tau)=\textstyle {{{|V_1'|\choose 2}\choose{\tau-x}}}^{-1}{\sum_{t=0}^m {(-1)^{m-t}}{{|V_1'|}\choose{m}}{m \choose t}{{t\choose 2}\choose{\tau-x}}} \label{equal-omega-2} \\
    &\Omega_3(r, \varphi)=\textstyle {r \choose r-\varphi}\cdot \frac{(\mathbb{D}-1)^{\varphi}}{\mathbb{D}^r} \label{equal-omega-3} \\
    &\Omega_4(x,r,m)=\textstyle \mathcal{H}\left(x+m-r; |V_1'|,m,x\right) \\ \label{equal-omega-4} 
    &\mathcal{H}(x ; M,K,N)=\textstyle{K \choose x}{{{M-K \choose N-x}}}{{M \choose N}}^{-1} \\
    &\mathbb{D}=\textstyle|\mathcal{L}_{V}| \cdot {|V_1'| + |\mathcal{L}_E| - 1 \choose |\mathcal{L}_E|}  \label{equal-ged-constant-d}
    \end{align}
    \endgroup
    
    \subsection{Closed Form of Equation (\ref{equal-prior-ged-jeff-final})}
    \begingroup
    \setlength{\abovedisplayskip}{-10pt}
    \setlength{\belowdisplayskip}{2pt}
    \setlength{\abovedisplayshortskip}{2pt}
    \setlength{\belowdisplayshortskip}{2pt}
    \begin{align}
    & \textstyle Pr[GED] = \frac{1}{C}  \sqrt{\sum_{\varphi=0}^{2\tau} \Lambda_1 \cdot \mathcal{Z}^2} \\
    \label{equal-closed-zfunc} & \hspace*{-2.7em} \textstyle \mathcal{Z} = \frac{1}{\Lambda_1}\Big\{\sum_x \Omega_1\sum_m\frac{\partial  \Omega_2}{\partial\tau}\sum_r\Omega_3\Omega_4  + \sum_x \frac{\partial \Omega_1}{\partial\tau} \sum_m\Omega_2\sum_r\Omega_3\Omega_4\Big\}
    \end{align}
    \endgroup
    where $\Lambda_1$ is defined in Equation (\ref{equal-theo-all-trans}),  
    and $\Omega_1$, $\Omega_2$, $\Omega_3$ and $\Omega_4$ are defined in Equations (\ref{equal-omega-1})$\sim$(\ref{equal-omega-4}), respectively. In addition, we have:
    
    \begingroup
    \setlength{\abovedisplayskip}{2pt}
    \setlength{\belowdisplayskip}{2pt}
    \setlength{\abovedisplayshortskip}{2pt}
    \setlength{\belowdisplayshortskip}{2pt}
    \begin{align}\label{equal-domega1}
    &\textstyle \frac{d}{d\tau} \Omega_1(x,\tau) = \textstyle \binom{\frac{1}{2}v(v+1)}{\tau}\binom{v}{x}\binom{\frac{1}{2}v(v-1)}{\tau-x} \cdot F_1 \\
    &\textstyle \frac{d}{d\tau} \Omega_2(m,x,\tau) = \textstyle \binom{\frac{1}{2}v(v-1)}{\tau-x}^{-1} \binom{v}{m} \cdot F_2\cdot \sum_t^m F_3\cdot F_4
    \end{align}
    \endgroup
    where 
    \begingroup
    \setlength{\abovedisplayskip}{2pt}
    \setlength{\belowdisplayskip}{2pt}
    \setlength{\abovedisplayshortskip}{2pt}
    \setlength{\belowdisplayshortskip}{2pt}
    \begin{align}
    F_1&=\textstyle H(\tau)-H(\frac{1}{2}v(v+1)-2\tau)- H(\tau-x) \nonumber \\
    &\hspace{10pt}+ \textstyle H(x-\tau+\frac{1}{2}v(v-1)) \\\displaybreak[0]
    F_2 &= \textstyle \psi(\tau-x+1)-\psi(x+1-\tau+\frac{1}{2}v(v-1)) \\\displaybreak[0]
    F_3 &= \textstyle (-1)^{m-t}\binom{m}{t}\binom{\frac{1}{2}t(t-1)}{\tau-x} \\\displaybreak[0]
    F_4 &= \textstyle 1 + \psi(x+1-\tau+\frac{1}{2}t(t-1)) - \psi(\tau-x+1) \label{equal-prior-ged-solve-final}
    \end{align}
    \endgroup
    
    Here, $v$ is short for $|V_1'|$, $\tau$ is short for of GED, and $x,m,t$ are summation subscripts in Equation (\ref{equal-theo-all-trans}). $H(n)$ is the $n$-th Harmonic Number, and $\psi(\cdot)$ is the Digamma Function \cite{wiki:digamma}.

       \section{Proof of Equation (8)} \label{proof-theo-all-trans}
   Define:
   \begin{align}\label{equal-theo-all-trans-begin}
   \Omega=Pr[GBD=\varphi \mid GED=\tau]
   \end{align}
   By marginalizing out $S$ from $\Omega$:
   \begin{align*}
   \Omega =\sum_{s}\{Pr[GBD=\varphi, S=s \mid GED=\tau]\}
   \end{align*}
   By applying Chain Rule on $\Omega$:
   \begin{align*}
   \Omega = \sum_{s}&\{Pr[GBD=\varphi \mid S=s, GED=\tau] \nonumber \cdot Pr[S=s \mid GED=\tau] \}
   \end{align*}
   In our model, every GED operation sequence is randomly selected with same probability, therefore,
   \begin{align*}
   Pr[S=s \mid GED=\tau] = 1 / |SEQ|
   \end{align*}
   Let $N=\lvert SEQ\rvert$ and we have:
   \begin{align*}
   \Omega &= \frac{1}{N}\sum_{s}\left\{Pr\left[GBD=\varphi \mid S=s, GED=\tau\right] \right\}
   \end{align*}
   By marginalizing out $X$ and $Y$ from $\Omega$:
   \begin{align}
   \Omega = \frac{1}{N}\sum_{s,x,y}&\{Pr[GBD=\varphi, X=x, Y=y \mid \nonumber S=s, GED=\tau] \}
   \end{align}
   Since $Y = \tau-x$ when given $GED=\tau$ and $X=x$:
   \begin{align}
   \Omega = \frac{1}{N}\sum_{s,x}\{Pr[GBD=\varphi, X=x, Y=\tau-x \mid \nonumber S=s, GED=\tau] \}
   \end{align}
   By applying Chain Rule on $\Omega$:
   \begin{align*}
   \Omega = \frac{1}{N}\sum_{s,x}&\{Pr[GBD=\varphi \mid X=x, Y=\tau-x, \nonumber S=s, GED=\tau] \nonumber \\
   &\cdot Pr[X=x, Y=\tau-x \mid S=s, GED=\tau] \}
   \end{align*}
   According to the Bayesian Network in Figure \ref{fig-bayes-net}, we have:
   \begin{align*}
   \Omega = \frac{1}{N}\sum_{s}\sum_{x}&\left\{Pr\left[GBD=\varphi \mid X=x, Y=\tau-x\right]\right. \nonumber \\
   &\left.\cdot Pr\left[X=x, Y=\tau-x \mid S=s\right] \right\}
   \end{align*}
   Define: 
   \begin{align}
   \Theta_1(x,\tau)&=Pr\left[GBD=\varphi \mid X=x, Y=\tau-x\right] \\
   \Omega_1(x,\tau)&=\frac{1}{N}\sum_{s}Pr\left[X=x, Y=\tau-x \mid S=s\right]
   \end{align}
   Then:
   \begin{align}
   \Omega = \sum_{x}\left\{\Omega_1(x,\tau) \cdot \Theta_1(x,\tau) \right\} 
   \end{align}
   Likewise, by marginalizing out $Z$ from $\Theta_1(x,\tau)$ and applying Chain Rule:
   \begin{align*}
   \Theta_1(x,\tau)&=\sum_{m}\{Pr[GBD=\varphi \mid Z=m, X=x,Y=\tau-x]\nonumber \\
   &   
   \cdot Pr[Z=m \mid X=x, Y=\tau-x]  \}  
   \end{align*}
   From Bayesian Network in Figure \ref{fig-bayes-net}, we have:
   \begin{align*}
   \Theta_1(x,\tau)=&\sum_{m}\left\{Pr\left[GBD=\varphi \mid X=x, Z=m\right] \right. \nonumber \\
   &\left. \cdot Pr\left[Z=m \mid Y=\tau-x\right] \right\}
   \end{align*}
   Define: 
   \begin{align}
   &\Omega_2(m,x,\tau)=Pr\left[Z=m \mid Y=\tau-x\right] \\
   &\Theta_2(m,x,\varphi)=Pr\left[GBD=\varphi \mid X= x, Z=m\right]
   \end{align}
   We have:
   \begin{align}\label{equal-omega-proof-omega2}
   \Theta_1(x,\tau,\varphi)=\sum_{m}\left\{\Omega_2(m,x,\tau) \cdot \Theta_2(m,x,\varphi) \right\} 
   \end{align}
   Likewise, by marginalizing out $R$ from $\Theta_2(m,x,\varphi)$ and then applying Chain Rule:
   \begin{align*}
   \Theta_2(m,x,\varphi)=\sum_{r}&\left\{Pr\left[GBD=\varphi \mid R=r, X=x, Z=m\right] \right. \nonumber \\
   &\left. \cdot Pr\left[R=r \mid X=x, Z=m\right] \right\}
   \end{align*}
   From Bayesian Network in Figure \ref{fig-bayes-net}, we have:
   \begin{align*}
   \Theta_2(m,x,\varphi)=\sum_{r}&\{Pr\left[GBD=\varphi \mid R=r\right] \\
   &\cdot Pr\left[R=r \mid X=x, Z=m\right] \}
   \end{align*}
   Define:
   \begin{align}
   \label{equal-omega-proof-omega4} \Omega_3(r,\varphi)&=Pr\left[GBD=\varphi \mid R=r\right] \\
   \label{equal-omega-proof-omega5} \Omega_4(x,r,m)&=Pr\left[R=r \mid X=x, Z=m\right]
   \end{align}
   Finally, from Equations (\ref{equal-theo-all-trans-begin}) $\sim$ (\ref{equal-omega-proof-omega5}) we can obtain:
   \begin{align*}
   \Omega&=\sum_{x}\Omega_1(x,\tau) \cdot \Theta_1(x,\tau) \nonumber \\
   &=\sum_{x}\Omega_1(x,\tau) \cdot \sum_{m}{\Omega_2(m,x,\tau) \cdot \Theta_2(m,x,\varphi)} \nonumber \\
   &=\sum_{x}\Omega_1(x,\tau) \cdot \sum_{m}{\Omega_2(m,x,\tau) \cdot \sum_{r} \Omega_3(r,\varphi) \cdot \Omega_4(x,r,m) } 
   \end{align*}
   
   Therefore, Equation (8) is proved, while the formulae for calculating $\Omega_1$, $\Omega_2$, $\Omega_3$ and $\Omega_4$ are given in Lemmas \ref{LEMMA_CAL_OMEGA1}, \ref{LEMMA_CAL_OMEGA2}, \ref{LEMMA_CAL_OMEGA3} and \ref{LEMMA_CAL_OMEGA4}, respectively. Please refer to Appendices \ref{proof-lemma-cal-omega1}, \ref{proof-lemma-cal-omega2}, \ref{proof-lemma-cal-omega3} and \ref{proof-lemma-cal-omega4} for Lemmas \ref{LEMMA_CAL_OMEGA1}, \ref{LEMMA_CAL_OMEGA2}, \ref{LEMMA_CAL_OMEGA3} and \ref{LEMMA_CAL_OMEGA4}, respectively.

   \section{Lemma \ref{LEMMA_CAL_OMEGA1} and its Proof} \label{proof-lemma-cal-omega1}
   
   \begin{lemma} \label{LEMMA_CAL_OMEGA1}
       Given $GED=\tau$, for any integer $x \in [0,\tau]$, we have:
       \begin{align}
       \Omega_1(x,\tau) & \textstyle = \frac{1}{|SEQ|}\sum_{s \in SEQ} Pr[X=x, Y=\tau-x |S=s] \nonumber \\
       & \textstyle = \mathcal{H}\left(x; |V_1'|+{|V_1'| \choose 2}, |V_1'|, \tau\right)
       \end{align}
       where function $\mathcal{H}(x;M,K,N)$ is defined in Equation (\ref{equal-hyge-pmf}).
   \end{lemma}
   
   \begin{proof}
       From the definitions, $\Omega_1(x)$ is the probability of a random graph edit sequence $seq_s$ exactly relabelling $x$ vertices and $\tau-x$ edges. Since the extended graph $G_1'$ is a complete graph, it has $|E_1'|={|V_1'| \choose 2}$ edges. Therefore, the number of ways to choose $x$ vertices for relabelling is $|V_1'| \choose x$, and the number of ways to choose $\tau-x$ edges for relabelling is ${|V_1'| \choose 2} \choose \tau-x$. Then we have:
       \begin{align*}
       \Omega_1(x,\tau)&=\frac{{|V_1'| \choose x} \cdot {{|V_1'| \choose 2} \choose \tau-x }}{{{|V_1'|+{|V_1'| \choose 2}} \choose \tau}} 
       \textstyle =\mathcal{H}\left(x; |V_1'|+{|V_1'| \choose 2}, |V_1'|, \tau\right)
       \end{align*}
       where function $\mathcal{H}(x;M,K,N)$ is defined in Equation (\ref{equal-hyge-pmf}).
   \end{proof}

   \section{Lemma \ref{LEMMA_CAL_OMEGA2} and its Proof} \label{proof-lemma-cal-omega2}
   
   \begin{lemma} \label{LEMMA_CAL_OMEGA2}
       Given $GED=\tau$, for any integer $x \in [0,\tau]$ and $m \in [0, |V_1'|] $, we have:
       \begin{align}
       \Omega_2(m,x,\tau)&= Pr\left[Z=m \mid Y=\tau-x\right] \nonumber \\
       & \textstyle = {{{|V_1'|\choose 2}\choose{\tau-x}}}^{-1}{\sum_{t=0}^m {(-1)^{m-t}}{{|V_1'|}\choose{m}}{m \choose t}{{t\choose 2}\choose{\tau-x}}}
       \end{align}
   \end{lemma}
   
   \begin{proof}
       Let $x'=\tau-x$. Since the extended graph $G_1'$ is a complete graph, it has $|E_1'|={|V_1'| \choose 2}$ edges. From definitions of $Y$ and $Z$, $\Omega_2(m,x)$ can be modelled by following problem: 
       
       \begin{itemize}[leftmargin=*]
           \item Randomly select $x'$ edges from a complete graph with $|V_1'|$ vertices and ${|V_1'| \choose 2}$ edges, what is the probability of these edges exactly covering $m$ vertices.
       \end{itemize}
       
       Let $V_X$ be the vertices covered by an edge subset $X \subseteq E_1'$. By the inclusion-exclusion principle, for any vertex subset $S \subseteq V_1'$ where $|S|=m$, the number of possible edge sets $X$ which satisfies $V_X=S$ is:
       \begin{equation*}
       \textstyle k_{x',m}=\sum_{t=0}^{m}(-1)^{m-t}{m \choose t}{{t \choose 2} \choose x'}
       \end{equation*}
       Since the number of sets $S \subseteq V_1'$ of size $m$ is ${|V_1'| \choose m}$, the total number of ways to pick an edge set $X$ which satisfies $|X|=x'$ and $|V_X|=m$ is:
       \begin{align*}
       \textstyle \mathcal{K}_{x', m}=\sum_{t=0}^m {(-1)^{m-t}}{{|V_1'|}\choose{m}}{m \choose t}{{t\choose 2}\choose{x'}} 
       \end{align*}
       Also the number of ways to pick an edge set $X\subseteq E_1'$ is:
       \begin{equation*} \textstyle \mathcal{K}_{x'} = \sum_m{\mathcal{K}_{x', m}}={{|V_1'|\choose 2}\choose{x'}}\end{equation*}
       Therefore, we have:
       \begin{equation*} \textstyle \Omega_2(m,x,\tau)=\frac{\mathcal{K}_{x', m}}{\mathcal{K}_{x'}}=\frac{\mathcal{K}_{\tau-x, m}}{\mathcal{K}_{\tau-x}}\end{equation*}
       So Lemma \ref{LEMMA_CAL_OMEGA2} is proved.
   \end{proof}

   \section{Lemma \ref{LEMMA_CAL_OMEGA3} and its Proof} \label{proof-lemma-cal-omega3}
   
   \begin{lemma} \label{LEMMA_CAL_OMEGA3}
       Given $GBD=\varphi$, for any integer $r \in [0, |V_1'|] $, we have:
       \begin{equation}
       \textstyle \Omega_3(r,\varphi)= Pr\left[GBD=\varphi \mid R=r\right] = {r \choose r-\varphi}\cdot \frac{(\mathbb{D}-1)^{\varphi}}{\mathbb{D}^r}
       \end{equation}
       where $\mathbb{D}$ is the number of all possible branch types, and 
       \begin{align}
       \textstyle \mathbb{D}=|\mathcal{L}_{V}| \cdot {|V_1'| + |\mathcal{L}_E| - 1 \choose |\mathcal{L}_E|}
       \end{align}
   \end{lemma}
   
   \begin{proof}
       
       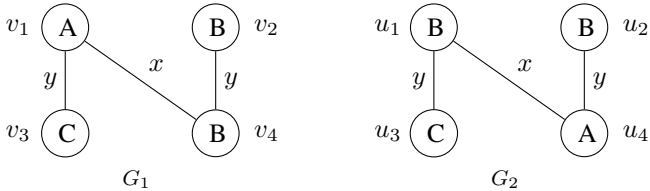
\begin{figure}[!htbp]
           \captionsetup[subfigure]{labelformat=empty}
           \subfloat[$G_1$]{
               \begin{tikzpicture}[
               vertex/.style={draw,circle,text width=5pt,align=center},
               tag/.style={text width=5pt,align=center}
               ]
               \node[vertex] (v1) at (0,1.41) {A};
               \node[tag] (t2) at (-0.7,1.41) {$v_1$};
               \node[vertex] (v2) at (2,1.41) {B};
               \node[tag] (t3) at (2.6,1.41) {$v_2$};
               \node[vertex] (v3) at (0,0) {C};
               \node[tag] (t2) at (-0.7,0) {$v_3$};
               \node[vertex] (v4) at (2,0) {B};
               \node[tag] (t3) at (2.6,0) {$v_4$};
               \path (v1) edge[-] (v3);
               \path (v1) edge[-] (v4);
               \path (v2) edge[-] (v4);
               \node[tag] (te1) at (-0.2,0.7) {$y$};
               \node[tag] (te3) at (1.2,0.9) {$x$};
               \node[tag] (te3) at (2.2,0.7) {$y$};
               \end{tikzpicture}
           }
           \hfill
           \subfloat[$G_2$]{
               \begin{tikzpicture}[
               vertex/.style={draw,circle,text width=5pt,align=center},
               tag/.style={text width=5pt,align=center}
               ]
               \node[vertex] (v1) at (0,1.41) {B};
               \node[tag] (t2) at (-0.7,1.41) {$u_1$};
               \node[vertex] (v2) at (2,1.41) {B};
               \node[tag] (t3) at (2.6,1.41) {$u_2$};
               \node[vertex] (v3) at (0,0) {C};
               \node[tag] (t2) at (-0.7,0) {$u_3$};
               \node[vertex] (v4) at (2,0) {A};
               \node[tag] (t3) at (2.6,0) {$u_4$};
               \path (v1) edge[-] (v3);
               \path (v1) edge[-] (v4);
               \path (v2) edge[-] (v4);
               \node[tag] (te1) at (-0.2,0.7) {$y$};
               \node[tag] (te3) at (1.2,0.9) {$x$};
               \node[tag] (te3) at (2.2,0.7) {$y$};
               \end{tikzpicture}
           }
           \caption{Graphs for the Proof of Lemma \ref{LEMMA_CAL_OMEGA3}}
           \label{fig-lemma3-example}
       \end{figure}
       
       From the definition, $R=r$ means there are exactly $r$ branches
       whose vertices or edges are relabelled when transforming $G_1$ into $G_2$
       during the graph edit process.
       For simplicity, we define these branches as \emph{relabelled branches}.  
       However, the value of $R$ could be larger than the difference between branches in $G_1$ and $G_2$,
       i.e., $GBD(G_1,G_2)$, because it is possible that a subset of relabelled branches 
       are just a re-ordering of the original ones, which are denoted by $\tilde{B}_{G_1}$.
       
       For instance, in the graph edit process of transforming 
       $G_1$ into $G_2$, as shown in Figure \ref{fig-lemma3-example}, 
       we need to relabel $v_1$ by label $B$ and $v_4$ by label $A$,
       which means that the number of relabelled branches $R=2$. However,
       $GBD(G_1,G_2)=0$ since this graph edit process just \emph{swaps}
       the branches rooted at $v_1$ and the one rooted at $v_4$,
       so $\tilde{B}_{G_1} = \{B(v_1), B(v_4)\}$ in this example.
       To model the probability of this situation, we define $t=r -\varphi$, so $t$ is essentially 
       the size of $\tilde{B}_{G_1}$. Here we assume that the occurrence probability
       of each branch type is equal, and denote the number of branch types by $\mathbb{D}$.
       Therefore, $\Omega_3(r)$ can be abstracted as the \emph{ball-pair colouring} problem as follows.
       \begin{itemize}[leftmargin=*]
           \item Given two lists of balls $A_1$, $A_2$ of size $r$, where
           each ball has been randomly coloured by one of the $\mathbb{D}$ colours.
           We define the \emph{ball pair} as two balls where one ball from $A_1$
           and another from $A_2$, where every ball can only be in one pair.
           What is the probability that there are exactly $t$ \emph{ball pairs} 
           where the colour inside each pair is the same.
       \end{itemize}
       
       To solve this problem, we take the following steps to count
       the possible ways to colour balls with the conditions above.
       \begin{enumerate}[leftmargin=*]
           \setlength\itemsep{-2pt}
           \item The number of ways to form \emph{ball pairs} is $r!$;
           \item When pairs are fixed, the number of ways to select $t$ pairs and
           assign $t$ colours to them is ${r \choose t} \mathbb{D}^t$;
           \item For the remaining $r-t$ pairs, the number of ways to 
           assign them different colours inside each pair is $[\mathbb{D}(\mathbb{D}-1)]^{r-t}$.
       \end{enumerate}
       
       Since there are totally $r!\cdot \mathbb{D}^{2r}$ ways to form ball pairs and assign colours to all balls, we have:
       \begin{align*}
       \Omega_3(r,\varphi) = \frac{r! {r \choose t} \mathbb{D}^t [\mathbb{D}(\mathbb{D}-1)]^{r-t}}{r! \cdot \mathbb{D}^{2r}} =  {r \choose r-\varphi}\cdot \frac{(\mathbb{D}-1)^{\varphi}}{\mathbb{D}^r}
       \end{align*}
       where the number of possible branch types $\mathbb{D}$ is the number of ways
       to assign $|\mathcal{L}_V|+1$ labels (including the \emph{virtual} label) to the vertex in a branch,
       multiplying the number of ways to assign $|\mathcal{L}_E|+1$ labels to $|V|-1$
       edges in the same branch. This is a variant of the classic object colouring 
       problem and we omit the detailed proof here. 
   \end{proof}

   \section{Lemma \ref{LEMMA_CAL_OMEGA4} and its Proof} \label{proof-lemma-cal-omega4}

   \begin{lemma} \label{LEMMA_CAL_OMEGA4}
       Given $GBD=\varphi$, for any integer $r \in [0, |V_1'|] $, we have:
       \begin{align}
       \Omega_4(x,r,m)&= Pr[R=r \mid X= x, Z=m] \nonumber \\
       & = \mathcal{H}\left(x+m-r; |V_1'|,m,x\right)
       \end{align}
       where $\mathcal{H}(x;M,K,N)$ is defined in Equation (\ref{equal-hyge-pmf}).
   \end{lemma}
   
   \begin{proof}
       Let $t=x+m-r$, so $t$ is the number of vertices both relabelled and covered by relabelled edges. Since the order of relabelling operations does not affect the graph edit result, $\Omega_4(x,r,m)$ can be modelled by the following problem: 
       \begin{itemize}[leftmargin=*]
           \item First randomly select $m$ vertices from $V_1'$ and tag these vertices as special ones. Then randomly select $x$ vertices from $V_1'$. What is the probability of exactly selecting $t$ special vertices in the second selection.
       \end{itemize}
       Since the number of ways to exactly select $t$ special vertices in the second pick is ${|V_1'| \choose m}{m\choose t}{{|V_1'|-m}\choose{x-t}}$, and the number of all ways to pick $m$ and $x$ vertices separately from $V_1'$ is ${|V_1'| \choose m}{{|V_1'|}\choose{x}}$, we have:
       \begin{align*}
       \Omega_4(x,r,m)&=\frac{{|V_1'| \choose m}{m\choose t}{{|V_1'|-m}\choose{x-t}}}{{|V_1'| \choose m}{{|V_1'|}\choose{x}}} = \frac{{m\choose t}{{|V_1'|-m}\choose{x-t}}}{{{|V_1'|}\choose{x}}} \nonumber \\
       &= \mathcal{H}\left(x+m-r; |V_1'|,m,x\right)
       \end{align*}
       where function $\mathcal{H}(x;M,K,N)$ is defined in Equation (\ref{equal-hyge-pmf}).
   \end{proof}
   
   \section{Graph Generating Algorithm} \label{append-graph-gen-algo}
   
   The algorithm for generating the synthetic graphs (i.e., Syn-1 and Syn-2 data sets) is as follows. 
   
   The algorithm aims to generate a set of graphs $G$, such that for $\forall g_i,g_j \in G$, the edit distance between $g_i$ and $g_j$ is known. In order to achieve this goal, we first defined a valid modification center, then we designed the generation algorithm which consists two phases: (1) Generate a random ``qualified'' graph as a template; (2) Modify the template to generate the graph set $G$.
   
   A modification center is a vertex $v_c$ in graph $g$ such that $\forall v_i,v_j \in \{\text{neighbors of }v_c\},i \neq j$, the edit distance between $g-e(i, c)$ and $g-e(j, c)$ is greater than 0, where $e(i,c)$ is the edge between vertices $v_i$ and $v_c$, and $e(j,c)$ is the edge between vertices $v_j$ and $v_c$. For any modification center $v_c$, if we randomly modify its adjacent edges, the edit distance between the original and modified graphs can be simply calculated by comparing the edges adjacent to their modification centers in polynomial time.
   
   However, identifying whether a vertex is a modification center is difficult. Therefore, we propose a relatively efficient signature that can help us to filter out some special cases that a vertex is certainly a modification center. For vertex $v_c$'s neighbor $v_i$, the signature is defined as $\{s_0, s_1, s_2,\ldots ,s_n\}$, where $s_0$ is a set contains $v_i$'s label, and other sets $s_k = \{(v_j.label \allowbreak, e(i,j).label); \forall v_j \in \text{k-hop neighbors of }v_i\}, k > 0$. If a vertex is a modification center, then its neighbors' signatures must be pair-wised different. That is, if we find a vertex whose neighbors' signatures are pair-wised different, this vertex must be a modification center. If there is no such a vertex, we re-generate the graph until success. 
   
   In our settings, the graph should be connected, and have at least one modification center of degree at least $d$, to produce a set of graphs among which the edit distance varies from 0 to $d$. To ensure the connectivity of the graph, we force each vertex $v_i$ is connected to another vertex $v_j$, where $i > j$. After all vertices being connected, we then add remaining edges according to the type of the graph. For \textit{random} graphs, we randomly add edges between in-adjacent vertices. For \textit{scale-free} graphs, we add constant number of edges to each vertex $v_i$, where the neighbors of $v_i$ is picked from $\{v_j; \forall j < i\}$ with the probability proportional to their degrees. 
   
   \section{Supplemental Figures For Experiments on Synthetic Data Sets} \label{append-exp-online-syn}
   
   Please refer to Figures \ref{fig-online-acc-syn1-tau15}$\sim$\ref{fig-online-f1-syn1-tau30} on the last page.
   
   \section{Theorem for Average Degree of Scale-free Graphs} \label{append-theorem-scale-free-degree}
   
   \begin{theorem} \label{theorem-scale-free-degree}
       The average degree of a scale free graph $G$ is $O(\log{n})$, where $n$ is the number of vertices in graph $G$.
   \end{theorem}
   
   \begin{proof}
       The fraction of vertices with degree $k$ in scale-free graphs is $C\cdot k^{-\delta}$ \cite{wiki:scale-free-graphs}, where 
       $2<\delta <3$, and $C$ is a constant. Therefore, the average degree in scale-free graphs is:
       \begingroup
       \setlength{\abovedisplayskip}{2pt}
       \setlength{\belowdisplayskip}{2pt}
       \setlength{\abovedisplayshortskip}{2pt}
       \setlength{\belowdisplayshortskip}{2pt}
       \begin{align}
       d = \sum_{k=1}^{n-1} k \cdot C k^{-\delta} < \sum_{k=1}^{n-1}\frac{C}{k} = C \cdot H(n-1) = O(\log{n})
       \end{align}
       \endgroup
       where $H(n)$ is the $n$-th Harmonic Number. 
   \end{proof}
   
   \begin{figure*}[p]
       \minipage{0.27\textwidth}
       \includegraphics[width=\linewidth]{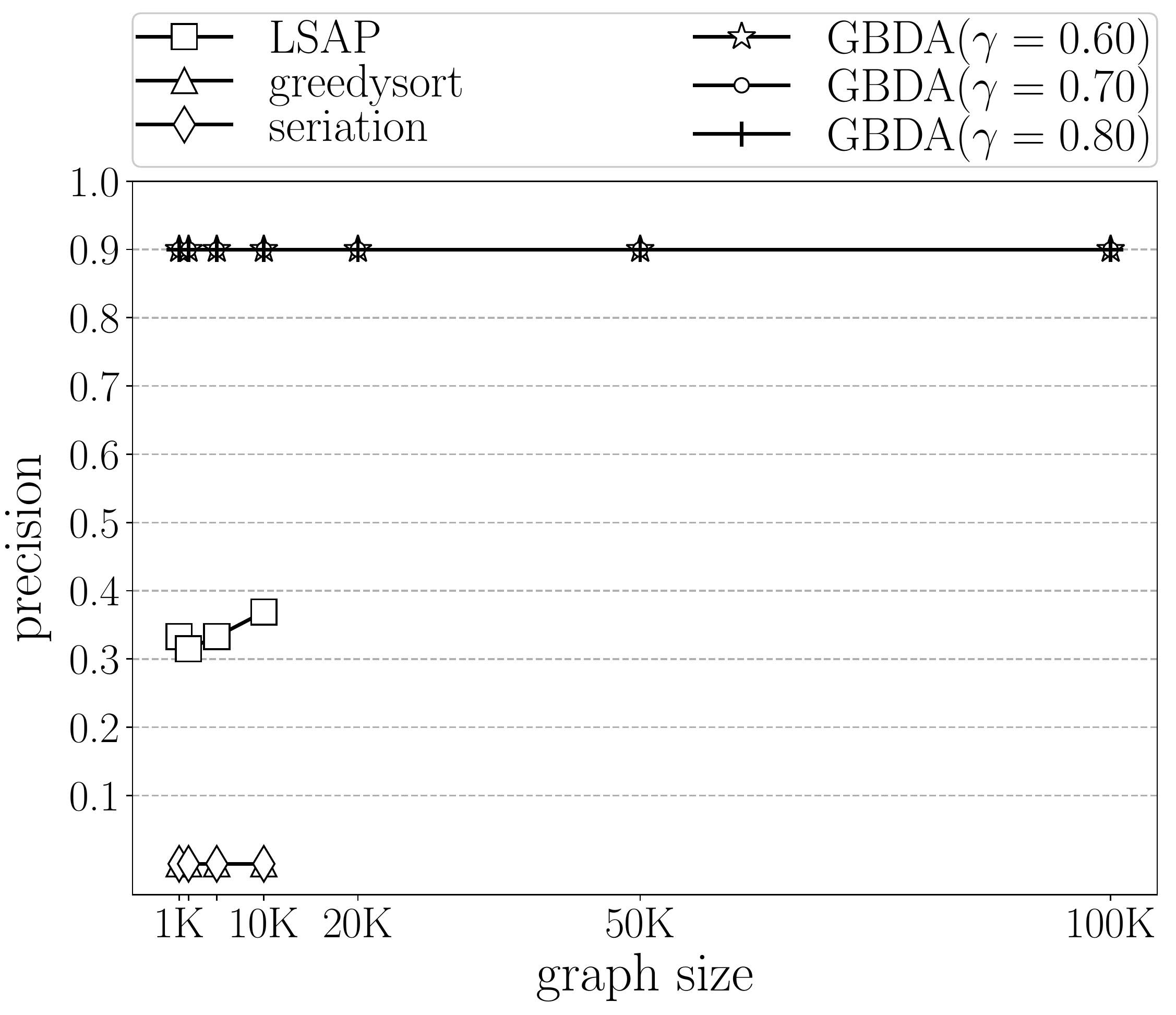}
       \vspace{-20pt}
       \caption{\small{Precision vs. graph size on \mbox{Syn-1} Data Set with various $\gamma$  ($\hat{\tau} = 15$)}}\label{fig-online-acc-syn1-tau15}
       \endminipage\hfill
       \minipage{0.27\textwidth}
       \includegraphics[width=\linewidth]{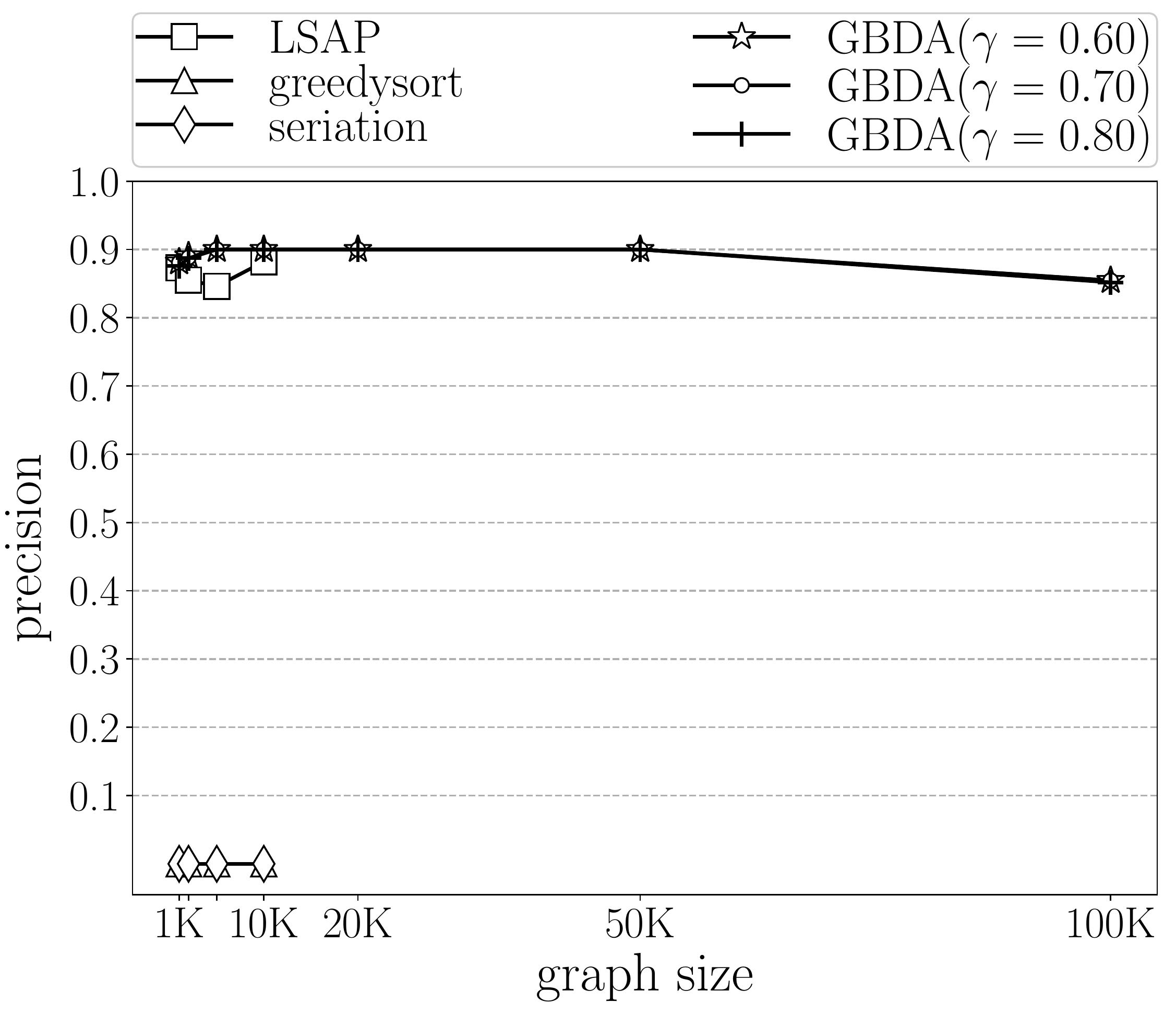}
       \vspace{-20pt}
       \caption{\small{Precision vs. graph size  on \mbox{Syn-1} Data Set with various $\gamma$  ($\hat{\tau} = 20$)}}\label{fig-online-acc-syn1-tau20}
       \endminipage\hfill
       \minipage{0.27\textwidth}
       \includegraphics[width=\linewidth]{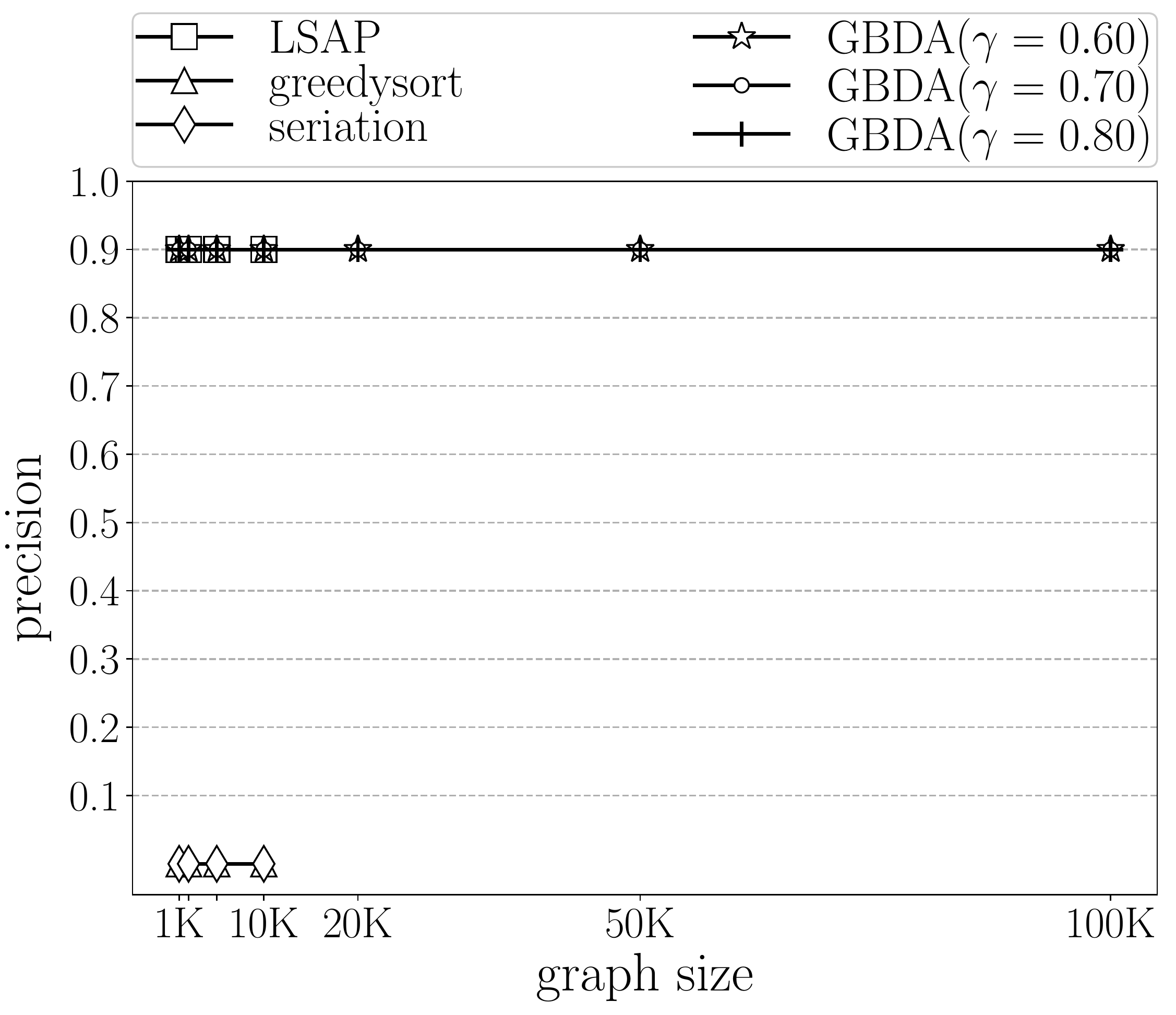}
       \vspace{-20pt}
       \caption{\small{Precision vs. graph size  on \mbox{Syn-1} Data Set with various $\gamma$  ($\hat{\tau} = 25$)}}\label{fig-online-acc-syn1-tau25}
       \endminipage
   \end{figure*}
   
   \begin{figure*}[p]
       \minipage{0.27\textwidth}
       \includegraphics[width=\linewidth]{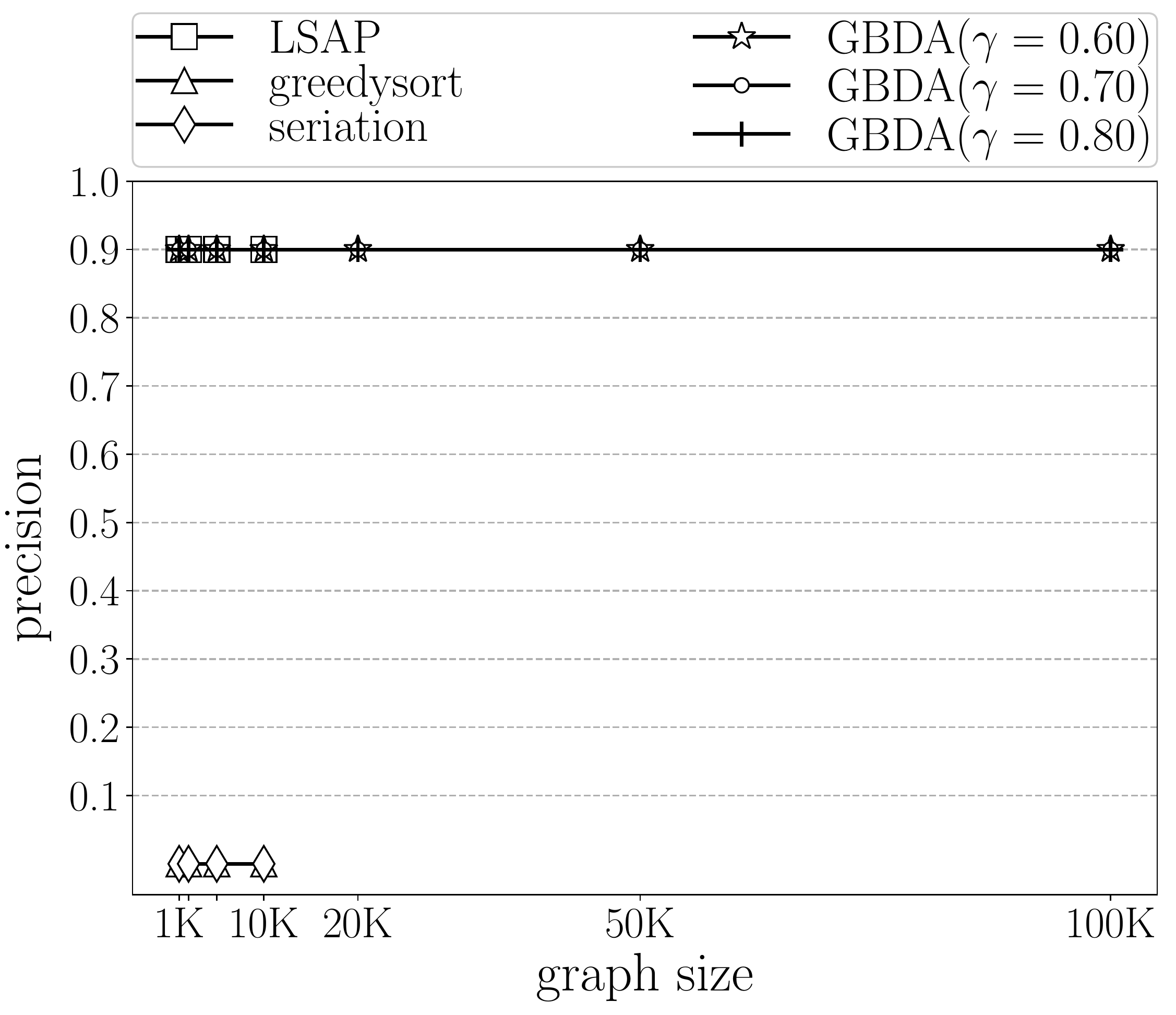}
       \vspace{-20pt}
       \caption{\small{Precision vs. graph size on \mbox{Syn-1} Data Set with various $\gamma$  ($\hat{\tau} = 30$)}}\label{fig-online-acc-syn1-tau30}
       \endminipage\hfill
       \minipage{0.27\textwidth}
       \includegraphics[width=\linewidth]{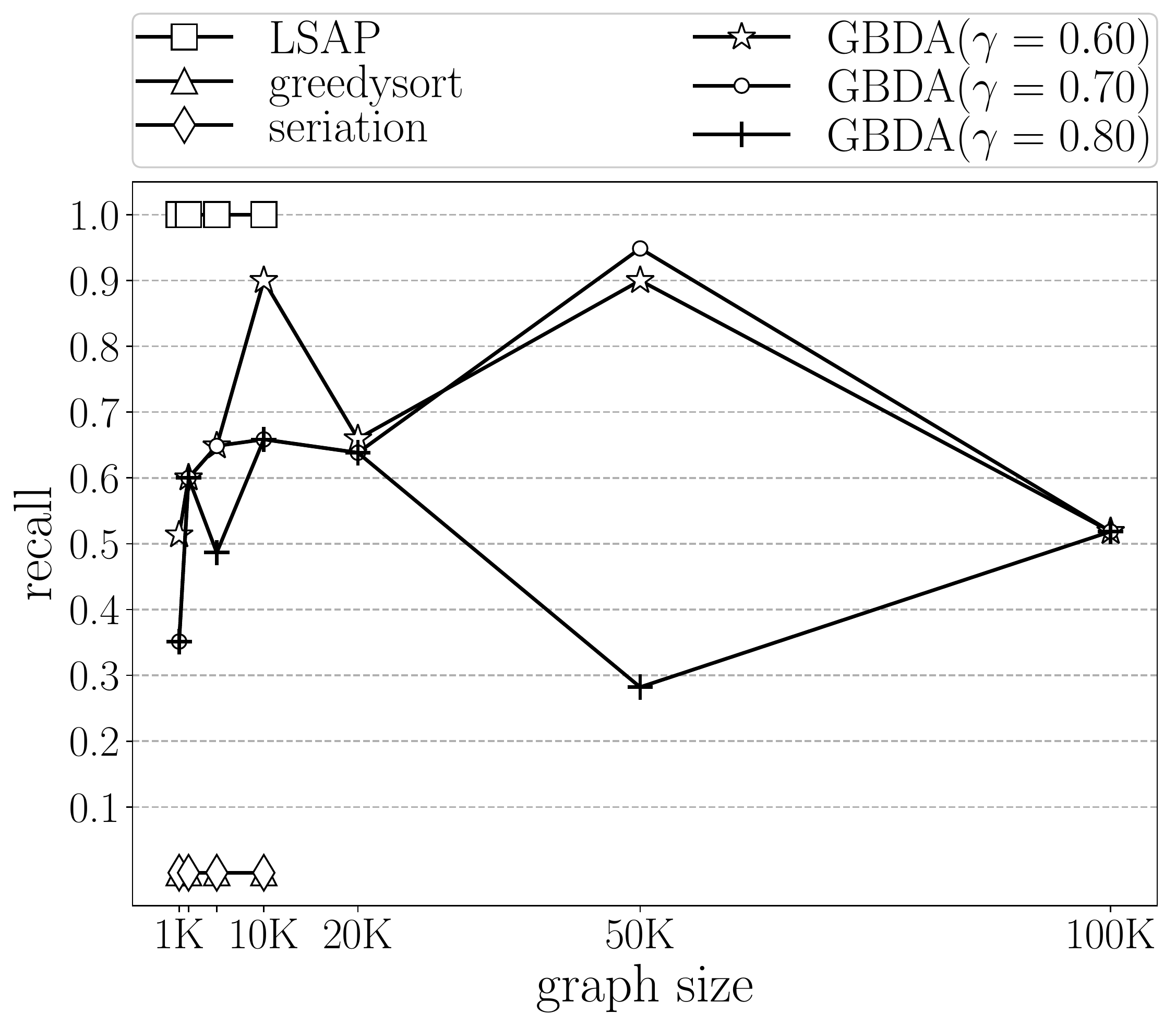}
       \vspace{-20pt}
       \caption{\small{Recall vs. graph size  on \mbox{Syn-1} Data Set with various $\gamma$  ($\hat{\tau} = 15$)}}\label{fig-online-recall-syn1-tau15}
       \endminipage\hfill
       \minipage{0.27\textwidth}
       \includegraphics[width=\linewidth]{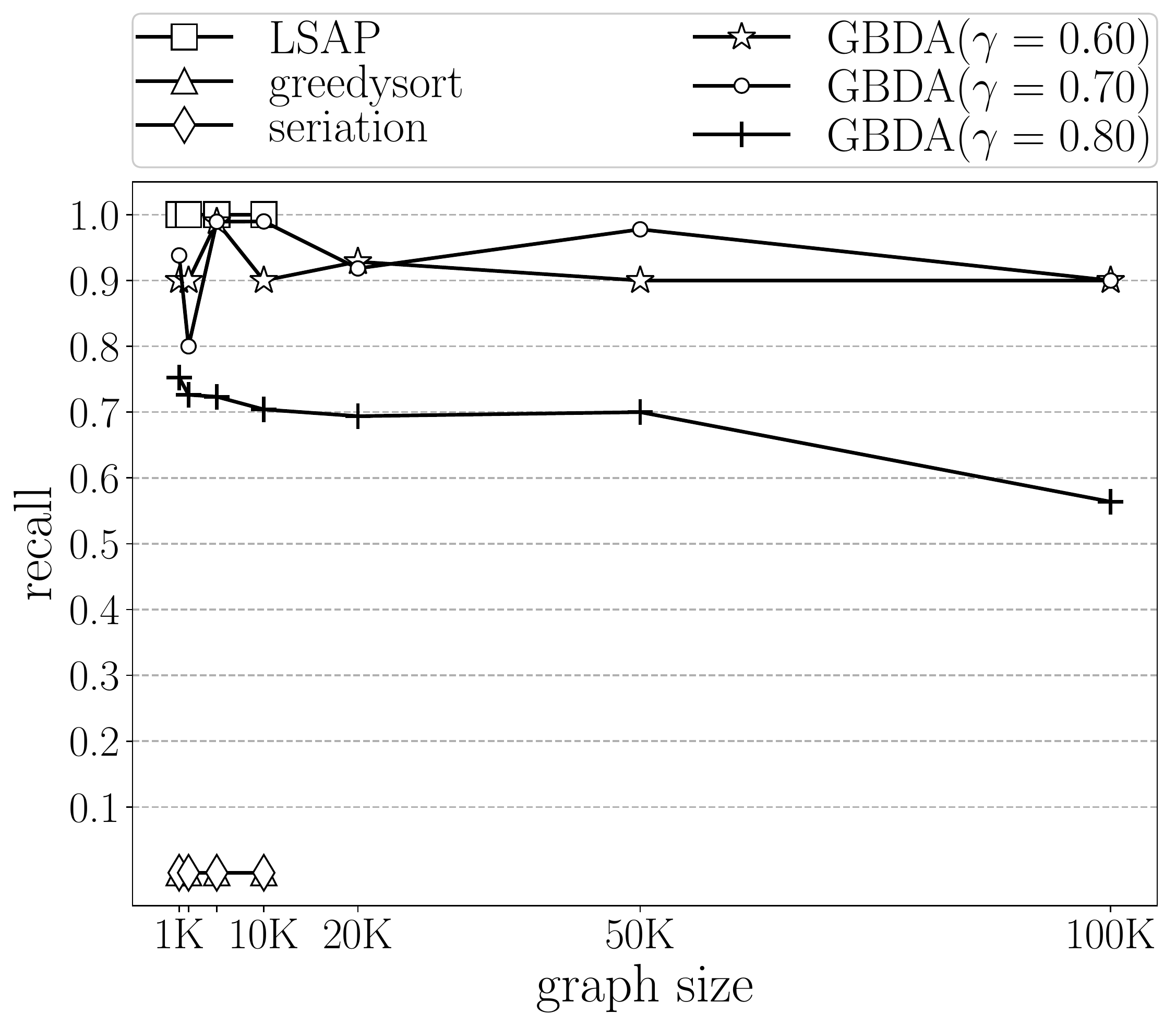}
       \vspace{-20pt}
       \caption{\small{Recall vs. graph size  on \mbox{Syn-1} Data Set with various $\gamma$  ($\hat{\tau} = 20$)}}\label{fig-online-recall-syn1-tau20}
       \endminipage
   \end{figure*}
   
   \begin{figure*}[p]
       \minipage{0.27\textwidth}
       \includegraphics[width=\linewidth]{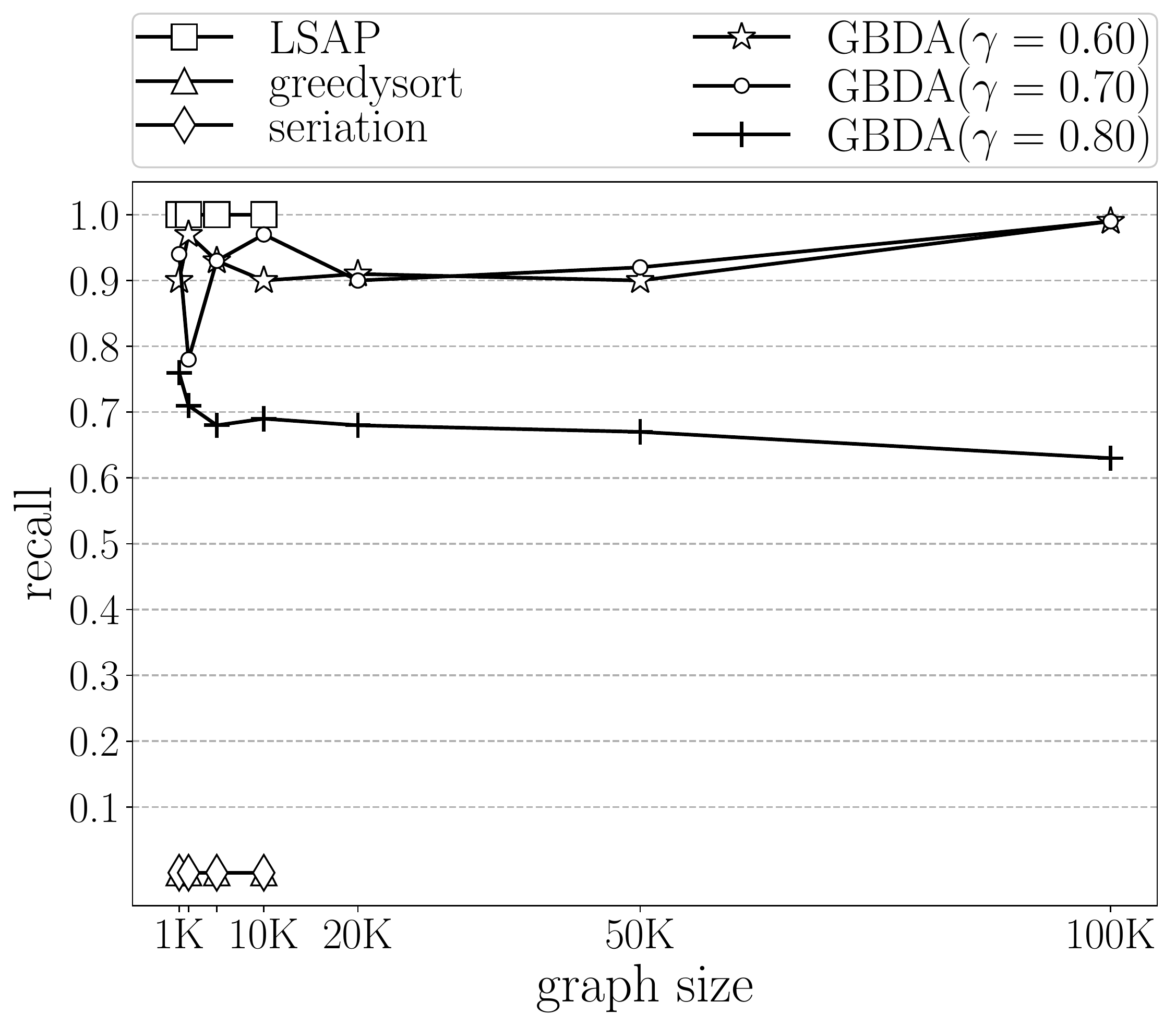}
       \vspace{-20pt}
       \caption{\small{Recall vs. graph size on \mbox{Syn-1} Data Set with various $\gamma$  ($\hat{\tau} = 25$)}}\label{fig-online-recall-syn1-tau25}
       \endminipage\hfill
       \minipage{0.27\textwidth}
       \includegraphics[width=\linewidth]{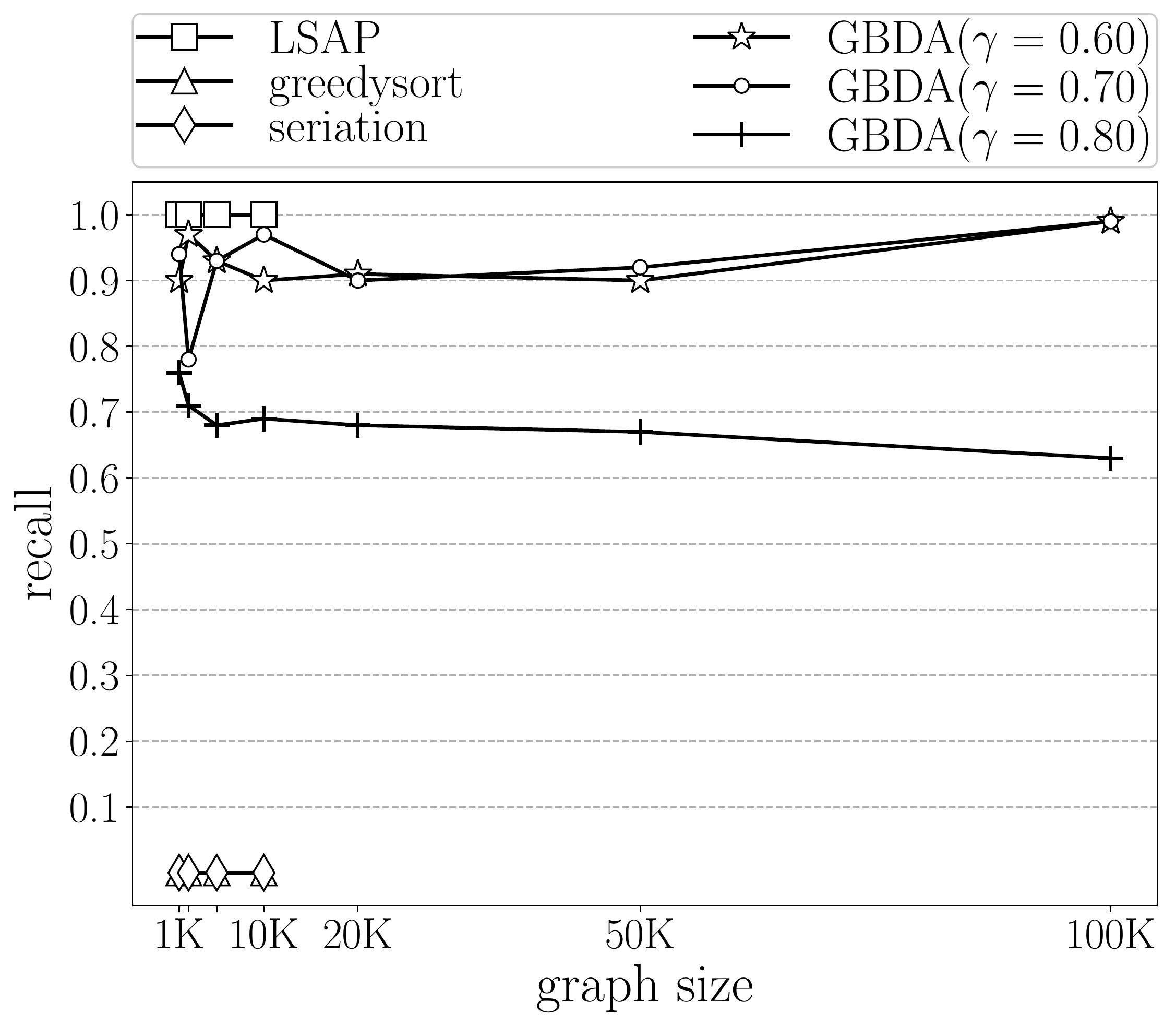}
       \vspace{-20pt}
       \caption{\small{Recall vs. graph size  on \mbox{Syn-1} Data Set with various $\gamma$  ($\hat{\tau} = 30$)}}\label{fig-online-recall-syn1-tau30}
       \endminipage\hfill
       \minipage{0.27\textwidth}
       \includegraphics[width=\linewidth]{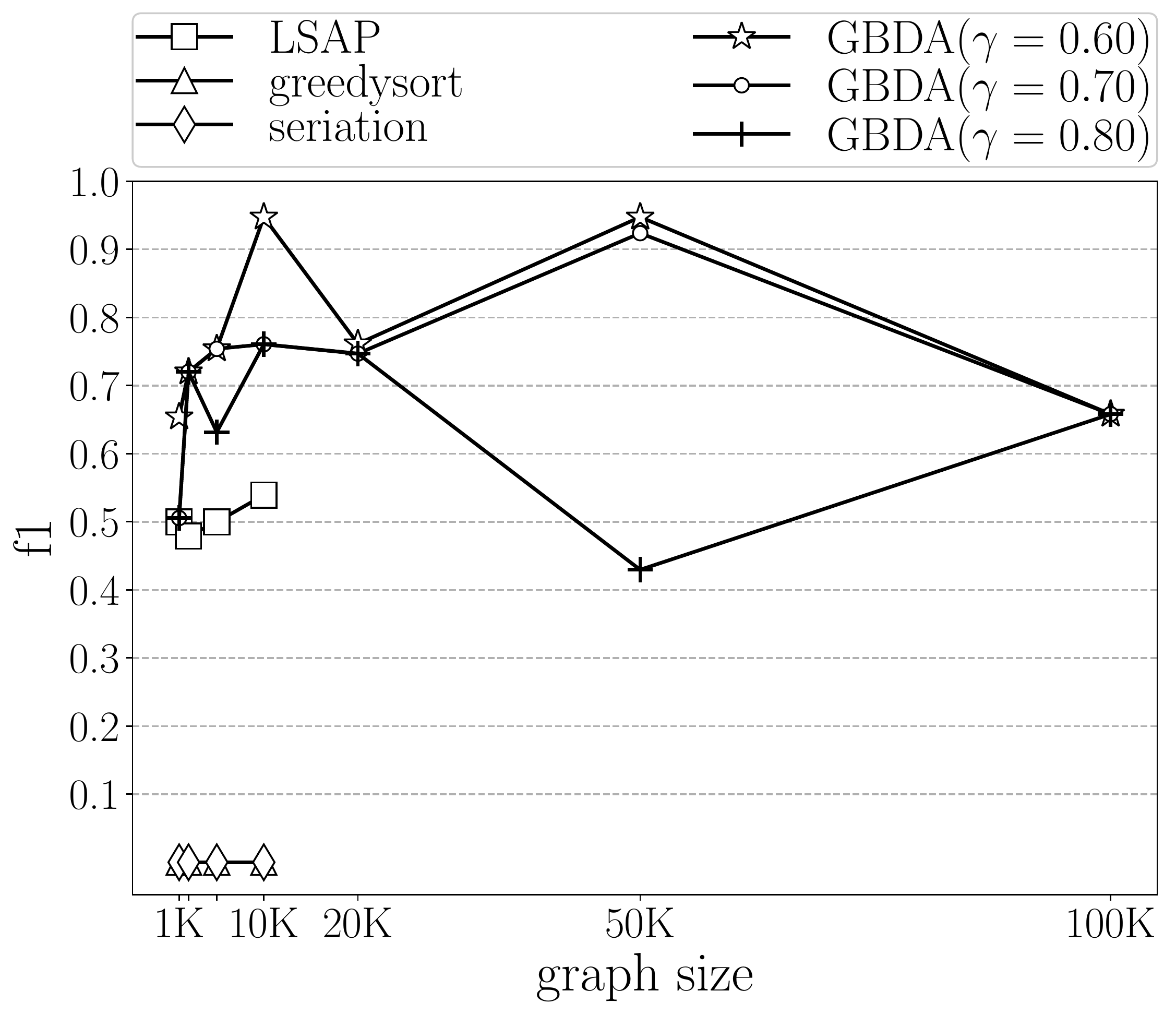}
       \vspace{-20pt}
       \caption{\small{F1-score vs. graph size on \mbox{Syn-1} Data Set with various $\gamma$  ($\hat{\tau} = 15$)}}\label{fig-online-f1-syn1-tau15}
       \endminipage
   \end{figure*}
   
   \begin{figure*}[p]
       \minipage{0.27\textwidth}
       \includegraphics[width=\linewidth]{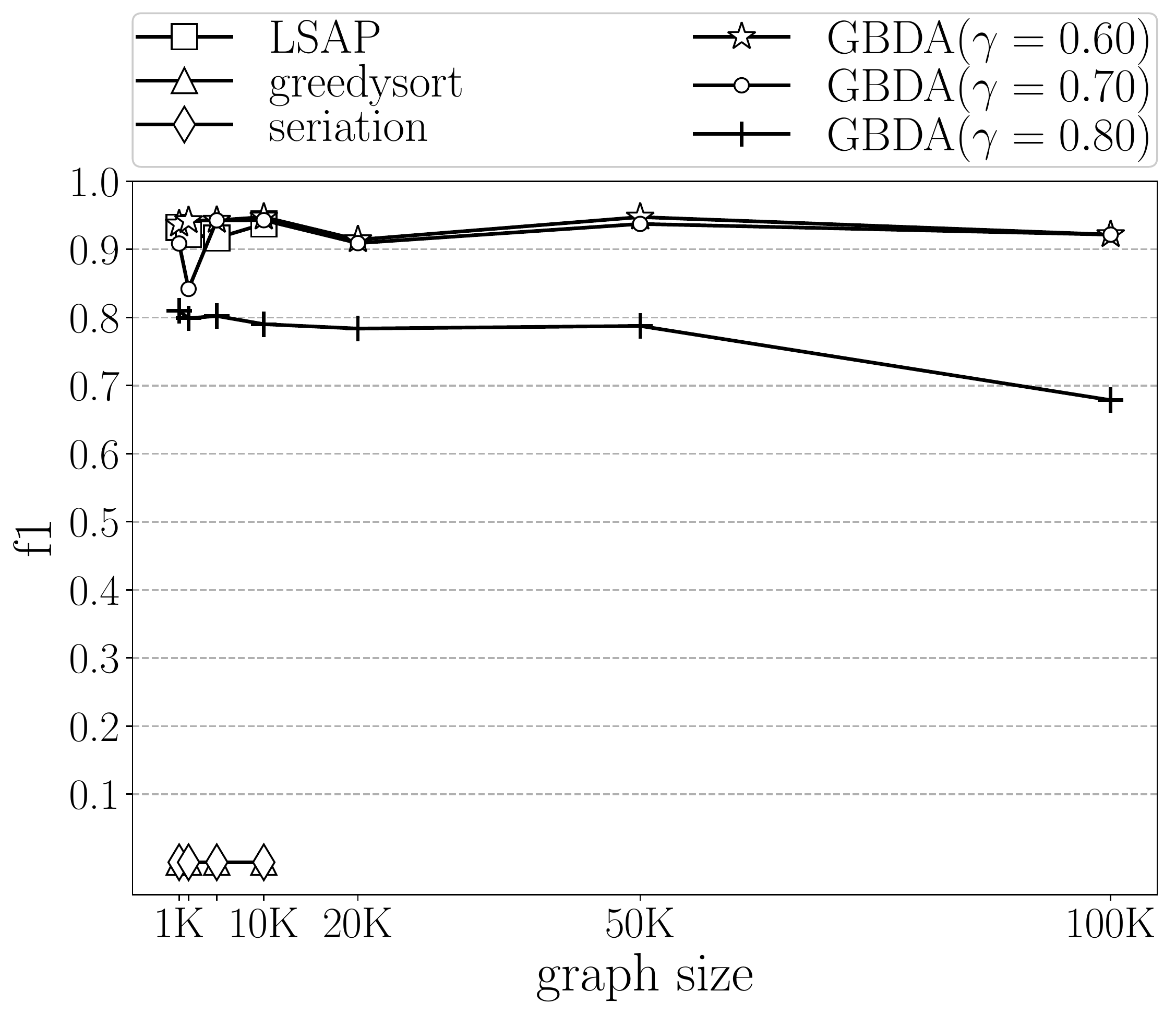}
       \vspace{-20pt}
       \caption{\small{F1-score vs. graph size on \mbox{Syn-1} Data Set with various $\gamma$  ($\hat{\tau} = 20$)}}\label{fig-online-f1-syn1-tau20}
       \endminipage\hfill
       \minipage{0.27\textwidth}
       \includegraphics[width=\linewidth]{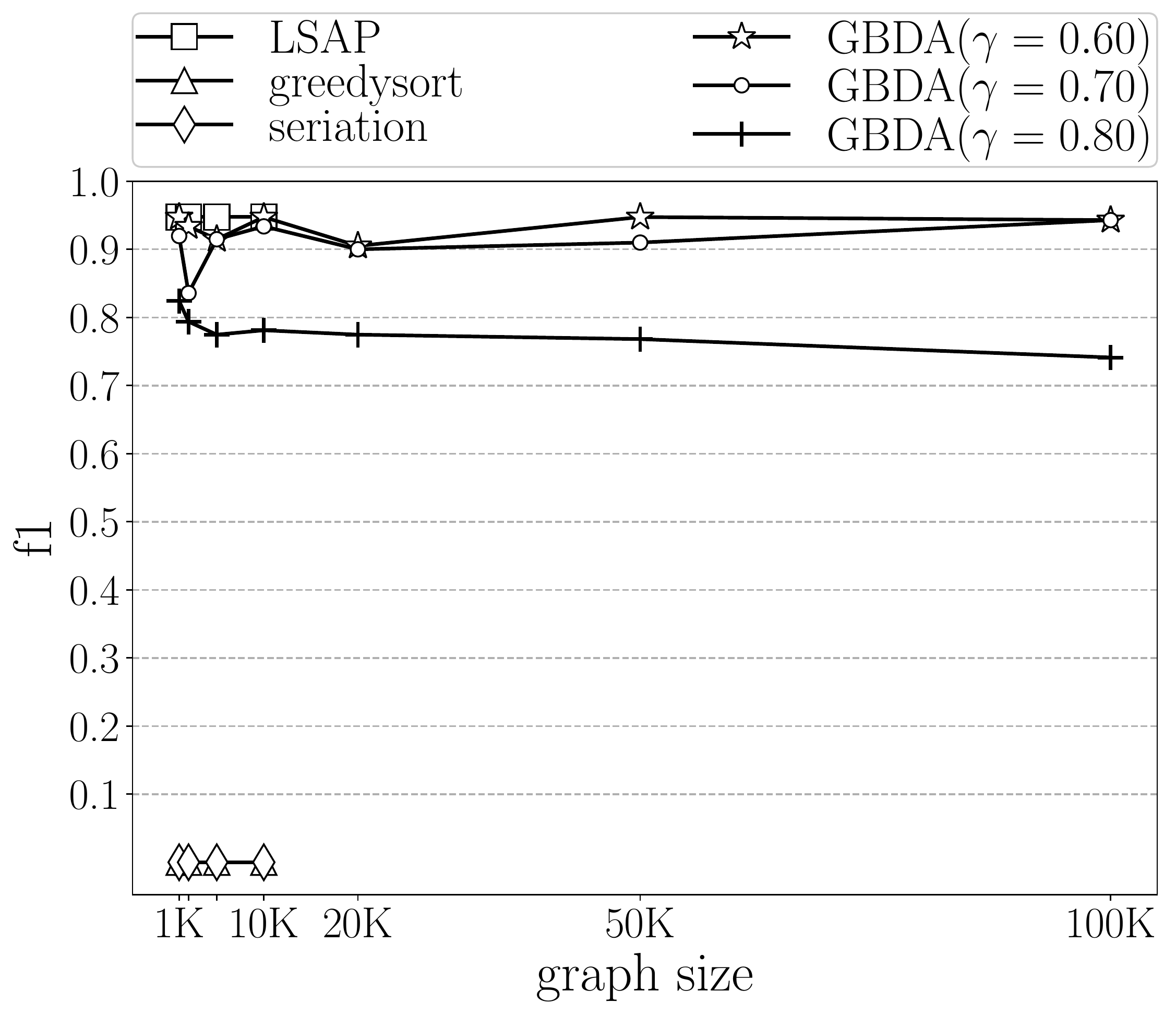}
       \vspace{-20pt}
       \caption{\small{F1-score vs. graph size  on \mbox{Syn-1} Data Set with various $\gamma$  ($\hat{\tau} = 25$)}}\label{fig-online-f1-syn1-tau25}
       \endminipage\hfill
       \minipage{0.27\textwidth}
       \includegraphics[width=\linewidth]{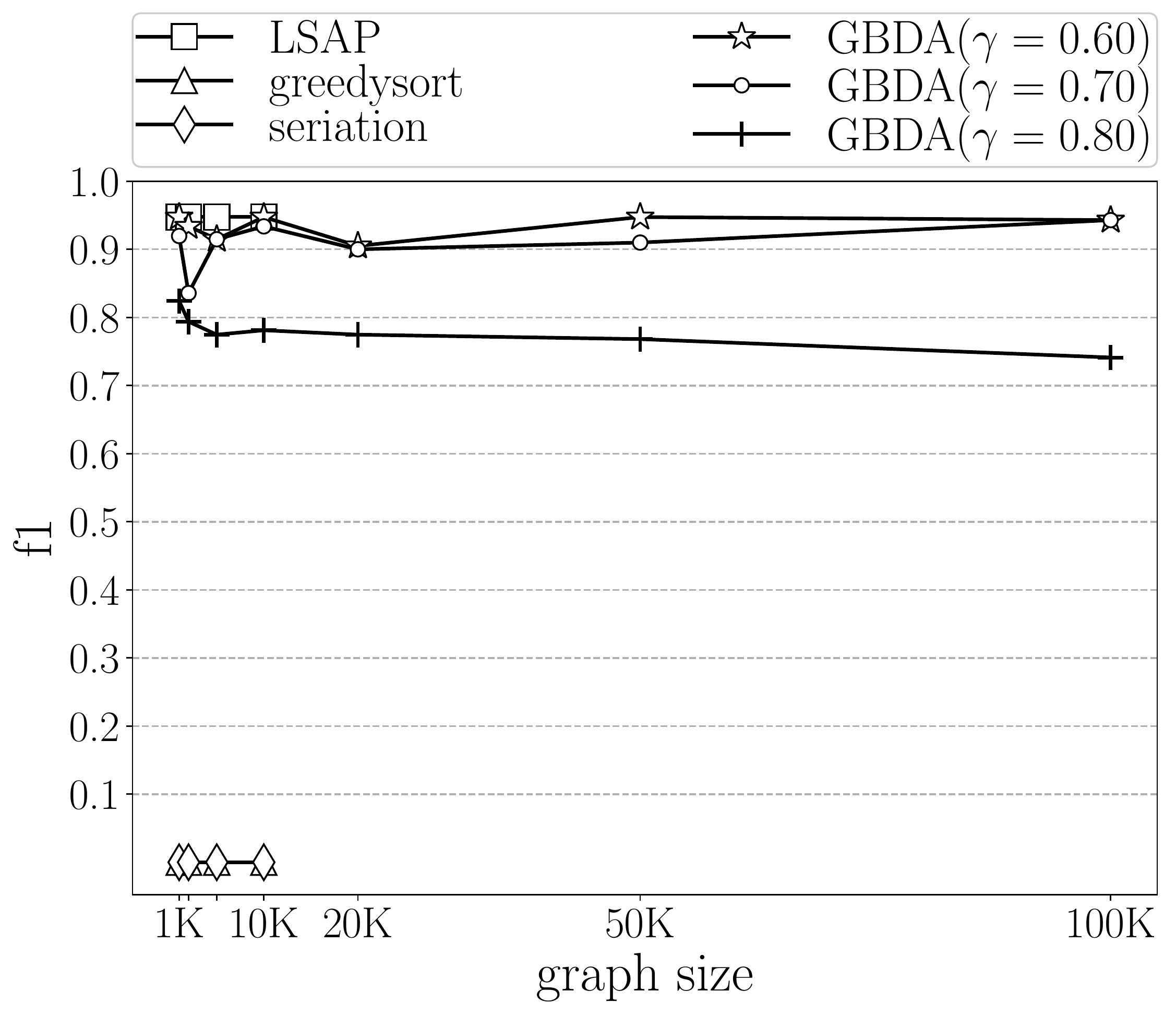}
       \vspace{-20pt}
       \caption{\small{F1-score vs. graph size on \mbox{Syn-1} Data Set with various $\gamma$  ($\hat{\tau} = 30$)}}\label{fig-online-f1-syn1-tau30}
       \endminipage
   \end{figure*}

   \end{appendices}

\end{document}